\documentclass[11pt,a4paper]{article}
\usepackage[margin=1in]{geometry}

\usepackage{amsmath}
\usepackage{amssymb}
\usepackage[noend]{algpseudocode}
\usepackage{algorithm}
\usepackage{amsthm}
\usepackage{amsfonts}
\usepackage{graphicx}
\usepackage{hyperref}
\usepackage{xcolor}
\usepackage{bm}
\usepackage{enumitem}
\usepackage{authblk}

\newtheorem{theorem} {Theorem}

\newtheorem{lemma}{Lemma}
\newtheorem{claim}[lemma]{Claim}

\newtheorem{corollary}{Corollary}
\newtheorem{obs}[lemma]{Observation}
\newtheorem{fact}[lemma]{Fact}
\numberwithin{lemma}{section}

\newcommand{\E}{{\mathbb{E}}}
\newcommand{\R}{{\mathbb{R}}}
\newcommand{\eps}{\varepsilon}
\let\epsilon\varepsilon

\newcommand{\rank}{\operatorname{R}}
\newcommand{\V}{\operatorname{Var}}

\newcommand{\calB}{\mathcal{B}}
\newcommand{\calC}{\mathcal{C}}

\newcommand{\calU}{\mathcal{U}}
\newcommand{\Sstar}{S^{\ast}}

\DeclareMathOperator{\err}{Err}
\newcommand{\schedulestate}{C}
\newcommand{\mycase}[1]{\smallskip\noindent{\bf Case #1:}}

\usepackage{thm-restate}

\begin{document}
	\title{Relative Error Streaming Quantiles\footnote{The research is performed in close collaboration with DataSeketches \url{https://datasketches.apache.org/}, the Apache open source project for streaming data analytics.}}
	
	\author[1]{Graham Cormode\thanks{Supported by European Research Council grant ERC-2014-CoG 647557.}}
	\author[2]{Zohar Karnin}
	\author[3]{Edo Liberty}
	\author[4]{Justin Thaler\thanks{Supported by NSF SPX award CCF-1918989, and NSF CAREER award CCF-1845125.}}
	\author[5]{Pavel Vesel{\'{y}}\thanks{Part of the work done while the author was at the University of Warwick
			and supported by European Research Council grant ERC-2014-CoG 647557.
			Partially supported by the project 19-27871X of GA ČR and by Center for Foundations of Modern Computer Science (Charles University project UNCE/SCI/004).}}
	
	\affil[1]{University of Warwick, \texttt{G.Cormode@warwick.ac.uk}}
	\affil[2]{Amazon, \texttt{zkarnin@gmail.com}}
	\affil[3]{Pinecone, \texttt{edo@edoliberty.com}}
	\affil[4]{Georgetown University, \texttt{justin.thaler@georgetown.edu}}
	\affil[5]{Charles University, \texttt{vesely@iuuk.mff.cuni.cz}}
	
	\date{}
	
	\maketitle
	
\begin{abstract}
	Estimating ranks, quantiles, and distributions 
	over streaming data is a central task 
	in data analysis and monitoring.
	Given a stream of $n$ items from a data universe
	equipped with a total order, the task is to compute a sketch (data structure) of size
	polylogarithmic in $n$.
	Given the sketch and a query item $y$, one should be able to approximate its rank in the stream, i.e., the number of stream elements smaller than or equal to $y$.
	
	Most works to date focused on additive $\eps n$ error approximation, culminating in the KLL sketch that achieved optimal asymptotic behavior. 
	This paper investigates \emph{multiplicative} $(1\pm\eps)$-error approximations to the rank. 
	Practical motivation for multiplicative error stems from demands to understand the tails of distributions,
	and hence for sketches to be more accurate near extreme values. 
	
	The most space-efficient algorithms due to prior work store either $O(\log(\eps^2 n)/\eps^2)$ or
	$O(\log^3(\eps n)/\eps)$ universe items. 
	We present a randomized sketch storing $O(\log^{1.5}(\eps n)/\eps)$ items 
	that can $(1\pm\eps)$-approximate the rank of each universe item with high constant probability;
	this space bound is within an $O(\sqrt{\log(\eps n)})$ factor of optimal. Our algorithm does not require 
	 prior knowledge of the stream length and is fully mergeable, rendering it suitable for parallel and distributed computing environments. 
	
\end{abstract}

\sloppy 
\section{Introduction}
Understanding the distribution of data is a fundamental task in data monitoring and analysis.
In many settings, we want to understand the cumulative distribution function (CDF)
of a large number of observations, for instance, to identify anomalies.
In other words, we would like to track the median, percentiles, and more generally quantiles of a massive input in a small space, without storing all the observations.
Although memory constraints make an exact computation of such order statistics
impossible~\cite{munro1980selection}, most applications can be satisfied with approximating the quantiles, which also yields a compact function with a bounded distance from the true CDF.

The problem of streaming quantile approximation captures this task in the context
of massive or distributed data\-sets.
Let $\sigma = (x_1,\ldots,x_n)$ be a stream of items, all drawn from a data universe $\mathcal{U}$
equipped with a total order.
For any $y \in \mathcal{U}$, let $\rank(y; \sigma) = \big|\big\{i\in \{1,\dots,n\} \ | \ x_i \leq y\big\}\big|$ be the rank of $y$ in the stream.  
When $\sigma$ is clear from context, we write $\rank(y)$.
The objective is to process the stream in one pass while storing a small number of universe items and $O(\log n)$-bit variables (e.g., counters), and then use those to approximate $\rank(y)$ for any $y \in \mathcal{U}$. 
A guarantee for an approximation $\hat{\rank}(y)$ is \emph{additive} if $|\hat{\rank}(y) - \rank(y)| \le \eps n$, 
and \emph{multiplicative} or \emph{relative} if $|\hat{\rank}(y) - \rank(y)| \le \eps \rank(y)$.

If the algorithm is randomized, the desired error guarantee holds for each item $y$ with high probability that can be specified upfront and that affects the space bound.
(On the other hand, the space bounds for all known algorithms hold in the worst case over the inputs and random bits.)
We remark that estimating ranks immediately yields approximate quantiles, and vice versa,
with a similar error guarantee.
More precisely, for $\phi\in [0, 1]$, we define a $\phi$-quantile as the $\lfloor\phi n\rfloor$-th smallest item in $\sigma$.
On quantile query $\phi$, the algorithm should return a $\phi'$-quantile such that $|\phi'-\phi|\le \epsilon$
for the additive error and $|\phi'-\phi|\le \epsilon\cdot \phi$ for the multiplicative error.

We stress that we do not assume any particular data distribution or that the stream is randomly-ordered, that is, we focus on worst-case inputs. However, we assume the input is independent of the random bits used by the algorithm, i.e., we do not aim to achieve adversarial robustness of randomized algorithms; cf.~\cite{Ben-EliezerJWY22}.

A long line of work has focused on achieving additive error guarantees~\cite{pohl1969minimum, agrawal1995one, manku1998approximate, poosala1999approximate, greenwald2001space, arasu2004approximate, ganguly2007nearly, agarwal2013mergeable, felber2015randomized, karnin2016optimal}.
However, additive error is not appropriate for many applications. 
Indeed, often the primary purpose of computing quantiles is to understand the tails of the data distribution.
When $\rank(y) \ll n$, a multiplicative guarantee is much more accurate and thus harder to obtain.
As pointed out by Cormode et al.~\cite{Cormode:2005:ECB:1053724.1054027}, a solution to this problem would also yield high accuracy when $n - \rank(y) \ll n$,
by running the same algorithm with the reversed total ordering (simply negating the comparator).

A quintessential application that demands relative error is monitoring network latencies. In practice, one often tracks response time percentiles $50$, $90$, $99$, $99.9$, etc. 
This is because latencies are heavily long-tailed. For example, Masson et al.~\cite{DBLP:journals/pvldb/MassonRL19} report that for web response times,
the 98.5th percentile can be as small as 2 seconds while the 99.5th percentile can be as large as 20 seconds. 
These unusually long response times affect network dynamics~\cite{Cormode:2005:ECB:1053724.1054027} and are problematic for users.
Furthermore, as argued by Tene in his talk about measuring latency~\cite{tene2015_latency_talk},
one needs to look at extreme percentiles such as $99.995$ to determine the latency 
such that only about $1\%$ of users experience a larger latency during a web session with several page loads. 
Hence, highly accurate rank approximations are required for items $y$ whose rank is very
large ($n - \rank(y) \ll n$); 
this is precisely the requirement captured by the multiplicative error
guarantee.

Achieving multiplicative guarantees is known to be \emph{strictly} harder than additive ones. 
There are randomized comparison-based additive error algorithms that store just $\Theta(\eps^{-1})$ items for constant failure probability~\cite{karnin2016optimal}, which is optimal. In particular, to achieve additive error, the number of items stored may be independent of the stream length $n$.
In contrast, any algorithm achieving multiplicative error must store $\Omega(\eps^{-1}\cdot \log(\eps n))$ items (see \cite[Theorem 2]{Cormode:2005:ECB:1053724.1054027} and Appendix~\ref{app:lower}).%
\footnote{Intuitively, the reason additive-error sketches can achieve space independent of the stream length
is because they can take a subsample of the stream of size about $\Theta(\eps^{-2})$ and then
sketch the subsample. For any fixed item, the additive error to its rank introduced by sampling is at most $\eps n$ 
with high probability. When multiplicative error is required, one cannot subsample the input: for low-ranked items, the multiplicative
error introduced by sampling will, with high probability, not be bounded by any constant.}

The study of the relative-error (rank) guarantee was initiated by Gupta and Zane~\cite{Gupta:2003:CIL:644108.644150},
and the best known algorithms achieving this guarantee are as follows. Zhang et al. \cite{zhang2006space} give a randomized algorithm storing
$O(\eps^{-2}\cdot \log( \eps^2 n))$ universe items. This is essentially a $\eps^{-1}$ factor away from the aforementioned lower bound.
There is also a deterministic algorithm of Cormode et al. \cite{cormode2006space} that
stores $O( \eps^{-1} \cdot \log(\eps n) \cdot \log |\mathcal{U}|)$ items.
However, this algorithm requires prior knowledge of the data universe $\mathcal{U}$
(since it builds a binary tree over $\mathcal{U}$), and is inapplicable when
$\mathcal{U}$ is huge or even unbounded (e.g., if the data can take arbitrary real values).
Finally, Zhang and Wang \cite{zhangwang} designed a deterministic algorithm requiring
$O(\eps^{-1}\cdot \log^3(\eps n))$ space. Recent work of Cormode and Vesel{\'y} \cite{cormode2019tight}
proves an $\Omega(\eps^{-1}\cdot \log^2(\eps n))$ lower bound for deterministic comparison-based
algorithms, which is within a $\log(\eps n)$ factor of Zhang and Wang's upper bound.

Despite both the practical and theoretical importance of multiplicative error (which is arguably an even more natural goal than additive error),
there has been no progress on upper bounds, i.e., no new algorithms, since 2007.

\paragraph{Our streaming result.}
In this work, we give a randomized algorithm that maintains the optimal linear dependence on $1/\eps$ achieved by Zhang and Wang,
with a significantly improved dependence on the stream length.
Namely, we design a one-pass streaming algorithm that given $\eps > 0$,
computes a sketch consisting of
$O\left( \eps^{-1}\cdot \log^{1.5} (\eps n)\right)$
universe items, from which one can derive rank or quantile estimates satisfying the relative error
guarantee with constant probability (see Theorem~\ref{thm:mergeability-full-intro} for a more precise statement).
Ours is the first algorithm to be strictly more space efficient
than \emph{any} deterministic comparison-based algorithm (owing to the $\Omega(\eps^{-1} \log^2(\eps n))$ lower bound in \cite{cormode2019tight}) and is within an $O(\sqrt{\log(\eps n)})$ factor of the known 
lower bound for randomized algorithms achieving multiplicative error.
Furthermore, it only accesses items through comparisons, i.e., is comparison-based,
rendering it suitable, e.g., for floating-point numbers or strings ordered lexicographically.
Finally, our algorithm processes the input stream efficiently,
namely, its amortized update time is a logarithm of the space bound, i.e.,
$O\left(\log (\eps^{-1}) + \log\log (n) \right)$;
see Section~\ref{s:updatetime} for details.

\paragraph{Mergeability.} 
The ability to merge sketches of different streams to get an accurate sketch for the concatenation of the streams is highly significant both in theory~\cite{agarwal2013mergeable} and in practice~\cite{datasketches}.
Such mergeable summaries
enable trivial, automatic parallelization and distribution of processing massive data sets, 
by splitting the data up into 
pieces, summarizing each piece separately, and then merging the results in an arbitrary way. 
We say that a sketch is \emph{fully mergeable} if building it using any sequence of merge operations (executed on singleton items) leads
to the same guarantees as
if the entire data set had been processed as a single stream (i.e.,
always merging the sketch with one item).
We show that our sketch satisfies this strong definition of mergeability.

The following theorem is the main result of this paper.
We stress that our algorithm, which we call ReqSketch, does \emph{not} require any advance knowledge about $n$, the total size of the
input, which indeed may not be available in many applications.\footnote{A proof-of-concept Python implementation of our algorithm is available at GitHub:
\url{https://github.com/edoliberty/streaming-quantiles/blob/master/relativeErrorSketch.py}.
A production-quality implementation of ReqSketch is available in the Apache DataSketches library at \url{https://datasketches.apache.org/}.
}

\begin{restatable}{theorem}{thmintro}
\label{thm:mergeability-full-intro}
For any parameters $0 < \delta \le 0.5$ and $0 < \eps \le 1$,
there is a randomized, comparison-based, one-pass streaming algorithm that,
when processing a data stream consisting
of $n$ items from a totally-ordered universe $\mathcal{U}$,  produces a summary $S$ satisfying the following property. 
Given $S$, for any $y \in \mathcal{U}$ one can derive an estimate
$\hat{\rank}(y)$ of $\rank(y)$ such that
$$ \Pr\bigg[ |\hat{\rank}(y) - \rank(y)| > \eps \rank(y) \bigg] < \delta\,,$$
where the probability is over the internal randomness of the streaming algorithm.
The size of $S$ in memory words\footnote{
	A memory word can store any universe item or an integer with $O(\log(n + |\mathcal{U}|))$ bits. We express all the space bounds in memory words.
} is
$$O\left(\eps^{-1}\cdot \log^{1.5}(\eps n)\cdot \sqrt{\log\frac1\delta}\right)\,.$$
Moreover,
the summary produced is fully mergeable.
\end{restatable}

\paragraph{All-quantiles approximation.}
As a straightforward corollary of Theorem \ref{thm:mergeability-full-intro}, we obtain a space-efficient algorithm whose estimates
are simultaneously accurate for \emph{all} $y \in \mathcal{U}$ with high probability.
The proof is a standard use of the union bound combined
with an epsilon-net argument; we include the proof in Appendix \ref{app:corollaryintro}. 

\begin{restatable}[All-Quantiles Approximation]{corollary}{corintro}
\label{corollaryintro}
The error bound from Theorem \ref{thm:mergeability-full-intro} holds for all $y \in \mathcal{U}$ simultaneously
with probability $1-\delta$ when the size of the sketch in memory words is
$$O\left( \eps^{-1}\cdot \log^{1.5} (\eps n) \cdot \sqrt{\log\left(\frac{\log(\eps n)}{\eps \delta}\right)}\right)\,.$$
\end{restatable}

\paragraph{Technical overview.}
A starting point of the design of our algorithm is the KLL sketch~\cite{karnin2016optimal} that achieves 
optimal accuracy-space trade-off (of randomized algorithms) for the additive error guarantee. 
The basic building block of the algorithm is a buffer, called a \emph{compactor}, that receives an input stream of $n$ items and outputs a stream of at most $n/2$ items, meant to ``approximate'' the input stream. The buffer simply stores items and once it is full, we sort the buffer, output all items stored at either odd or even indexes (with odd vs.\ even selected via an unbiased coin flip), and clear the contents of the buffer---this is called the \emph{compaction operation}.
Note that a randomly chosen half of items in the buffer is simply discarded, whereas the other half of items in the buffer is ``output'' by the compaction
operation. 

The overall KLL sketch is built as a sequence of at most $\log_2(n)$ such compactors, such that the output stream of a compactor is treated as the input stream of the next compactor. We thus think of the compactors as arranged into \emph{levels}, with the first one at level~0.
Similar compactors were already used, e.g.,\ in~\cite{manku1998approximate,manku1999random,agarwal2013mergeable,luo16_quantiles_experimental}, and additional ideas are needed to get the optimal space bound for additive error, of $O(1/\eps)$ items stored across all compactors~\cite{karnin2016optimal}.

The compactor building block is not directly applicable to our setting for the following reasons. A first observation is that to achieve the relative error guarantee,
we need to always store the $1/\epsilon$ smallest items. 
This is because the relative error guarantee demands that estimated ranks for the $1/\eps$ lowest-ranked items in the data stream
are \emph{exact}. If even a single one of these items is deleted from the summary, then these 
estimates will not be exact.
Similarly, among the next $2/\epsilon$ smallest items, the summary must store essentially every other item to achieve multiplicative error.
Among the next $4/\epsilon$ smallest items in the order, the sketch must store roughly every fourth item, and so on. 

The following simple modification of the compactor from the KLL sketch indeed achieves the above. Each buffer of size $B$ ``protects'' the $B/2$
smallest items stored inside, meaning that these items are not involved in any compaction (i.e., the compaction operation only removes the $B/2$ largest items from the buffer).
Unfortunately, it turns out that this simple approach requires space $\Theta(\eps^{-2}\cdot \log(\eps^2 n))$, which merely matches
the space bound achieved in~\cite{zhang2006space}, and in particular has a (quadratically) suboptimal
dependence on $1/\eps$.

The key technical contribution of our work is to enhance this simple approach with a more sophisticated rule for selecting the number 
of protected items in each compaction. 
One solution that yields our upper bound
is to  choose this number in each compaction at random from an appropriate exponential distribution.
However, to get a cleaner analysis and a better dependency on the failure probability $\delta$, we in fact derandomize this distribution.

While the resulting algorithm is relatively simple,
analyzing the error behavior brings new challenges that do not arise in the additive error setting. 
Roughly speaking, 
when analyzing the accuracy of the estimate for $\rank(y)$ for any fixed item $y$, all error can be ``attributed'' to compaction
operations. In the additive error setting, one may suppose that \emph{every} compaction operation contributes to the error
and still obtain a tight error analysis \cite{karnin2016optimal}. Unfortunately,
this is not at all the case for relative error: as already indicated, to obtain our accuracy bounds it is essential to show that the estimate for 
any low-ranked item $y$ is affected by very few compaction operations. 

Thus, the first step of our analysis is to carefully bound the number of compactions on each level
that affect the error for $y$, using a charging argument that relies on the derandomized exponential distribution
to choose the number of protected items.
To get a suitable bound on the variance of the error, we also need to show that the rank of $y$
in the input stream to each compactor falls by about a factor of two at every level of the sketch.
While this is intuitively true (since each compaction operation outputs a randomly
chosen half of ``unprotected'' items stored in the compactor), it only holds with high probability and hence requires a 
careful treatment in the analysis.
Finally, we observe that the error in the estimate for $y$ is a zero-mean sub-Gaussian variable 
with variance bounded as above, and thus applying a standard Chernoff tail bound yields our
final accuracy guarantees for the estimated rank of $y$. 

There are substantial additional technical difficulties to analyze the algorithm under an arbitrary sequence of merge operations,
especially with no foreknowledge of the total size of the input.
Most notably, when the input size is not known in advance, the parameters of the sketch
must change as more inputs are processed. This
makes obtaining a tight bound on the variance of the resulting estimates highly involved.
In particular, as a sketch processes more and more inputs, it protects more and more items,
which means that items appearing early in the stream may \emph{not} be protected by the sketch,
even though they \emph{would have been protected} if they appeared later in the stream. 
Addressing this issue is reasonably simple in the streaming setting,
because every time the sketch parameters need to change, one can afford to allocate an
entirely new sketch with the updated parameters, without discarding the previous sketch(es);
see Section \ref{s:unknownlength} for details. 
Unfortunately, this simple approach does not work for a general sequence of merge operations,
and we provide a much more intricate analysis to give a fully mergeable summary.

A second challenge when designing and analyzing merge operations arises from working with our \emph{derandomized} exponential distribution,
since this requires each compactor to maintain a ``state'' variable determining
the current number of protected items, and these variables need to be ``merged'' appropriately.
It turns out that the correct way to merge state variables is to take a bitwise $\mathsf{OR}$
of their binary representations.
With this technique for merging state variables in hand, we extend the charging argument bounding the number of compactions affecting the error 
in any given estimate so as to handle an arbitrary sequence of merge operations.

\paragraph{Analysis with extremely small probability of failure.}
We close by giving an alternative analysis of our algorithm that achieves a space bound with an exponentially better (double logarithmic) dependence on $1/\delta$, compared to Theorem~\ref{thm:mergeability-full-intro}. However,
this improved dependence on $1/\delta$ comes at the expense of the exponent of $\log(\eps n)$ increasing from $1.5$ to $2$. 
Formally, we prove the following theorem in Section~\ref{s:extremeDelta},
where we also show that it directly implies a deterministic space bound of $O(\eps^{-1}\cdot \log^3(\eps n))$,
matching the state-of-the-art result in~\cite{zhangwang}.
For simplicity, we only prove the theorem in the streaming setting, although we conjecture that an appropriately
modified proof of Theorem~\ref{thm:mergeability-full-intro} would yield the same result even when 
the sketch is built using merge operations.

\begin{theorem} \label{thm:extremeDelta intro}
For any $0 < \delta \le 0.5$ and $0 < \eps \le 1$,
there is a ran\-do\-mi\-zed, comparison-based, one-pass streaming algorithm that computes a sketch consisting of
$O\left(\eps^{-1}\cdot \log^2(\eps n)\cdot \log\log(1/\delta)\right)$ universe items, and from which an estimate $\hat{\rank}(y)$ of $\rank(y)$ can be derived for every $y \in \mathcal{U}$. 
For any fixed $y \in \mathcal{U}$, with probability at least $1-\delta$,
the returned estimate satisfies the multiplicative error guarantee $|\hat{\rank}(y)-\rank(y) | \le \eps\rank(y)$.
\end{theorem}

We remark that this alternative analysis builds on an idea from~\cite{karnin2016optimal} to analyze the top few levels of compactors deterministically
rather than obtaining probabilistic guarantees on the errors introduced to estimates by these levels.

\paragraph{Organization of the paper.}
Since the proof of full mergeability in Theorem~\ref{thm:mergeability-full-intro} is quite involved, we proceed in several steps of increasing complexity.
We describe our sketch in the streaming setting in Section~\ref{s:description}, where we also give a detailed but informal outline of the analysis.
We then formally analyze the sketch in the streaming setting in Sections~\ref{s:relativecompactor} and~\ref{s:fullanalysis},
also assuming that a polynomial upper bound on the stream length $n$ is known in advance.
The space usage of the algorithm grows polynomially with the logarithm of this upper bound, so if this upper bound is at most $n^c$ for some constant $c \ge 1$, then 
the space usage remains as stated in Theorem \ref{thm:mergeability-full-intro}, with only the hidden constant factor changing. 
Then, in Section \ref{s:unknownlength}, we explain how to remove the assumption of a foreknowledge of $n$ in the streaming setting, yielding
an algorithm that works
without \emph{any} information about the final stream length. 

Finally, we fully describe the merge procedure and analyze the accuracy of our sketch under an arbitrary sequence of merge operations in Section~\ref{s:mergeability} (for didactic purposes, we outline a simplified merge operation in Section~\ref{s:mergeOperation-simplified}).
As mentioned above, Section~\ref{s:extremeDelta} contains an alternative analysis that yields better space bounds for
extremely small failure probabilities $\delta$.

\label{s:intromerge}

 \subsection{Detailed Comparison to Prior Work}
 \label{s:priorwork}
 Some prior works on streaming quantiles consider queries to be \emph{ranks} $r \in \{1, \dots, n\}$, 
and the algorithm must identify an item $y \in \mathcal{U}$ such that $\rank(y)$ is close to $r$; this is called the \emph{quantile query}. 
In this work, we focus on the dual problem of \emph{rank queries}, where we consider queries to be universe items $y \in \mathcal{U}$ and the algorithm must yield an accurate estimate for $\rank(y)$. 
Unless specified otherwise, algorithms described in this section directly solve both formulations (this holds for our algorithm as well).
Algorithms are randomized unless stated otherwise.
For simplicity, randomized algorithms are assumed to have constant failure probability $\delta$.
All reported space costs refer to the number of universe items stored.
(Apart from storing items, the algorithms may store, for example, bounds on ranks of stored items
or some counters, but the number of such $O(\log n)$-bit variables is proportional to the number of items stored or even smaller.)
 
 \paragraph{Additive Error.} 
 Manku, Rajagopalan, and Lindsay~\cite{manku1998approximate,manku1999random} built on the work of Munro and Paterson~\cite{munro1980selection} and gave a deterministic solution that stores at most $O(\eps^{-1}\cdot \log^2(\eps n))$ items,
 assuming prior knowledge of $n$. Greenwald and Khanna~\cite{greenwald2001space} created an intricate deterministic streaming algorithm that stores $O(\eps^{-1}\cdot \log(\eps n))$ items.
This is the best known deterministic algorithm for this problem, with a matching lower bound for comparison-based streaming algorithms \cite{cormode2019tight}. 
Agarwal et al.~\cite{agarwal2013mergeable} provided a mergeable sketch of size $O(\eps^{-1}\cdot \log^{1.5}(1/ \eps))$.
This paper contains many ideas and observations that were used in later work. Felber and Ostrovsky \cite{felber2015randomized} managed to reduce the space complexity to $O(\eps^{-1}\cdot\log(1/\eps))$ items by combining sampling with the Greenwald-Khanna sketches in non-trivial ways.
Finally, Karnin, Lang, and Liberty \cite{karnin2016optimal} resolved the problem by providing an $O(1/\eps)$-space solution, which is optimal. For general (non-constant)
failure probabilities $\delta$, the space upper bound becomes $O( \eps^{-1}\cdot \log\log(1/\delta))$, and they also prove a matching lower bound for 
comparison-based randomized algorithms, assuming $\delta \le 1 / n!$ (i.e., $\delta$ is exponentially small in $n$).

\paragraph{Multiplicative Error.}
A large number of works sought to provide more accurate quantile estimates for low or high ranks.
Only a handful offer solutions to the relative error quantiles problem considered in this work (sometimes also called the biased quantiles problem).
 Gupta and Zane \cite{Gupta:2003:CIL:644108.644150} gave
an algorithm for relative error quantiles that stores $O(\eps^{-3}\cdot \log^2(\eps n))$ items, and used this to approximately count
the number of inversions in a list; their algorithm requires prior knowledge of the stream length $n$. As previously mentioned, Zhang et al.~\cite{zhang2006space} presented an algorithm storing
$O(\eps^{-2}\cdot \log( \eps^2 n))$ universe items. Cormode et al.~\cite{cormode2006space} designed a deterministic sketch
storing $O( \eps^{-1}\cdot \log(\eps n) \cdot \log |\mathcal{U}|)$ items, which requires prior knowledge of the data universe $\mathcal{U}$. 
Their algorithm is inspired by the work of Shrivastava et al. \cite{shrivastava2004medians} in the additive error setting and
it is also mergeable (see \cite[Section 3]{agarwal2013mergeable}).
Zhang and Wang \cite{zhangwang} gave a deterministic merge-and-prune algorithm storing $O(\eps^{-1}\cdot \log^3(\eps n))$ items,
which can handle arbitrary merges with an upper bound on $n$, and streaming updates for unknown $n$.
However, it does not tackle the most general case of merging without a prior bound on $n$.
Cormode and Vesel{\'y} \cite{cormode2019tight}  recently showed a space lower bound of $\Omega(\eps^{-1}\cdot \log^2(\eps n))$ items for any deterministic comparison-based algorithm.

Other related works that do not fully solve the relative error  quantiles problem are as follows. Manku, Rajagopalan, and Lindsay \cite{manku1999random} designed 
an algorithm that, for a specified number $\phi \in [0,1]$, stores $O(\eps^{-1}\cdot \log(1/\delta))$ items and can return an item $y$ with $\rank(y)/n \in [(1-\eps)\phi, (1+\eps)\phi]$
(their algorithm requires prior knowledge of $n$).
Cormode et al. \cite{Cormode:2005:ECB:1053724.1054027} gave a deterministic algorithm that is meant to achieve
error properties ``in between'' additive and relative error guarantees.
That is, their algorithm aims to provide multiplicative guarantees only up to some minimum rank $k$;
for items of rank below $k$, their solution only provides additive guarantees. Their algorithm does not solve the relative error  quantiles problem: \cite{zhang2006space} observed that for adversarial item ordering, the algorithm of \cite{Cormode:2005:ECB:1053724.1054027} requires linear space to achieve relative error for all ranks.  

Dunning and Ertl \cite{DBLP:journals/corr/abs-1902-04023,dunning21} describe a 
heuristic algorithm called $t$-digest that is intended to achieve
relative error, but they provide no formal accuracy analysis
(note that $t$-digest is deterministic but \emph{not} comparison-based).
Indeed, Cormode et al.~\cite{cormode21_theory_meets_practice} show that the error of $t$-digest
may be arbitrarily large on adversarially generated inputs.
This latter paper also compares $t$-digest and ReqSketch (i.e., the algorithm of Theorem~\ref{thm:mergeability-full-intro}) on randomly generated inputs and proposes
implementation improvements for ReqSketch that make it process an input stream faster than $t$-digest.

Most recently, Masson, Rim, and Lee
\cite{DBLP:journals/pvldb/MassonRL19} considered a notion of relative \emph{value} error
for quantile sketches, which is distinct from the notion of relative \emph{rank} error considered in this work.
They require that for a query percentile $\phi \in [0, 1]$, if $y$ denotes the item in the data stream satisfying $\rank(y)=\phi n$, 
 then the algorithm should return an item $\hat{y} \in \mathcal{U}$ such that $|y - \hat{y}| \leq \eps \cdot |y|$. 
This definition only makes sense for data universes with a notion of magnitude and distance (e.g., numerical data), and
the definition is not invariant to natural data transformations, such as incrementing every data item $y$ by a large constant. 
It is also trivially achieved by maintaining a (mergeable) histogram with buckets
 $((1+\epsilon)^i, (1+\epsilon)^{i+1}]$.
In contrast, the standard notion of relative error considered in this
work does not refer to the data items themselves, only to their ranks,
and is arguably of more general applicability. 

\section{Description of the Algorithm}

\label{s:description}

\subsection{The Relative-Compactor Object}
\label{s:relative-compactor-def}
The crux of our algorithm is a building block that we call the relative-compactor. Roughly speaking, this object processes a stream of $n$ items and outputs a stream of at most $n/2$ items (each ``up-weighted'' by a factor of 2), meant
to ``approximate'' the input stream. It does so by maintaining a buffer of limited capacity.

Our complete sketch, described in Section \ref{s:fullsketch} below, is composed of a sequence of relative-compactors, where the input of the $(h+1)$-th relative-compactor is the output of the $h$-th. With at most $\log_2(\eps n)$ such relative-compactors of size $\Omega(\eps^{-1})$ (where $n$ is the length of the input stream), the output of the last relative-compactor is of size $O(1/\eps)$, and hence can be stored in memory.

\paragraph{Compaction Operations.} The basic subroutine used by our relative-compactor is a compaction operation. The input to a compaction operation is a list $X$ of $2m$ items $x_1 \leq x_2 \leq \ldots \leq x_{2m}$, and the output is a sequence $Z$ of $m$ items. This output is chosen to be one of the following two sequences, uniformly at random: Either $Z=\{x_{2i-1}\}_{i=1}^{m}$ or $Z=\{x_{2i}\}_{i=1}^{m}$. That is, the output sequence $Z$ equals either the even or odd indexed items in the sorted order of $X$, with both outcomes equally probable.

Consider an item $y \in {\calU}$ and recall that $\rank(y; X) = |\{ x\in X \ |\  x\le y\}|$ is the number of items $x\in X$ satisfying $x\le y$ (we remark that both $X$ and $\{ x\in X \ |\ x\le y\}$ are multisets of universe items).
The following is a trivial observation regarding the error of the rank estimate of $y$ with respect to the input $X$ of a compaction operation when using $Z$.
We view the output $Z$ of a compaction operation, with all items up-weighted by a factor of 2, as an approximation to the input $X$; for any $y$, its weighted rank in $Z$ should be close to its rank in $X$. Observation \ref{obs:even y} below states
 that this approximation incurs \emph{zero error} on items that have an even rank in $X$. Moreover,
for items $y$ that have an odd rank in $X$, the error for $y \in \mathcal{U}$ introduced by the compaction operation is $+1$ or $-1$ with equal probability. 
Note that the ranks are only w.r.t.\ to the input $X$ of the operation,
	which will contain a number of the largest items in a relative-compactor.
However, the observation holds for any universe item that may not be present in $X$.

\begin{obs} \label{obs:even y}
A universe item $y \in {\calU}$ is said to be even (odd) w.r.t.\ a compaction operation if $\rank(y; X)$ is even (odd), where $X$ is the input sequence to the operation. If $y$ is even w.r.t.\ the compaction, then $\rank(y; X) - 2\rank(y; Z) = 0$, where $Z$ is the output sequence of the operation. Otherwise, $\rank(y; X) - 2\rank(y; Z)$ is a variable taking a value from $\{-1,1\}$ uniformly at random.
\end{obs}

The observation that items of even rank (and in particular items of rank zero) suffer no error from a compaction operation plays an especially important role in the error analysis of our full sketch.

\paragraph{Full Description of the Relative-Compactor Object.}

\begin{figure}
 \begin{center}
  \includegraphics[width=0.7\textwidth]{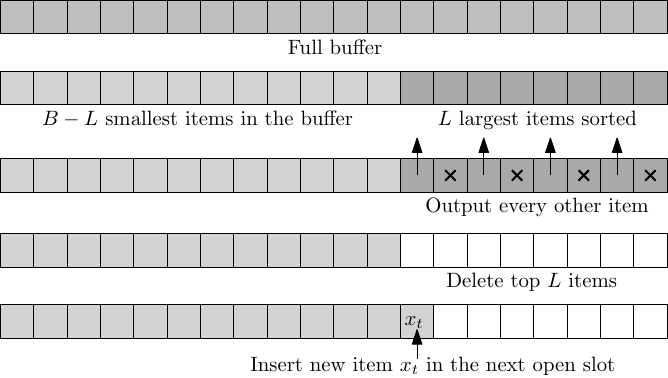}
 \end{center}
\caption{Illustration of the execution of a relative-compactor when inserting a new item $x_t$ into a buffer that is full at time $t$.
See lines \ref{line:9}-\ref{line:last} of Algorithm \ref{code:single}.}
\label{fig:buffer}
\end{figure}

The complete description of the relative-compactor object is given in Algorithm~\ref{code:single}. The high-level idea is as follows. The relative-compactor maintains a buffer of size $B = 2\cdot k\cdot \lceil\log_2(n/k)\rceil$ where
$k$ is an even integer parameter controlling the error and
$n$ is the upper bound on the stream length. (For now, we assume that such an upper bound is available;
we remove this assumption in Section~\ref{s:unknownlength}.)
The incoming items are stored in the buffer until it is full. At this point, we perform a compaction operation, as described above.

The input to the compaction operation is not all items in the buffer, but rather the largest $L$ items in the buffer
for a parameter $L \leq B/2$ such that $L$ is even.
These $L$ largest items are then removed from the buffer, and the output of the compaction operation is sent to the output stream of the buffer. This intuitively lets low-ranked items stay in the buffer longer than high-ranked ones. Indeed, by design the lowest-ranked half of items in the buffer are \emph{never} removed. We show later that this facilitates the multiplicative error guarantee.

The crucial part in the design of Algorithm \ref{code:single} is to select the parameter $L$ correctly, as $L$
controls the number of items compacted each time the buffer is full.
If we were to set $L = B/2$ for all compaction operations, then analyzing the worst-case behavior reveals that we need $B\approx 1/\eps^2$,
resulting in a sketch with a quadratic dependency on $1/\eps$.
To achieve the linear dependency on $1/\eps$,
we choose the parameter $L$ via a \emph{derandomized} exponential distribution
subject to the constraint that $L \leq B/2$.\footnote{A prior version of this manuscript used an
actual exponential distribution; see \url{https://arxiv.org/abs/2004.01668v1}. The algorithm presented here uses randomness only to select
which items to place in the output stream, not how many items to compact. This leads to a cleaner analysis and isolates the one component
of the algorithm for which randomness is essential. \label{footnotederandom}}

\begin{figure*}
	\begin{center}
		\includegraphics[width=\textwidth]{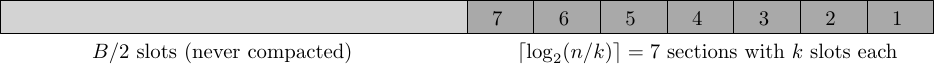}
	\end{center}
	\caption{Illustration of a relative-compactor and its sections, together with the indexes of the sections.}
	\label{fig:buffer-sections}
\end{figure*}

\begin{algorithm}[t]
	\caption{Relative-Compactor}\label{code:single}
\begin{algorithmic}[1]
	\Require Parameters $k \in 2\mathbb{N}^+$ and $n\in \mathbb{N}^+$, and a stream of items $x_1,x_2,\ldots$ of length at most $n$
	\State Set $B=2\cdot k\cdot \lceil\log_2(n/k)\rceil$ \label{line:setB}
	\State Initialize an empty buffer $\calB$ of size $B$, indexed from 1 
	\State Set $\schedulestate = 0$ \Comment{State of the compaction schedule} 
	\For{$t = 1\ldots $}
		\label{line:B} 
		\If{$\calB$ is full \label{line:9}} \Comment{Compaction operation}

			\State Compute $z({\schedulestate}) = $ the number of trailing
                        ones in the binary representation of $\schedulestate$ \label{line:S1}
			\State Set $L_\schedulestate = (z({\schedulestate})+1)\cdot k$ and $S_\schedulestate = B - L_\schedulestate + 1$  \label{line:S2}
			\State Pivot $\calB$ s.t.\ the largest $L_\schedulestate$ items occupy $\calB[S_\schedulestate:B]$
			\State	\Comment{$\calB[S_\schedulestate:B]$ are the last $L_\schedulestate$ slots of $\calB$}
			\State Sort $\calB[S_\schedulestate:B]$ in non-descending order 
			\State Output either even or odd indexed items
                        in the range $\calB[S_\schedulestate:B]$ with equal probability
			\State Mark slots $\calB[S_\schedulestate:B]$ in the buffer as clear
			\State Increase $\schedulestate$ by $1$ \label{line:compaction_end}
		\EndIf
		\State Store $x_t$ to the next available slot in the buffer $\calB$. \label{line:last}
	\EndFor
	
\end{algorithmic}
\end{algorithm}

In more detail, one can think of Algorithm \ref{code:single} as choosing $L$ as follows. During each compaction operation, the second half of the buffer (with $B/2$ largest items)
is split into $\lceil\log_2(n/k)\rceil$ sections, each of size $k$ and numbered from the right so that 
the first section contains the $k$ largest items, the second one next $k$ largest items, and so on;
see Figure~\ref{fig:buffer-sections}. \label{s:sections}
The idea is that the first section is involved in every compaction (i.e., we always have $L\ge k$),
the second section in every other compaction (i.e., $L\ge 2k$ every other time),
the third section in every fourth compaction, and so on.
This can be described concisely as follows: Let $\schedulestate$ be the number of compactions performed so far.
During the next (i.e., the $\schedulestate+1$-st) compaction of the relative-compactor, 
we set $L_\schedulestate = (z(\schedulestate) + 1)\cdot k$, where $z(\schedulestate)$ is the number of trailing ones in the binary representation of $\schedulestate$
(that is, if $\schedulestate$ viewed as a bitstring can be written as $(\mathbf{x}, 0, \mathbf{1}^j)$, for some $\mathbf{x}$, then $z(\schedulestate) = j$).
We call the variable $\schedulestate$ the \emph{state} of this ``compaction schedule''  (i.e., a particular way of choosing $L$).
See lines~\ref{line:S1}-\ref{line:S2} of Algorithm~\ref{code:single},
where we also define $S_\schedulestate = B - L_\schedulestate + 1$ as the first index in the compacted part of the buffer.

Observe that $L_\schedulestate \leq B/2$ always holds in Algorithm \ref{code:single}. Indeed, there are at most $n/k$ compaction operations (as each discards at least $k$ items),
so the binary representation of $\schedulestate$ never has more than $\lceil\log_2(n/k)\rceil$ bits, not even after the last compaction.
Thus, $z(\schedulestate)$, the number of trailing ones in the binary representation of $\schedulestate$, is always less than $\lceil\log_2(n/k)\rceil$
and hence, $L_\schedulestate \le \lceil\log_2(n/k)\rceil\cdot k = B/2$.
It also follows that there is at most one compaction operation that compacts all $\lceil\log_2(n/k)\rceil$ sections at once. 
Our deterministic compaction schedule has the following crucial property: 

\begin{obs} \label{keyfact} Between any two compaction operations that involve
exactly $j$ sections (i.e., both have $L = j\cdot k$), there is at least one compaction operation
that involves more than $j$ sections.
\end{obs}

\begin{proof}
Let $\schedulestate < \schedulestate'$ denote the states of the compaction schedule in two steps $t < t'$ with a compaction operation
involving exactly $j$ sections. Then we
can express the binary representations of $\schedulestate$ and $\schedulestate'$ as $(\mathbf{x}, 0, \mathbf{1}^{j-1})$
and $(\mathbf{x'}, 0, \mathbf{1}^{j-1})$, respectively, where $\mathbf{1}^{j-1}$ denotes the all-1s vector 
of length $j-1$, and $\mathbf{x}$ and $\mathbf{x'}$ are respectively the binary representations of two numbers $y$ and $z$ with $y < z$. 
Consider the binary vector $(\mathbf{x}, \mathbf{1}^{j})$. This is the binary representation
of a number $\hat{\schedulestate} \in (\schedulestate, \schedulestate')$ with strictly more trailing ones than the binary representations of $\schedulestate$ and $\schedulestate'$. 
The claim follows as there must be a step $\hat{t}\in (t, t')$ when the state equals $\hat{\schedulestate}$ and a compaction
operation is performed.
\end{proof}

\begin{algorithm}[t]
	\caption{ReqSketch (Relative-Error Quantiles sketch)}\label{code:full}
	\begin{algorithmic}[1]
		\Require Parameters $k \in 2\mathbb{N}^+$ and $n\in \mathbb{N}^+$, and a stream of items $x_1,x_2,\ldots$ of length at most $n$
		\Ensure A sketch answering rank queries
		\State Let RelCompactors be a list of relative-compactors
		\State Set $H=0$, initialize relative-compactor at RelCompactors[$0$], with parameters $k$ and $n$
		\For{$t = 1\ldots $}
		\State \Call{Insert}{$x_t, 0$}
		\EndFor
		
		\Function{Insert}{$x$,$h$}
		\If{$H<h$}
		\State Set $H=h$ 
		\State Initialize relative-compactor at RelCompactors[$h$], with parameters $k$ and $n$ 
		\EndIf
		\State Insert item $x$ into RelCompactors[$h$]
		\For{$z$ in output stream of RelCompactors[$h$]}
		\State \Call{Insert}{$z, h+1$}
		\EndFor
		\EndFunction
		
		\Function{Estimate-Rank}{$y$}
		\State Set $\hat{\rank}(y)=0$
		\For{$h=0$ to $H$}
		\For{each item $y' \leq y$ stored in RelCompactors[$h$]}
		\State Increment $\hat{\rank}(y)$ by $2^h$
		\EndFor
		\EndFor
		\Return $\hat{\rank}(y)$
		\EndFunction
	\end{algorithmic}
\end{algorithm}

\subsection{The Full Sketch}
\label{s:fullsketch}
Following prior work \cite{manku1998approximate, agarwal2013mergeable, karnin2016optimal}, the full sketch uses a sequence of relative-compactors. At the very start of the stream, it consists of a single relative-compactor (at level~0) and opens a new one (at level~1) once items are fed to the output stream of the first relative-compactor (i.e., after the first compaction operation, which occurs on the first stream update during which the buffer is full). In general, when the newest relative-compactor is at level $h$,
the first time the buffer at level $h$ performs a compaction operation (feeding items into its output stream for the first time), we open a new relative-compactor at level $h+1$ and feed it these items. Algorithm~\ref{code:full} describes the logic of this sketch.

To answer rank queries, we use the items in the buffers of the relative-compactors as a weighted coreset. That is, the union of these items is a weighted set $\calC$ of items, where the weight of items in relative-compactor at level $h$ is $2^h$ (recall that $h$ starts from 0), and the approximate rank of $y$, denoted $\hat{\rank}(y)$, is the sum of weights of items in $\calC$ smaller than or equal to $y$. 
Similarly, ReqSketch can answer quantile queries, i.e., for a given rank $r\in \{1, \dots, n\}$,
return an item $y\in \calU$ with $\rank(y)$ close to $r$; the algorithm just
returns an item $y$ stored in one of the relative-compactors with $\hat{\rank}(y)$
closest to the query rank $r$ among all items in the sketch.

The construction of layered exponentially-weighted compactors and the
subsequent rank estimation is virtually identical to that explained in
prior works \cite{manku1998approximate, agarwal2013mergeable,
  karnin2016optimal}. 
Our essential departure from prior work is in the definition of the compaction operation, not in how compactors are plumbed together to form a complete sketch. 

\subsection{Merge Operation}
\label{s:mergeOperation-simplified}

We describe a merge operation that takes as input two sketches $S'$ and $S''$ which have processed two separate streams $\sigma'$ and $\sigma''$, and that outputs a sketch $S$ summarizing
the concatenated stream $\sigma = \sigma'\circ \sigma''$ (the order of $\sigma'$ and $\sigma''$ does not matter here).
For simplicity, we assume w.l.o.g.\ that sketch $S'$ has at least as many levels as sketch $S''$.
Then, the resulting sketch $S$ inherits parameters $k$ and $n$ from sketch $S'$,
and in fact, we will merge sketch $S''$ into $S'$.
We further assume that both $S'$ and $S''$ have the same value of $k$ (otherwise, it would not be meaningful to analyze the error of $S$)
and that $n$ is still an upper bound on the combined input size.
Later, in Section~\ref{s:mergeability}, we show how to remove the
latter assumption and provide a tight analysis of the sketch created by an arbitrary sequence of merge operations
without any advance knowledge about the total input size, thus proving Theorem~\ref{thm:mergeability-full-intro}.

\begin{algorithm}[t]
	\caption{Merge operation}\label{code:merge-simple}
	\begin{algorithmic}[1]
		\Require Sketches $S'$ and $S''$ to be merged such that $S'.H \ge S''.H$ and $S'.k = S''.k$
		\Ensure A sketch answering rank queries for the combined inputs of $S'$ and $S''$
		\For{$h = 0, \ldots, S''.H$} \Comment{Merge $S''$ into $S'$}
		\State $S'$.RelCompactors[$h$].$\schedulestate$ = $S'$.RelCompactors[$h$].$\schedulestate$ \textsf{OR} $S''$.RelCompactors[$h$].$\schedulestate$ \label{algLine:merge-simple:combine_states}
		\State Insert all items in $S''$.RelCompactors[$h$] into $S'$.RelCompactors[$h$]
		\EndFor
		\For{$h = 0, \ldots, S'.H $}
		\If{buffer $S'$.RelCompactors[$h$] exceeds its capacity}
		\State Perform compaction operation as in lines~\ref{line:S1}-\ref{line:compaction_end}
		of Algorithm~\ref{code:single} and insert output items into $S'$.RelCompactors[$h+1$]
		\EndIf
		\EndFor
		\State \Return $S'$
	\end{algorithmic}
\end{algorithm}

The basic idea of the merge operation is straightforward: At each level,
concatenate the buffers and if that causes the capacity of the compactor to be exceeded, perform the compaction operation,
as in Algorithm~\ref{code:single}. However, there is a crucial subtlety: We need to combine the states $\schedulestate$ of the compaction schedule at each level in a manner 
that ensures that relative-error guarantees are satisfied for the merged sketch.
Consider a level $h$
and let $\schedulestate'$ and $\schedulestate''$ be the states of the compaction schedule at level $h$ in $S'$ and $S''$, respectively.
The new state $\schedulestate$ at level $h$ will be the bitwise OR of $\schedulestate'$ and $\schedulestate''$.
We explain the intuition behind using the bitwise OR in Section~\ref{s:mergeability}.
Note that while in the streaming setting, the state corresponds to the number of compaction operations already performed,
after a merge operation this may not hold anymore.
Still, if the state is zero, this indicates that the buffer has not
yet been subject to any compactions. 
Algorithm~\ref{code:merge-simple} provides a pseudocode of the merge operation,
where we use $S.H$ for the index of the highest level of sketch $S$ and similarly,
$S.k$ for the parameter $k$ of $S$.

\subsection{Informal Outline of the Analysis}

To analyze the error of the full sketch, we focus on the error in the
estimated rank of an arbitrary item $y \in {\calU}$.
For clarity in this informal overview,
we consider the failure probability $\delta$ to be constant, and 
we assume that $\eps^{-1} > \sqrt{\log_2(\eps n)}$, or equivalently, $n < \eps^{-1}\cdot 2^{\eps^{-2}}$.
Recall that in our algorithm, all buffers have size $B=\Theta(k \log(n/k))$;
we ultimately will set $k=\Theta\left(\eps^{-1}/\sqrt{\log(\eps n)}\right)$, in which case $B=O\left(\eps^{-1} \sqrt{\log(\eps n)}\right)$. 
 
Let $\rank(y)$ be the rank of item $y$ in the input stream, and $\err(y) = \hat{\rank}(y) - \rank(y)$ the error of the estimated rank for $y$.
Our analysis of $\err(y)$ relies on just two properties.
\begin{enumerate}
\item
The level-$h$ compactor only does at most $\rank(y)/(k\cdot 2^h)$ compactions that affect the error of $y$ (up to a constant factor). 

Roughly speaking, this holds by the following reasoning. 
First, recall from Observation~\ref{obs:even y} that $y$ needs to be odd w.r.t.\ any compaction affecting the error of $y$,
which implies that at least one item $x\le y$ must be removed during that compaction.
We show that as we move up one level at a time, $y$'s rank with respect to
the input stream fed to that level falls by about half (this is formally established in Lemma \ref{lem:halving}). This is 
the source of the $2^h$ factor in the denominator. Second,
we show that each compaction operation that affects $\err(y)$
can be ``attributed'' to $k$ items smaller than or equal to $y$
inserted into the buffer, which relies on using our particular compaction schedule (see Lemma \ref{lem:single layer}). 
This is the source of the $k$ factor in the denominator.

\item  Let $H_y$ be the smallest positive integer such that $2^{H_y}\gtrsim \rank(y)/B$ (the approximate inequality $\gtrsim$ hides a universal constant). 
Then no compactions occurring at levels above $H_y$ affect $\err(y)$, because $y$'s rank relative to the input stream of any such buffer is less than $B/2$ and no relative-compactor ever compacts the lowest-ranked $B/2$ items that it stores.

Again, this holds because, as we move up one level at a time, $y$'s rank w.r.t.\
each level falls by about half  (see Lemma \ref{lem:halving}).
\end{enumerate}
Together, this means that the variance of the estimate for $y$ is at most (up to constant factors):
\begin{align} \sum_{h=1}^{H_y} \frac{\rank(y)}{k\cdot 2^h}  \cdot 2^{2 h} =
 \sum_{h=1}^{H_y} \frac{\rank(y)}{k}  \cdot 2^h\,, \label{die}
\end{align}
where in the LHS, $ \rank(y)/(k 2^h) $ bounds the number of level-$h$ compaction operations affecting the error (this exploits Property 1 above), and $2^{2h}$ is the variance contributed by each such compaction, due to Observation~\ref{obs:even y} and because each item processed by the relative-compactor at level $h$ represents $2^h$ items in the original stream.

The RHS of Equation \eqref{die} is dominated by the term for $h=H_y$, and the term for that value of $h$ is
at most (up to constant factors)
\begin{equation} \label{keyeq}
\frac{\rank(y)}{k}  \cdot 2^{H_y} \lesssim  \frac{\rank(y)}{k}\cdot \frac{\rank(y)}{B} = \frac{\rank(y)^2}{k\cdot B}
\simeq \frac{\rank(y)^2\cdot \log(\eps n)}{B^2}\,.
\end{equation}
The first inequality in Equation \eqref{keyeq} exploits Property~2 above, while the last equality exploits the fact that $B = O(k \cdot\log(\eps n))$.\footnote{In the derivations within Equation \eqref{keyeq}, there is a couple of important subtleties. The first is that when we replace $2^{H_y}$ with $\Theta(\rank(y)/B)$, that substitution is only valid if $\rank(y)/B \geq \Omega(1)$. However, we can assume w.l.o.g.\ that $\rank(y)\geq B/2$, as otherwise the algorithm will make no error on $y$ by virtue of storing the lowest-ranked $B/2$ items deterministically. The second subtlety is that the algorithm
is only well-defined if $k\geq 2$, so when we replace $k$ with $\Theta(B/\log(\eps n))$, that is a valid substitution only if $B \geq \Omega( \log(\eps n))$, which holds by the assumption that $\eps^{-1} > \sqrt{\log_2(\eps n)}$.} We obtain the desired accuracy guarantees so long as this variance is at most $\eps^2 \rank(y)^2$,  as this will imply that the standard deviation is at most $\eps \rank(y)$. This hoped-for variance bound holds so long as $B\gtrsim \eps^{-1} \cdot \sqrt{\log_2(\eps n)}$, or
equivalently $k \gtrsim \eps^{-1}/\sqrt{\log_2(\eps n)}$.

\subsection{Roadmap for the Formal Analysis}
Section \ref{s:relativecompactor} establishes the necessary properties of a single relative-compactor (Algorithm \ref{code:single}), namely that, roughly speaking, each compaction operation that affects a designated item $y$
can be charged to $k$ items smaller than or equal to $y$ added to the buffer.
Section \ref{s:fullanalysis} then analyzes the full sketch (Algorithm \ref{code:full}),
completing the proof of our result in the streaming setting when a polynomial upper bound on $n$ is known in advance.
In Section~\ref{s:unknownlength}, we provide a simple argument that the assumption of having such an upper bound on $n$ is not needed in the streaming setting.

For the most general analysis under an arbitrary sequence of merge operations (i.e., for the proof of full mergeability) and without assuming a foreknowledge of $n$, we refer to Section~\ref{s:mergeability}.

\section{Analysis of the Relative-Compactor in the Streaming Setting}
\label{s:relativecompactor}

To analyze our algorithm, we keep track of the error associated with an arbitrary fixed item $y$. 
Throughout this section, we restrict our attention to any single relative-compactor at some level $h$ (Algorithm~\ref{code:single}) maintained by our sketching algorithm (Algorithm~\ref{code:full}),
and we use ``time $t$'' to refer to the $t$-th insertion operation to this particular relative-compactor.

We analyze the error introduced by the relative-compactor for an item $y$. 
Specifically, at time $t$, let $X^t=(x_1,\ldots,x_t)$ be the input stream
to the relative-compactor, $Z^t$ be the output stream, and $\calB^t$ be the items in the buffer after inserting item $x_t$.
The error for the relative-compactor at time $t$ with respect to item $y$ is defined as 
\begin{equation} \label{eq:err def}
\err^t_h(y) = \rank(y; X^t) - 2\rank(y; Z^t) - \rank(y; \calB^t).
\end{equation}

Conceptually, $\err^t_h(y)$ tracks the difference between $y$'s rank in the input stream $X^t$ at time $t$ versus its rank as estimated by the combination
of the output stream
and the remaining items in the buffer at time $t$ (output items are upweighted by a factor of $2$ while items remaining in the buffer 
are not). 
The overall error of the relative-compactor is $\err^n_h(y)$, where $n$ is the length of its input stream. To bound $\err^n_h(y)$, we keep track of the error associated with $y$ over time, and define the increment or decrement of it as
$$ \Delta^t_h(y) = \err^t_h(y) - \err^{t-1}_h(y),$$
where $\err^0_h(y) = 0$.

Clearly, if the algorithm performs no compaction operation in a time step $t$, then  $\Delta^t_h(y)=0$.
(Recall that a compaction is an execution of lines~\ref{line:S1}-\ref{line:compaction_end} of
Algorithm \ref{code:single}.)
Let us consider what happens in a step $t$ in which a compaction operation occurs. Recall from Observation~\ref{obs:even y} that if $y$ is even with respect to the compaction (i.e., $y$ has even rank w.r.t.\ the $L$ largest items in the relative-compactor),
then $y$ suffers no error, meaning that $\Delta^t_h(y)=0$. Otherwise, $\Delta^t_h(y)$ is uniform in $\{-1,1\}$.

Our aim is to bound the number of steps $t$ with $\Delta^t_h(y) \neq
0$, equal to $\sum_{t=1}^n |\Delta^t_h(y)|$, and use this in turn to
help us bound $\err^n_h(y)$.
We call a step $t$ with $\Delta^t_h(y) \neq 0$ \emph{important}.
Likewise, call an item $x$ with $x\le y$ \emph{important}.
Let $\rank_h(y)$ be the rank of $y$ in the input stream to level $h$; so there
are $\rank_h(y)$ important items inserted to the buffer at level $h$
(in the notation above, we have $\rank_h(y) = \rank(y; X^n)$).
Recall that $k$ denotes the parameter in Algorithm \ref{code:single} controlling
the size of the buffer sections of each relative-compactor and that $B$ denotes the buffer's capacity.

Our main analytic result regarding relative-compactors is
that there are at most $\rank_h(y)/k$ important steps.
Its proof explains the intuition behind our compaction schedule,
i.e., why we set $L$ as described in Algorithm~\ref{code:single}.

\begin{lemma} \label{lem:single layer}
Consider the relative-compactor at level $h$, fed an input stream of length at most $n$.
For any fixed item $y \in \mathcal{U}$ with rank $\rank_h(y)$ in the input stream to level $h$,
there are at most $\rank_h(y) / k$ important steps.
In particular,
$$ \sum_{t=1}^n |\Delta^t_h(y)| \leq \rank_h(y) / k \quad \text{and} \quad \left|\err^n_h(y)\right| \leq \rank_h(y) / k\,.$$ 
\end{lemma}

\begin{proof}
We focus on steps $t$ in which the algorithm performs a level-$h$ compaction operation
(possibly not important), 
and call a step $t$ a $j$-step for $j\ge 1$ if the compaction operation in step $t$ (if any) involves exactly $j$ sections
(i.e., $L_\schedulestate = j\cdot k$ in line~\ref{line:S2} of Algorithm \ref{code:single}).
Recall from Section \ref{s:sections} that sections are numbered from the right, so that 
the first section contains the $k$ largest items in the buffer, the second section contains the next $k$ largest items, and so on.
Note that we think of the buffer as being sorted at all times.

For any $j\ge 1$, let $s_j$ be the number of important $j$-steps.
Further, let $\rank_{h,j}(y)$ be the number of important items that are either removed from 
the $j$-th section during a compaction, or remain in the $j$-th section 
at the end of execution, i.e., after the relative-compactor has processed its entire input stream. 
We also define $\rank_{h,j}(y)$ for $j = \lceil\log_2(n/k)\rceil + 1$; for this $j$, we define the $j$-th section 
to be the last $k$ slots in the first half of the buffer (which contains $B/2$ smallest items).
This special section is never involved in any compaction.

Observe that $\sum_{j\ge 1} s_j$ is the number of important steps and that $\sum_{j\ge 1} \rank_{h,j}(y) \le \rank_{h}(y)$.
We will show 
\begin{equation}\label{eq:singleLayer-chargingIneq}
s_j\cdot k \le \rank_{h,j+1}(y)\,.
\end{equation}
Intuitively, our aim is to ``charge'' each important $j$-step to $k$ important items that are either removed from section $j+1$, or remain in section $j+1$ at the end 
of execution, so that each such item is charged at most once.

Equation~\ref{eq:singleLayer-chargingIneq} implies the lemma as the number of important steps is
$$\sum_{t=1}^n |\Delta^t(y)|
= \sum_{j = 1}^{\lceil\log_2(n/k)\rceil} s_j
\le \sum_{j = 1}^{\lceil\log_2(n/k)\rceil} \frac{\rank_{h,j+1}(y)}{k}
\le \frac{\rank_{h}(y)}{k}\,.$$

To show the lower bound on $\rank_{h,j+1}(y)$ in~\eqref{eq:singleLayer-chargingIneq}, consider an important $j$-step $t$.
Since the algorithm compacts exactly $j$ sections and $\Delta^t_h(y) \neq 0$,
there is at least one important item in section $j$ by Observation~\ref{obs:even y}.
As section $j+1$ contains smaller-ranked (or equal-ranked) items than section $j$,
section $j+1$ contains important items only.
We have two cases for charging the important $j$-step $t$:

\mycase{A} There is a compaction operation after step $t$ that involves at least $j+1$ buffer sections,
i.e., a $j'$-step for $j'\ge j+1$.
Let $t'$ be the first such step. Note that just before the compaction in step $t'$, the $(j+1)$-st section 
contains important items only as it contains important items only immediately after step $t$.
We charge the important step $t$ to the $k$ important items that are in the $(j+1)$-st section just before
step $t'$. Thus, all of these charged items are removed from level $h$ in step $t'$.

\mycase{B} Otherwise, there is no compaction operation after step $t$ that involves at least $j+1$ buffer sections.
Then, we charge step $t$ to the $k$ important items that are in the $(j+1)$-st section 
at the end of execution. 

\smallskip

It remains to observe that each important item $x$ accounted for in $\rank_{h,j+1}(y)$ is charged at most once.
(Note that different compactions may be charged to the items which are
consumed during the same later compaction, but our charging will
ensure that these are assigned to different sections.
For example, consider a sequence of three consecutive important steps (there is no compaction in other steps in between them) such that in the first one the algorithm compacts 2
sections, then 1 section, and 3 sections in the third important step.
The first compaction will be charged to section~3 of the last
compaction, and the second compaction is charged to section 2 of the
last compaction.)

Formally, suppose that $x$ is removed from section $j+1$ during some
compaction operation in a step $t'$.
Item $x$ may only be charged by some number of important $j$-steps
before step $t'$ (satisfying the condition of Case~A).
To show there is at most one such important step, we use the crucial property of
our compaction schedule (Observation~\ref{keyfact}) that between every two compaction operations involving exactly $j$ sections,
there is at least one compaction that involves more than $j$ sections.
Since any important $j$-step is charged to the first subsequent compaction
that involves more than $j$ sections, item $x$ is charged at most once.

Otherwise, $x$ remains in section $j+1$ of the level-$h$ buffer at the
end of processing.
The proof in this case is similar to the previous case.
Item $x$ may only be charged by some number of important $j$-steps (that fall into Case~B) such that there
are no subsequent compaction operations involving at least $j+1$ buffer sections, 
and there is at most one such important step by Observation~\ref{keyfact}.
This shows~\eqref{eq:singleLayer-chargingIneq}, which implies the lemma as noted above. 
\end{proof}

\section{Analysis of the Full Sketch in the Streaming Setting}
\label{s:fullanalysis}
We denote by $\err_h(y)$ the error for item $y$ at the end of the stream when comparing the input stream to the compactor of level $h$ and its output stream and buffer. That is, letting $\calB_h$ be the items in the buffer of the level-$h$ relative-compactor after 
Algorithm \ref{code:full} has processed the input stream,
\begin{equation} \label{eq:layer err}
\err_h(y) = \rank_h(y) - 2\rank_{h+1}(y) - \rank(y; \calB_h).
\end{equation}

For the analysis, we first set the value of parameter $k$ of Algorithm \ref{code:full}.
Namely, given (an upper bound on) the stream length $n$,
the desired accuracy $0 < \eps \le 1$, and the desired upper bound $0 < \delta \le 0.5$ on failure probability, we let
\begin{equation}\label{eq:setk}
k = 2\cdot \left\lceil \frac{4}{\eps}\cdot \sqrt{\frac{\ln\frac{1}{\delta}}{\log_2(\eps n)}}\right\rceil\,.
\end{equation}
In the rest of this section, we suppose that parameters $\eps$ and $\delta$ satisfy $\delta > 1/\exp(\eps n/64)$
(note that this a very weak assumption as for $\delta \le 1/\exp(\eps n/64)$, the accuracy guarantees hold nearly deterministically,
the space cost of $\sqrt{\ln(1 / \delta)}$ becomes $\Omega(\sqrt{\eps n})$,
and furthermore, the analyses in Sections~\ref{s:mergeability} and~\ref{s:extremeDelta} do not require such an assumption).
We start by showing a lower bound on $k\cdot B$.

\begin{claim}\label{clm:kBlowerBound}
If parameter $k$ is set according to Equation~\eqref{eq:setk} and $B$ is set as in Algorithm~\ref{code:single} (line \ref{line:setB}), then
the following inequality holds: 
\begin{equation}\label{eqn:boundOn_kB}
k\cdot B \ge 2^6\cdot \frac{1}{\eps^2}\cdot \ln\frac{1}{\delta}\,.
\end{equation}
\end{claim}

\begin{proof}
We first need to relate
$\log_2(n/k)$ (used to define $B$, see line~\ref{line:setB} of Algorithm~\ref{code:single}) and $\log_2(\eps n)$ (that appears in the definition of $k$, see Equation \eqref{eq:setk}).
Using the assumption $\delta > 1/\exp(\eps n/64)$,
we have
$k \le 8\eps^{-1}\cdot \sqrt{\ln(1/\delta)} \le 8\eps^{-1}\cdot \sqrt{\eps n/64} = \eps^{-1}\cdot \sqrt{\eps n}$,
which gives us $$\log_2 \left( \frac{n}{k} \right) \ge \log_2 \left(\frac{\eps n}{\sqrt{\eps n}}\right) = \frac{\log_2 (\eps n)}{2}\,.$$
Using this and the definition of $k$, we bound $k\cdot B$ as follows:
$$
k\cdot B = 2\cdot k^2\cdot \left\lceil\log_2 \frac{n}{k}\right\rceil 
\ge 2\cdot 2^6\cdot \frac{1}{\eps^2}\cdot \frac{\ln\frac{1}{\delta}}{\log_2(\eps n)}\cdot \frac{\log_2 (\eps n)}{2}
= 2^6\cdot \frac{1}{\eps^2}\cdot \ln\frac{1}{\delta}\,.
$$
\end{proof}

We now provide bounds on the rank of $y$ on each level,
starting with a simple one that will be useful for bounding the maximum
level $h$ with $\rank_h(y) > 0$.

\begin{obs} \label{obs:decay rank 1}
$\rank_{h+1}(y) \leq \max\{0, \rank_h(y) - B/2\}$ for any $h\ge 0$.
\end{obs}

\begin{proof}
Since the lowest-ranked $B/2$ items in the input stream to the level-$h$ relative-compactor are stored in the buffer $\calB_h$ and never given to the output stream of the relative-compactor,
it follows immediately that $\rank_{h+1}(y) \leq \max\{0, \rank_h(y) - B/2\}$. 
\end{proof}

Next, we prove that $\rank_h(y)$ roughly halves with every level.
This is easy to see in expectation and we show that it is true with high probability up to a certain crucial level $H(y)$.
Here, we define $H(y)$ to be the minimal $h$ for which $2^{2-h} \rank(y) \leq B/2$.
For $h = H(y) -1$ (assuming $H(y) > 0$), we particularly have $2^{3-H(y)} \rank(y) \geq B/2$,
or equivalently \begin{equation} 2^{H(y)} \leq 2^4\cdot \rank(y) / B. \label{eq:2to4}\end{equation}
Below, in Lemma~\ref{lem:decay rank}, we show that no important item (i.e., one smaller than or equal to $y$) can ever reach level $H(y)$ with high probability, unless $H(y) = 0$.
(For $H(y) = 0$, all important items fit into the level-$0$ buffer, so the estimated rank $\hat{\rank}(y)$ equals $\rank(y)$.)
Recall that a zero-mean random variable $X$ with variance $\sigma^2$ is sub-Gaussian
if $\E[\exp(sX)] \le \exp(-\frac12\cdot s^2\cdot \sigma^2)$ for any $s\in \R$; note that 
a (weighted) sum of independent zero-mean sub-Gaussian variables is a zero-mean sub-Gaussian random variable as well.
We will use the following standard (Chernoff) tail bound for sub-Gaussian variables~(see, e.g., Lemma 1.3 in~\cite{subgaussian_chernoff}):

\begin{fact}\label{fact:tailBound-subGaussian}
Let $X$ be a zero-mean sub-Gaussian variable with variance at most $\sigma^2$.
Then for any $a > 0$, it holds that
$$\Pr[X > a] \le \exp\left(-\frac{a^2}{2\sigma^2}\right)\quad\text{and}\quad\Pr[X < -a] \le \exp\left(-\frac{a^2}{2\sigma^2}\right)\,.$$
\end{fact}

\begin{lemma} \label{lem:decay rank} \label{lem:halving}
Assuming $H(y) > 0$,
with probability at least $1 - \delta$ it holds that $\rank_h(y) \leq 2^{-h+1}\rank(y)$ for any $h < H(y)$.
\end{lemma}

\begin{proof}
We prove by induction on $0\le h < H(y)$ that, conditioned on $\rank_{\ell}(y) \leq 2^{-\ell+1}\rank(y)$ for any $\ell < h$,
with probability at least $1 - \delta\cdot 2^{h - H(y)}$ it holds that $\rank_h(y) \leq 2^{-h+1}\rank(y)$.
Taking the union bound over all $0\le h < H(y)$ implies the claim.
As $\rank_0(y) = \rank(y)$, the base case follows immediately.

Next, consider $h > 0$ and condition on $\rank_{\ell}(y) \leq 2^{-\ell+1}\rank(y)$ for any $\ell < h$.
Observe that any compaction operation at any level $\ell$ that involves $a$ important items
inserts $\frac12 a$ such items to the input stream at level $\ell+1$ in expectation, no matter whether $a$ is odd or even.
Indeed, if $a$ is odd, then the number of important items promoted is $\frac12 (a + X)$,
where $X$ is a zero-mean random variable uniform on $\{-1, 1\}$.
For an even $a$,
the number of important items that are promoted is $\frac{1}{2}a$ with probability 1.

Thus, random variable $\rank_{\ell}(y)$ for any level $\ell > 0$ is generated by the following random process:
To get $\rank_{\ell}(y)$, start with $\rank_{\ell - 1}(y)$ important items and remove those stored
in the level-$(\ell-1)$ relative-compactor $\calB_{\ell-1}$ at the end of execution;
there are $\rank(y; \calB_{\ell-1}) \le B$ important items in $\calB_{\ell-1}$.
Then, as described above, each compaction operation at level $\ell-1$ involving $a > 0$ important items promotes to level $\ell$
either $\frac12 a$ important items if $a$ is even, or $\frac12 (a + X)$ important items if $a$ is odd, i.e., the compaction is important.
In total, $\rank_{\ell - 1}(y) - \rank(y; \calB_{\ell-1})$ important items are involved in compaction operations at level $\ell-1$.
Summarizing and letting $m_{\ell-1}$ be the number of important compaction operations at level $\ell-1$, we have
\begin{equation}\label{eqn:rank_process}
\rank_\ell(y) = \frac12\cdot \left(\rank_{\ell - 1}(y) - \rank(y; \calB_{\ell-1}) + \mathrm{Binomial}(m_{\ell-1})\right) \,,
\end{equation}
where $\mathrm{Binomial}(n)$ represents the sum of $n$ zero-mean i.i.d.\ random variables uniform on $\{-1, 1\}$.

To simplify~\eqref{eqn:rank_process},
consider the following sequence of random variables $Y_0, \dots, Y_h$: Start with $Y_0 = \rank(y)$ and for $0 < \ell < h$ let
\begin{equation}\label{eqn:rank_process_Y}
Y_\ell = \frac12\cdot \left(Y_{\ell - 1} + \mathrm{Binomial}(m_{\ell-1}) \right)\,.
\end{equation}
Note that $\E[Y_\ell] = 2^{-\ell} \rank(y)$. 
Since variables $Y_\ell$ differ from $\rank_\ell(y)$ only by not subtracting $\rank(y; \calB_{\ell-1})$ at every level $\ell > 0$,
variable $Y_h$ stochastically dominates variable $\rank_h(y)$, so in particular,
\begin{equation}\label{eqn:rank_stoch_dominance}
\Pr[\rank_h(y) > 2^{-h+1}\rank(y)] \le \Pr[Y_h > 2^{-h+1}\rank(y)]\,,
\end{equation}
which implies that it is sufficient to bound $\Pr[Y_h > 2^{-h+1}\rank(y)]$.
Unrolling the definition of $Y_h$ in~\eqref{eqn:rank_process_Y}, we obtain
\begin{equation}\label{eqn:rank_process_Y2}
Y_h = 2^{-h}\cdot \rank(y) + \sum_{\ell = 0}^{h - 1} 2^{-h + \ell}\cdot \mathrm{Binomial}(m_\ell) \,.
\end{equation}
Observe that $Y_h$ equals a fixed amount ($2^{-h}\cdot \rank(y)$) plus a zero-mean sub-Gaussian variable
\begin{equation}\label{eqn:rank_process_Z}
Z_h = \sum_{\ell = 0}^{h - 1} 2^{-h + \ell}\cdot \mathrm{Binomial}(m_\ell)\,,
\end{equation}
since $\mathrm{Binomial}(n)$ is a sum of $n$ independent zero-mean sub-Gaussian variables (with variance 1).

To bound the variance of $Z_h$, first note that for any $\ell < h$, we have
$m_\ell\le \rank_\ell(y) / k \le 2^{-\ell+1}\rank(y) / k$ by Lemma~\ref{lem:single layer} and
by conditioning on $\rank_{\ell}(y) \leq 2^{-\ell+1}\rank(y)$.
As $\V[\mathrm{Binomial}(n)] = n$, the variance of $Z_h$ is
\begin{align*}
\V[Z_h] \le \sum_{\ell = 0}^{h - 1} 2^{-2h + 2\ell}\cdot m_\ell
\le \sum_{\ell = 0}^{h - 1} 2^{-2h + 2\ell}\cdot \frac{2^{-\ell+1}\rank(y)}{k}
= \sum_{\ell = 0}^{h - 1} \frac{2^{-2h + \ell + 1}\rank(y)}{k}\le \frac{2^{-h + 1}\cdot \rank(y)}{k}\,.
\end{align*}

Note that $\Pr[Y_h > 2^{-h+1}\rank(y)] = \Pr[Z_h > 2^{-h}\rank(y)]$.
To bound the latter probability, we apply the tail bound for sub-Gaussian variables (Fact~\ref{fact:tailBound-subGaussian}) to get
\begin{align}
\Pr[Z_h > 2^{-h}\rank(y)]
&< \exp\left(-\frac{2^{-2h}\cdot \rank(y)^2}{2\cdot (2^{-h + 1}\cdot \rank(y) / k)} \right)
\nonumber\\
&= \exp\left(-2^{-h - 2}\cdot \rank(y)\cdot k \right)
\nonumber\\
&= \exp\left(-2^{-h + H(y) - 6}\cdot 2^{4 - H(y)} \rank(y)\cdot k \right)
\nonumber\\
&\le \exp\left(-2^{-h + H(y) - 6}\cdot B\cdot k \right) \label{eqn:Z_h_tail_bound_ineq2}
\\
&\le \exp\left(-2^{-h + H(y) - 6}\cdot 2^6\cdot \frac{1}{\eps^2}\cdot \ln\frac{1}{\delta} \right) \label{eqn:Z_h_tail_bound_ineq3}
\\
&\le \exp\left(-2^{-h + H(y)}\cdot \ln\frac{1}{\delta} \right)
= \delta^{2^{H(y) - h}}
\le \delta\cdot 2^{ - H(y) + h}\,, \label{eqn:Z_h_tail_bound_ineq4}
\end{align}
where inequality~\eqref{eqn:Z_h_tail_bound_ineq2} uses $2^{4 - H(y)} \rank(y) \ge B$ (by the definition of $H(y)$, cf. Equation \eqref{eq:2to4}),
inequality~\eqref{eqn:Z_h_tail_bound_ineq3} follows from Claim~\ref{clm:kBlowerBound},
inequality~\eqref{eqn:Z_h_tail_bound_ineq4} uses $\eps \le 1$, and the last inequality uses $\delta \le 0.5$.
As explained above, this concludes the proof.
\end{proof}

In what follows, we condition on the bound on $\rank_h(y)$ in Lemma~\ref{lem:decay rank} for any $h < H(y)$.

\begin{lemma} \label{lem:ranky0}
Assume that $H(y) > 0$. Conditioned on the bound on $\rank_{H(y)-1}(y)$ in Lemma~\ref{lem:decay rank},
it holds that $\rank_{H(y)}(y) = 0$.
\end{lemma}
\begin{proof}
According to Lemma~\ref{lem:decay rank} and the definition of $H(y)$
as the minimal $h$ for which $2^{2-h} \rank(y) \leq B/2$, 
$$ \rank_{H(y)-1}(y) \leq 2^{2-H(y)} \rank(y) \leq \frac12 B \,. $$
Invoking Observation~\ref{obs:decay rank 1}, we get
$ \rank_{H(y)}(y) \leq \max\{0, \rank_{H(y)-1}(y) - B/2\} = 0$.
\end{proof}

We are now ready to bound the overall error of the sketch for item $y$, i.e.,
$\err(y) = \hat{\rank}(y) - \rank(y)$ where $\hat{\rank}(y)$ is the estimated rank of $y$.
It is easy to see that 
$$ \err(y) = \sum_{h=0}^H 2^h \err_h(y),$$
where $H$ is the highest level with a relative-compactor (that never produces any output). To bound this error, we refine the guarantee of Lemma~\ref{lem:single layer}. Notice that 
for any particular relative-compactor,
the bound $ \sum_{t=1}^n |\Delta^t_h(y)|$ referred to in Lemma~\ref{lem:single layer} applied to a level $h$ is a potentially crude upper bound on $\err_h(y)=\sum_{t=1}^n \Delta^t_h(y)$:
Each non-zero term $\Delta^t_h(y)$ is positive or negative with equal probability, so the terms are likely to involve a large amount of cancellation. To take advantage of this, we bound the variance of $\err(y)$.

\begin{lemma}\label{lem:var}
Conditioned on the bound on $\rank_h(y)$ in Lemma~\ref{lem:decay rank} for any $h < H(y)$,
$\err(y)$ is a zero-mean sub-Gaussian random variable with $\V[\err(y)]\leq 2^5\cdot \rank(y)^2 / (k\cdot B)$.
\end{lemma}

\begin{proof}
Consider the relative-compactor at any level $h$.
By Lemma~\ref{lem:single layer}, $\err_h(y)$ is a sum of at most $\rank_h(y)/k$ random variables, i.i.d.\ uniform in $\{-1,1\}$. In particular, $\err_h(y)$ is a zero-mean sub-Gaussian random variable with $\V[\err_h(y)]\leq \rank_h(y)/k$.
Thus, $\err(y)$ is a sum of independent zero-mean sub-Gaussian random variables, and as such is itself a zero-mean sub-Gaussian random variable.

It remains to bound the variance of $\err(y)$, for which we first bound $\V[\err_h(y)]$ for each $h$.
If $\rank_h(y) = 0$, then Observation \ref{obs:even y} implies that $\err_h(y)=0$, and hence that $\V[\err_h(y)]=0$.
Thus, using Lemma~\ref{lem:ranky0}, we have $\V[\err_h(y)]=0$ for any $h \geq H(y)$. 
For $h < H(y)$, we use $\V[\err_h(y)] \leq \rank_h(y)/k$ to obtain:
\begin{align*}
\V[\err(y)]
&= \sum_{h=0}^{H(y)-1} 2^{2h} \V[\err_h(y)]
\\
&\leq \sum_{h=0}^{H(y)-1} 2^{2h} \cdot \frac{\rank_h(y)}{k}
\leq \sum_{h=0}^{H(y)-1} 2^{h+1} \cdot \frac{\rank(y)}{k}
	\leq 2^{H(y) + 1}\cdot \frac{\rank(y)}{k}
	\leq 2^5\cdot \frac{\rank(y)^2}{k\cdot B}\,,
\end{align*}
where the second inequality is due to Lemma~\ref{lem:decay rank} and the last inequality follows from~\eqref{eq:2to4}.
\end{proof}

To show that the space bound is maintained, we also need to bound the number of relative-compactors.

\begin{obs} \label{obs:H bound}
The number of relative-compactors ever created by the full algorithm (Algorithm \ref{code:full}) is at most $\lceil \log_2(n/B) \rceil + 1$.
\end{obs}

\begin{proof}
Each item on level $h$ has weight $2^h$, so there are at most $n /
2^h$ items inserted to the buffer at that level.
Applying this observation to $h = \lceil \log_2(n/B) \rceil$, we get that on
this level, there are fewer than $B$ items
inserted to the buffer, which is consequently not compacted, so the highest level has index at most $\lceil \log_2(n/B) \rceil$.
The claim follows (recall that the lowest level has index 0).
\end{proof}

We are now ready to prove the main result of this section, namely, the accuracy guarantees in the streaming setting
when the stream length is essentially known in advance.

\begin{theorem} \label{thm:sketch}
Assume that (a polynomial upper bound on) the stream length $n$ is known in advance.
For any parameters $0 < \delta \le 0.5$ and $0 < \eps \le 1$ satisfying $\delta > 1/\exp(\eps n/64)$,
there is a randomized, comparison-based, one-pass streaming algorithm that,
when processing a data stream consisting
of $n$ items from a totally-ordered universe $\mathcal{U}$,  produces a summary $S$ satisfying the following property. 
Given $S$, for any $y \in \mathcal{U}$ one can derive an estimate
$\hat{\rank}(y)$ of $\rank(y)$ such that
$$ \Pr\bigg[ |\hat{\rank}(y) - \rank(y)| > \eps \rank(y) \bigg] < \delta\,,$$
where the probability is over the internal randomness of the streaming algorithm.
The size of $S$ in memory words is
$$O\left(\eps^{-1}\cdot \log^{1.5}(\eps n)\cdot \sqrt{\log\frac1\delta}\right)\,.$$
\end{theorem}

\begin{proof}
First, suppose that  $\eps \le 4\cdot\sqrt{\ln(1/\delta) / \log_2(\eps n)}$.
Then we use Algorithm~\ref{code:full} with parameters $k$ and $n$, where $k$ is set as in~\eqref{eq:setk}.
Note that $k$ is an even positive integer as required by Algorithm~\ref{code:full}.
By Lemma~\ref{lem:decay rank}, with probability at least $1-\delta$, 
we have $\rank_h(y) \leq 2^{-h+1}\rank(y)$ for any $h < H(y)$
and we condition on this event happening.

We again apply the standard (Chernoff) tail bound for sub-Gaussian variables (Fact~\ref{fact:tailBound-subGaussian})
together with Lemma~\ref{lem:var} (for which we need the bound on $\rank_h(y)$ for any $h < H(y)$) and obtain
\begin{align*}
\Pr\left[ |\err(y)| \geq \eps \rank(y) \right] 
&< 2 \exp\left( -\frac{\eps^2\cdot \rank(y)^2}{2\cdot 2^5\cdot \rank(y)^2/(k\cdot B)} \right) 
\\
&\le 2 \exp\left( -\frac{\eps^2\cdot 2^6\cdot \eps^{-2}\cdot \ln \frac{1}{\delta}}{2^6} \right)
= 2 \exp\left( -\ln \frac{1}{\delta} \right)
= 2\delta\,,
\end{align*}
where we use Claim~\ref{clm:kBlowerBound} in the second inequality.
This concludes the calculation of the failure probability (up to scaling $\delta$ by a factor of $1/3$).

Regarding the memory usage, there are at most $\lceil\log_2(n/B)\rceil + 1 \le \log_2(\eps n)$  relative-compactors by Observation~\ref{obs:H bound}, and
each requires $B = 2\cdot k\cdot \lceil\log_2(n/k)\rceil$ memory words.
Thus, the memory needed to run the algorithm is at most
\begin{align}\label{eqn:spaceBound}
\log_2(\eps n)\cdot 2\cdot k\cdot \left\lceil\log_2 \frac{n}{k}\right\rceil
\le \log_2(\eps n)\cdot 2\cdot 2\cdot \left\lceil \frac{4}{\eps}\cdot \sqrt{\frac{\ln \frac{1}{\delta}}{\log_2(\eps n)}}\right\rceil \cdot O\left(\log(\eps n)\right)\,,
\end{align}
where we use that $\lceil\log_2(n/k)\rceil\le O\left(\log(\eps n)\right)$, which follows from $k\ge \eps^{-1} / \sqrt{\log_2(\eps n)}$.
Using $\eps \le 4\cdot\sqrt{\ln(1/\delta) / \log_2(\eps n)}$, we have $a := 4\eps^{-1}\cdot \sqrt{\ln(1/\delta) / \log_2(\eps n)} \ge 1$, so $\lceil a\rceil \le 2a$ and it follows that~\eqref{eqn:spaceBound}
is bounded by $O\left(\eps^{-1}\cdot \log^{1.5}(\eps n)\cdot \sqrt{\log(1/\delta)}\right)$.

For $\eps > 4\cdot\sqrt{\ln(1/\delta) / \log_2(\eps n)}$, we use the comparison-based streaming algorithm by Zhang et al.~\cite{zhang2006space} that requires space  $O\left(\eps^{-2}\cdot \log(\eps^2 n)\cdot \log(1/\delta)\right)$ and otherwise satisfies the same error guarantee as our algorithm.
To get the desired space bound, we observe that the case condition 
implies $\sqrt{\log_2(\eps n)} > 4\cdot\sqrt{\ln(1/\delta)}\cdot \eps^{-1}$ and thus,
 $O\left(\eps^{-2}\cdot \log(\eps^2 n)\cdot \log(1/\delta)\right) \le O\left(\eps^{-1}\cdot \log^{1.5}(\eps n)\cdot \sqrt{\log(1/\delta)}\right)$.\footnote{
	In fact, as we show in Section~\ref{s:mergeability}, one may use a variant of our algorithm also for the case of large $\eps$, that is, when $\eps > 4\cdot\sqrt{\ln(1/\delta) / \log_2(\eps n)}$.
	Namely, we compute the largest value of $\overline{n}$ such that 
	$1 < k = 2\cdot \left\lceil (4/\eps)\cdot \sqrt{\ln(1/\delta) / \log_2(\eps \overline{n})}\right\rceil$ (for given $\eps$ and $\delta$);
	cf.~\eqref{eqn:lambda} in Section~\ref{s:mergeability}.
	If $n > \overline{n}$, then using buffers of size $\Theta(\log \eps \overline{n})$ is sufficient and we
	do not need to use the compaction schedule (intuitively, the section size $k$ is too small to be useful).
	In this section, we omit these details for brevity and focus just on the main case of relatively small $\eps$.
} We remark that the algorithm from~\cite{zhang2006space} does not require any foreknowledge of the total input length $n$.
\end{proof}

\label{s:analysis}

\paragraph{Update time.}
\label{s:updatetime}
We now analyze the amortized update time of Algorithm~\ref{code:full} and show that it can be made 
$O(\log B) = O(\log(k) + \log\log(n/k))$, i.e., the algorithm processes $n$ streaming updates in total time $O(n\cdot \log B)$.
To see this, first observe that the time complexity is dominated, up to a constant factor, by running Algorithm~\ref{code:single} for the relative-compactor
at level $0$. Indeed, the running time can be decomposed into the operations done by Algorithm~\ref{code:full} itself, plus
the running time of Algorithm~\ref{code:single} for each level of the sketch, and the former is bounded by the latter.
Moreover, at level $h$ there are at most $n/2^h$ items added to the buffer, implying that the running time 
of Algorithm~\ref{code:single} decreases exponentially with the level.
At level $0$, the update time is $O(1)$, except for performing compaction operations (lines~\ref{line:S1}-\ref{line:compaction_end}
of Algorithm~\ref{code:single}).
To make those faster, we maintain the buffer sorted after each insertion, which can be achieved by using
an appropriate data structure in time $O(\log B)$ per update.
Then the time to execute each compaction operation is linear in the number of items removed from the buffer,
making it amortized constant.
Hence, the amortized update time with such adjustments is $O(\log B)$.

\section{Handling Unknown Stream Lengths}
\label{s:unknownlength}

The algorithm of Section \ref{s:fullsketch} and analysis in Sections~\ref{s:relativecompactor}-\ref{s:fullanalysis}
proved Theorem \ref{thm:sketch} in the streaming setting assuming that (an upper bound on) $n$ is known, where $n$ is the true stream length. The space usage of the algorithm grows polynomially with the logarithm of this upper bound, so if this upper bound is at most $n^c$ for some constant $c \ge 1$, then the space usage of the algorithm will remain as stated in Theorem \ref{thm:sketch}, with only the hidden constant factor changing. 

In the case that such a polynomial upper bound on $n$ is not known, we
modify the algorithm slightly, and start with an initial estimate $N_0$
of $n$, namely, $N_0 = \Theta(\eps^{-1})$. That is, we begin by running Algorithm~\ref{code:full} with parameters $k$ and $N_0$.
As soon as the stream length hits the current estimate $N_i$,
the algorithm ``closes out'' the current data structure and continues to store it in ``read only'' mode,
while initializing a 
new summary based on the estimated stream length of $N_{i+1} = N_i^2$ (i.e., we execute Algorithm~\ref{code:full} with parameters $k$ and $N_{i+1}$; 
only if $\eps > 4\cdot\sqrt{\ln(1/\delta) / \log_2(\eps N_{i+1})}$ we switch to the algorithm from~\cite{zhang2006space} as in the proof of Theorem \ref{thm:sketch}).\footnote{
In a practical implementation, we suggest not to close out the current summary, but rather
recompute the parameters $k$ and $B$ of every relative-compactor in the summary, according
to the new estimate $N_{i+1}$, and continue with using the summary. 
The analysis in Section~\ref{s:mergeability} (which applies in the more general mergeability setting) shows that the same accuracy guarantees as in Theorem \ref{thm:sketch} hold for this
variant of our algorithm. 
Here, we choose to have one summary for each estimate of $n$ because it is amenable to a much simpler analysis (it is not clear how to extend this simpler analysis from
the streaming setting to the
general mergeability setting of Section~\ref{s:mergeability}).}
This process occurs at most $\log_2 \log_2 (\eps n)$ many times, before the guess is at least the true stream length $n$.
At the end of the stream, the rank of any item $y$ is estimated by summing the estimates returned
by each of the at most $\log_2 \log_2 (\eps n)$ summaries stored by the algorithm.

To prove a variant of Theorem \ref{thm:sketch} for unknown stream lengths, we need to
bound the space usage of the algorithm,
and the probability of having a too large error for a fixed item $y$.
We start with some notation.
Let $\ell$ be the biggest index $i$ of estimate $N_i$ used by the algorithm; note that $\ell \leq \log_2 \log_2(\eps n)$.
Let $\sigma_i$ denote the substream processed by the summary with the $i$-th guess for the stream length for $i=0, \dots \ell$.
Let $\sigma' \circ \sigma''$ denote the concatenation of two streams $\sigma'$ and $\sigma''$.
Then the complete stream processed by the algorithm is $\sigma = \sigma_0 \circ \sigma_1 \circ \dots \circ \sigma_\ell$.
Let $k_i$ and $B_i$ be the values of parameters $k$ and $B$ computed for estimate $N_i$.

\paragraph{Space bound.}
We claim that the sizes of summaries for the substreams $\sigma_0, \sigma_1, \dots, \sigma_\ell$ sum up 
to $O\left( \eps^{-1}\cdot \log^{1.5} (\eps n)\cdot \sqrt{\log(1/\delta)}\right)$,
as required.
By Theorem~\ref{thm:sketch}, the size of the summary for $\sigma_i$ is 
$O\left( \eps^{-1}\cdot \log^{1.5} (\eps N_i)\cdot \sqrt{\log(1/\delta)}\right)$.
In the special case $\ell=0$, the size of the summary for $\sigma_0$ satisfies the bound
provided that $N_0 = O(\eps^{-1})$. 
For $\ell \ge 1$, since $N_{\ell-1} < n$ and $N_\ell = N_{\ell-1}^2$, it holds that $N_\ell \le n^2$ and
thus, the size of the summary for $\sigma_\ell$ satisfies the claimed bound.
As $N_{i+1} = N_i^2$, the $\log^{1.5}(\eps N_i)$ factor
in the size bound from Theorem~\ref{thm:sketch} increases by a factor of $2^{1.5}$ when we increase $i$.
It follows that 
the total space usage is dominated, up to a constant factor, by the size of the summary for $\sigma_\ell$. \qed

\paragraph{Failure probability.}
We need to show that $|\err(y)| = |\hat{\rank}(y)-\rank(y) | \le \eps\rank(y) $ with probability at least $1-\delta$
for any fixed item $y$. Note that $\rank(y) = \rank(y; \sigma) = \sum_{i=0}^\ell \rank(y; \sigma_i)$.

We apply the analysis in Section~\ref{s:fullanalysis} to all of the summaries at once.
Observe that for the tail bound in the proof of Theorem~\ref{thm:sketch},
we need to show that $\err(y)$ is a zero-mean sub-Gaussian random variable with a suitably bounded variance.
Let $\err^i(y)$ be the error introduced by the summary for $\sigma_i$.
By Lemma~\ref{lem:var}, $\err^i(y)$ is a zero-mean sub-Gaussian random variable with
$\V[\err^i(y)]\leq 2^5\cdot \rank(y; \sigma_i)^2 / (k_i\cdot B_i)$.
As $\err(y) = \sum_i \err^i(y)$ and as the summaries are created with independent randomness,
variable $\err(y)$ is also zero-mean sub-Gaussian and its variance is bounded by
$$\V[\err(y)]
= \sum_{i=0}^\ell \V[\err^i(y)]
\le \sum_{i=0}^\ell 2^5\cdot \frac{\rank(y; \sigma_i)^2}{k_i\cdot B_i}
\le \frac{\eps^2\cdot \rank(y)^2}{2\cdot \ln(1/\delta)}\,$$
where the last inequality uses that $\sum_{i=0}^\ell \rank(y; \sigma_i)^2 \le \rank(y)^2$, 
which follows from $\rank(y) = \sum_{i=0}^\ell \rank(y; \sigma_i)$, and
that $k_i\cdot B_i = \Omega(\eps^{-2}\cdot \ln(1/\delta))$, which holds by Claim~\ref{clm:kBlowerBound}.
Applying the tail bound for sub-Gaussian variables similarly as in the proof of Theorem~\ref{thm:sketch}
concludes the proof of (a variant of) Theorem \ref{thm:sketch} for unknown stream lengths. \qed

\section{Full Mergeability}
\label{s:mergeability}

Fully-mergeable sketches allow us to sketch many different streams (or any inputs)
and then merge the resulting sketches (via an arbitrary sequence of pairwise merge operations) to get an accurate summary
of the concatenation of the streams. Mergeable sketches form an essential primitive
for parallel and distributed processing of massive data sets.
We show that our sketch maintains its accuracy guarantees even in these settings, and therefore, it is fully mergeable.

The merge operation takes as input two sketches $S'$ and $S''$
that processed two separate streams $\sigma'$ and $\sigma''$ and outputs a sketch $S$ that summarizes
the concatenated stream $\sigma = \sigma'\circ \sigma''$ (the order of $\sigma'$ and $\sigma''$ does not matter here).
For full mergeability, $S$ must satisfy the same space and accuracy guarantees
as if it was created by processing stream $\sigma$ in one pass.
Moreover, we do not assume that we built $S'$ by processing stream $\sigma'$ directly and similarly for $S''$,
but we allow to create $S'$ and $S''$ using merge operations.
Thus, we may create the resulting summary
from many summaries by merging them in an arbitrary way
(i.e., using an arbitrary merge tree).

We stress that we do not assume any advance knowledge about $n$, the total size of all the inputs merged,
which indeed may not be available in many applications.

\subsection{Merge Operation}\label{s:mergeOperation}

In this section, we describe the merge operation of our sketch, without assuming a foreknowledge of the total input size $n$.
The description builds on Section~\ref{s:mergeOperation-simplified}, which outlines a simplified merge procedure
under the assumption that a polynomial upper bound on $n$ is available.
To facilitate the merge operation, each sketch maintains list \textsc{RelCompactors}
of its relative-compactors and the following parameters:
\begin{description}[nosep]
\item{$H = $} index of the highest level with a relative-compactor in the sketch.
\item{$n = $} size of the input currently summarized by the sketch.
\item{$N = $} an upper bound on $n$, based on which the subsequent parameters $k$ and $B$ (defined below) are calculated.
\item{$\hat{k} = $} a parameter that depends on the desired accuracy $\eps$ and failure probability $\delta$, namely, 
    $\hat{k} = 4\eps^{-1}\cdot \sqrt{\ln(1/\delta)}$.

	Unlike $N$, the parameter $\hat{k}$ remains constant during the computation.
	The section size parameter $k$ (defined below) depends on $\hat{k}$ in addition to $N$.
\item{$k = $} size of a buffer section.
\item{$B = $} size of the buffer at each level.
\end{description}

\paragraph{Parameters $N, k,$ and $B$.}
The parameter $N$ is set similarly as in
Section~\ref{s:unknownlength}, that is,
it is equal to $N_i$ for some $i$, where $N_0 = \lceil 2^{10}\cdot \hat{k}\rceil$ and $N_{i+1} = N_i^2$.
We set the parameters $k$ and $B$ based on $N$ similarly as in Section~\ref{s:analysis} (cf.\ Equation~\eqref{eq:setk})
so that $k$ decreases and $B$ increases as we increase $N$.
Importantly, we no longer change $k$ and $B$ once $\hat{k}\le \sqrt{\log_2(N_i / \hat{k})}$.
To facilitate this, we define $\lambda \ge 0$ as the smallest integer $i$ such that
\begin{equation}\label{eqn:lambda}
\frac{\hat{k}}{\sqrt{\log_2(N_i / \hat{k})}} \le 1\,,
\end{equation}
and then for $i\ge 0$ we set
\begin{equation}\label{eqn:setk-merge}
k_i := 2^5\cdot \left\lceil \frac{\hat{k}}{\sqrt{\log_2(\overline{N_i} / \hat{k})}}\right\rceil
\quad\text{and}\quad
B_i := 2\cdot k_i\cdot \left\lceil \log_2 \left(\frac{\overline{N_i}}{k_i}\right)\right\rceil
\quad\text{where~} \overline{N_i} = \min\{N_i, N_\lambda\}\,.
\end{equation}
From a practical point of view, since $N_\lambda$ is about $\hat{k}\cdot 2^{\hat{k}^2}$, we have that $\overline{N_i} = N_i$ unless $N_i$ is extremely large or $\hat{k} = 4\eps^{-1}\cdot \sqrt{\ln(1/\delta)}$ is small 
(say, even for $\epsilon = 0.2$ we have $N_\lambda \gg 2^{400}$).
We use this truncation of $N_i$ to guarantee the space bound when $n > N_\lambda$.
Furthermore, observe that once we reach $n\ge N_\lambda$, the values of $k_i$ and $B_i$ do not change;
this is because, intuitively, the section size $k_i$ becomes too small to help in the analysis and our algorithm
can in fact be simplified by involving all sections in every compaction, without violating the error guarantees (i.e., when $n \ge N_\lambda$, the compaction schedule is no longer relevant).
The most challenging part of the analysis is bounding the error for $i\le \lambda$.

\paragraph{Description of the merge operation.}
The merge operation that creates sketch $S$ from $S'$ and $S''$ goes as follows:
Suppose that both $S'$ and $S''$ are based on the same parameter $\hat{k}$ 
and that $S'$ has at least as many levels as $S''$ (otherwise, we swap the sketches).
Then, via the following procedure, we merge $S''$ into $S'$, so $S''$ acts as a source sketch, while $S'$ is a target sketch of the merge operation.
First, we compute the parameters of the resulting sketch. 
For sketch $S$ resulting from the merge operation, $S.n$ is just the sum of $S'.n$ and $S''.n$.
If $S'.N \ge S.n$, then we keep parameters $N, k,$ and $B$ as they are set in $S'$.
Otherwise, $S'.N < S.n = S'.n + S''.n$, so $S'.N$ would be too small after merging.
In this case, we choose the next upper bound by setting $S.N = S'.N^2$ and also recompute $k$ and $B$ as described in Equation \eqref{eqn:setk-merge} above.

Recall from Section~\ref{s:mergeOperation-simplified} that
the crucial part of the merge operation is to combine the states of the compaction schedules at each level
without violating the relative-error guarantees even when many merge operations are executed.\footnote{By the state
of the compaction schedule, we mean the variable that determines how many sections of the buffer to include in a compaction operation
if one is performed.
In the streaming setting (Algorithm \ref{code:single}), we denoted this variable by $\schedulestate$, and maintain this notation
in the mergeability setting.}
Consider a level $h$
and let $\schedulestate'$ and $\schedulestate''$ be the states of the compaction schedule at level $h$ in $S'$ and $S''$, respectively.
The new state $\schedulestate$ at level $h$ will be the bitwise OR of $\schedulestate'$ and $\schedulestate''$;
we explain the intuition behind using the bitwise OR below.
Note that while in the streaming setting, the state corresponds to the number of compaction operations already performed,
after a merge operation this may not hold anymore.
Still, if the state is zero, this indicates that the level-$h$ buffer has not
yet been subject to any compactions. 

Having set up the parameters and states at each level,
we concatenate the level-$h$ buffers of $S'$ and of $S''$ at each level that appears in both of them.
Then we perform a single compaction operation at each level that has at least $S.B$ items,
in the bottom-up fashion. For such a compaction operation, all but the smallest $S.B$ items
in the buffer are automatically included in the compaction, while the
smallest $B$ items are treated exactly as a full buffer is treated in
the streaming setting to determine what suffix is compacted.
That is, the state variable $\schedulestate$ of the compaction schedule determines how many sections among the smallest $B$ items in the buffer are compacted,
via the number of trailing 1s in the binary representation of $\schedulestate$. If this number of trailing 1s is $j\ge 0$, 
then $j+1$ sections are compacted and we say that the compaction \emph{involves exactly $j+1$ sections of the buffer}.
Note that there is at most one compaction per level during the merge operation.
Finally, when $N_i> N_\lambda$, we do not use the compaction schedule as the section size becomes too small, i.e.,
we compact all buffer sections.

Algorithm~\ref{code:merge} provides pseudocode describing the merge operation specified above.
We note that inserting a single item $x$ can be viewed as a trivial
merge with a summary consisting just of $x$ (with weight $1$).

\begin{algorithm}[t]
	\caption{Merge operation of ReqSketch}\label{code:merge}
\begin{algorithmic}[1]{\small
	\Require Sketches $S'$ and $S''$ to be merged such that $S'.\hat{k} = S''.\hat{k}$ and $S'.H \ge S''.H$
	\Ensure A sketch answering rank queries for the combined inputs of $S'$ and $S''$

	\Comment{We merge $S''$ into $S'$}
	\State Set $S'.n = S'.n + S''.n$ \Comment{Combined input size}

	\If{$S'.N < S'.n$} \Comment{Upper bound on input size is too small} \label{li:growsize}

		\State Set $S'.N =  S'.N^2$ \Comment{Square the upper bound}
		\State Set $S'.k$ and $S'.B$ according to~\eqref{eqn:setk-merge}
	\EndIf
	\For{$h = 0, \ldots, S''.H$} \Comment{Combine buffers and states of compaction schedules}

		\State Insert all items in $S''$.\textsc{RelCompactors}[$h$] into $S'$.\textsc{RelCompactors}[$h$]
		\State $S'$.\textsc{RelCompactors}[$h$].$\schedulestate$ = $S'$.\textsc{RelCompactors}[$h$].$\schedulestate$ \textsf{OR} $S''$.\textsc{RelCompactors}[$h$].$\schedulestate$
		\label{algMerge-ln:combineStates}
	\EndFor
	\For{$h = 0, \ldots, S'.H $}
		\If{there are at least $S'.B$ items in $S'$.\textsc{RelCompactors}[$h$]}

			\State \Call{PerformCompaction}{$S', h$}
		\EndIf
	\EndFor
	\State \Return $S'$
	\medskip
	\Function{PerformCompaction}{$S', h, s$}
		\If{$h = S'.H$}
			\State Increase $S'.H$ by one
			\State Initialize relative-compactor at \textsc{RelCompactors}[$h+1$]
		\EndIf
		\State Set $\calB = $ $S'$.\textsc{RelCompactors}[$h$] \Comment{The level-$h$ buffer of $S'$}
		\State Sort items in $\calB$ in non-descending order
		\If{$S'.N \le N_\lambda$} \Comment{$\lambda$ is defined in~\eqref{eqn:lambda}}
			\State Compute $z = $ number of trailing 1s in binary representation of $\calB.\schedulestate$
			\State Set $s = S'.B - (z + 1)\cdot S'.k + 1$ \Comment{First slot of the buffer involved in the compaction}
			\label{ln-merge:firstSlotCompaction}
		\Else
			\Comment{Then $S'.k = \Theta(1)$}
			\State Set $s = S'.B / 2$
			\Comment{Compaction schedule not used when $S'.k$ is small}
		\EndIf
		\State Let $|\calB|$ be the number of items stored in $\calB$
			\Comment{$|\calB|$ may be larger than $S'.B$}
		\State Set $Z$ = equally likely either even or odd indexed items in the range $\calB[s:|\calB|]$
			\State \Comment{Note that the range $\calB[s:|\calB|]$ may be of an odd size, which does not cause any issues}
		\State Insert each item in $Z$ to $S'$.\textsc{RelCompactors}[$h + 1$]
		\State Mark slots $\calB[s:|\calB|]$ in the buffer as clear
		\State Increase $\calB.\schedulestate$ by $1$
	\EndFunction }
\end{algorithmic}
\end{algorithm}

Several remarks and observations are in order.
First, the combined buffer contains at most $2\cdot S.B$ items before the merge procedure begins performing compactions level-by-level, because
each buffer of $S'$ and each buffer of $S''$ stores at most $S.B$ items. Second,
when we perform a compaction on a level-$h$ buffer during the merge procedure, it contains no more than $\frac72 \cdot S.B$ items.
To see this, observe that there are three sources of input to the buffer at level $h$ during a merge operation: the at most $S.B$ items in $S'$ at level $h$ at the start of the merge operation,
the at most $S.B$ items in $S''$ at level $h$ at the start of the merge operation, and the output of the level-$(h-1)$ buffer during the merge procedure. 
An easy inductive argument shows
that the third source of inputs consists of at most $\frac32 \cdot
S.B$ items, as follows: 
Observe that if the level-$(h-1)$
  buffer has size at most $\frac72 S.B$ when it is compacted, then the number of items compacted 
  by that buffer is at most $\frac72 S.B - \frac{1}{2} S.B = 3 S.B$, and hence, the
  number of items output by the compaction is at most $\frac32 \cdot S.B$
  (here, we also use that $S.B$ as defined in~\eqref{eqn:setk-merge} is divisible by four, so $\frac32 \cdot S.B$ is even).
  This guarantees that at the time a level-$h$ buffer is actually compacted during a merge procedure, it contains no more than $\frac72 \cdot S.B$ items.

Third, using the bitwise OR in line~\ref{algMerge-ln:combineStates} to combine the states has two simple but important implications.
\begin{fact} \label{factstar}
When the $j$-th bit of $\schedulestate'$ or of $\schedulestate''$ is set to 1, then the $j$-th bit of
$\schedulestate = \schedulestate'$ \textsf{OR} $\schedulestate''$ is also set to~1.
\end{fact}
	
\begin{fact} The bitwise OR of $\schedulestate'$ and $\schedulestate''$ (interpreted
  as bitstrings) is no larger than $\schedulestate' + \schedulestate''$ (interpreted
  as integers). \label{factstarstar}
\end{fact}
Fact \ref{factstarstar} will be used later to show that the state $\schedulestate$ never has more than $\lceil\log_2 (S.N / S.k)\rceil$ bits,
so we never compact more than $\lceil\log_2 (S.N / S.k)\rceil$ buffer sections during a compaction.
See Observation \ref{obs:numberOfBits} for details.
(Note that this is only relevant for $S.N\le N_\lambda$.)

\subsection{Preliminaries for the Analysis of the Merge Procedure}

Consider a sketch $S$ built using an arbitrary sequence of merge operations from an input of size $n$.
We will show that the space bound holds for $S$ using an argument similar to the one in the proof of Theorem~\ref{thm:sketch},
but the calculation of the failure probability needs to be modified
compared to Section~\ref{s:fullanalysis}.
The main challenge is that the parameters $k$ and $B$ change as more and more merge operations are performed.

To prove that the accuracy guarantees hold for $S$,
consider the binary tree $T$ in which each of $n$ leaves corresponds to a single item of the input.
Internal nodes correspond to merge operations (recall that inserting one item to the sketch can be seen as the merge of the sketch with a trivial sketch storing the item to be inserted),
and hence
each internal node $t$ in $T$ represents a sketch $S_t$
resulting from the merge operation that corresponds to node $t$.
Also, for a particular level $h$, node $t$ represents the level-$h$ buffer of $S_t$.
Finally, we say that $t$ represents the level-$h$ compaction operation (if any);
recall that the merge operation captured by an internal node $t$ performs at most one compaction operation at each level $h$.
The root of $T$ represents the final merge operation, which outputs the \emph{final sketch}.

Recall that we set the upper bounds $N$ on the input size used by the sketches
as $N_0 = \lceil 2^{10}\cdot \hat{k}\rceil$
and 
$N_i = N_{i-1}^2$ for $1 \leq i \leq \ell \le \lceil \log_2 \log_2(\eps n) \rceil$
(as $N_0\ge \hat{k}\ge 1/\eps$).
We assume that $\ell > 0$, otherwise the whole input can be stored in space $O(\hat{k}) = O(\eps^{-1}\cdot \sqrt{\log(1/\delta)})$.

We say that an (internal) node $t$ in tree $T$ is an \emph{$i$-node} for $0\le i\le \ell$
if the sketch $S_t$ represented by $t$ satisfies $S_t.N = N_i$, i.e.,
it uses the parameters $k_i$ and $B_i$.
Note that this means that if parameter $N$ is updated from $N_{i-1}$ to $N_i$ during the merge operation represented by $t$, then $t$ is considered an $i$-node.
Moreover, we say that node $t$ is a \emph{topmost $i$-node} if the parent of $t$ is a $j$-node for some $j>i$ or $t$ is the root of $T$.
Note that for any $i$, the subtrees of topmost $i$-nodes are disjoint.

As in Sections~\ref{s:relativecompactor} and~\ref{s:fullanalysis},
we consider a fixed item $y$ and analyze the error of the estimated rank of~$y$.
Let $\rank(y)$ denote the rank of $y$ in the input summarized by the sketch,
and let $\hat{\rank}(y)$ be the estimated rank of $y$ obtained from the final sketch $S$;
recall that we get this estimate by summing over all levels $h\ge 0$ the number of items $x\le y$ in the level-$h$ buffer of the final sketch, multiplied by $2^h$.
Our aim is to show that $|\err(y)| = |\hat{\rank}(y) - \rank(y)| \le \eps \rank(y)$ with probability at least $1 - \delta$.

\subsection{Analysis of a Single Level for Mergeability}
\label{s:mergeability-singleLevel}

For the duration of this section, we consider a single level $h$ and solely focus on $i$-nodes for $i\le \lambda$;
recall that the compaction schedule helps to decrease the error from compactions
and that we do not use the schedule during compactions represented by $i$-nodes for $i>\lambda$ (since the buffer section 
size is too small to make a difference).
For convenience, $\lambda$ refers to $\min\{\lambda, \ell\}$, i.e., if $\lambda > \ell$ we decrease $\lambda$ to $\ell$ compared to~\eqref{eqn:lambda}.
This is to ensure that, e.g., topmost $\lambda$-nodes are well-defined.
Note that when $\lambda = \ell$, then the only topmost $\lambda$-node is the root of the merge tree $T$.

We start by showing that the binary representation of the state $\schedulestate$
at level $h$ never has more than $\lceil\log_2 (S.N / S.k)\rceil$ bits, or equivalently, $\schedulestate \le S.N / S.k$.
Consequently, $\schedulestate$ (viewed as a bitstring) never has $\lceil\log_2 (S.N / S.k)\rceil$ trailing ones
just before a compaction operation (as after the operation, it would have more than $\lceil\log_2 (S.N / S.k)\rceil$ bits).

\begin{obs}\label{obs:numberOfBits}
Consider a node $t$ of tree $T$ and sketch $S$ represented by $t$.
Let $\schedulestate$ be the state of the level-$h$ buffer of $S$.
Then $\schedulestate\le S.N / S.k$.
\end{obs}

\begin{proof}
Let $r$ be the number of items removed from the level-$h$ buffer of $S$ during all compactions represented by nodes
in the subtree of $t$.
We show that $\schedulestate\le r / S.k$ by induction. This implies $\schedulestate\le S.N / S.k$ as $r \le S.n\le S.N$.

The base case of a leaf node follows as $\schedulestate = 0$ and $r = 0$.
Let $S$ be the sketch represented by an internal node and let $S'$ and $S''$ be the sketches represented by its children. Let $\schedulestate'$ and $\schedulestate''$ be the states
of the level-$h$ buffers of $S'$ and $S''$, and let $r'$ and $r''$ be the number of items removed from the level-$h$ buffer
during compactions represented by nodes in the subtrees of $S'$ and $S''$, respectively. By the induction hypothesis, we have $\schedulestate'\le r' / S'.k$ and $\schedulestate''\le r'' / S''.k$.
Note that $r$ equals $r' + r''$ plus the number of items removed from
the level-$h$ buffer during the compaction represented by $t$ if there
is one.
Let $b\in \{0, 1\}$ be the indicator variable
with $b = 1$ iff there is a level-$h$ compaction represented by $t$.
Observe that $\schedulestate = (\schedulestate'\,\text{OR}\,\schedulestate'') + b$ and 
if $b = 1$, then the compaction removes at least $S.k$ items from the level-$h$ buffer.
We thus have $r \ge r' + r'' + b\cdot S.k$ and using this, we obtain
\begin{align*}
\schedulestate = (\schedulestate'\,\text{OR}\,\schedulestate'') + b 
\le \schedulestate' + \schedulestate'' + b 
\le \frac{r'}{S'.k} + \frac{r''}{S''.k} + b 
\le \frac{r'}{S.k} + \frac{r''}{S.k} + \frac{b\cdot S.k}{S.k}
\le \frac{r}{S.k}\,,
\end{align*}
where the penultimate inequality uses $S.k \le \min\{S'.k, S''.k\}$, which follows from $k_0 \ge k_1 \ge \cdots \ge k_\lambda$.
\end{proof}

For $i \le \lambda$, we recall that the second half of the buffer of size $B_i$ has $\lceil\log_2(N_i / k_i)\rceil$ sections of size $k_i$ (see Equation \eqref{eqn:setk-merge})
and that these sections are indexed from 1 such that the rightmost section
(with slots $B_i - k_i + 1, \dots, B_i$) is section~1
and section~$j$ consists of slots $B_i - j\cdot k_i + 1, \dots, B_i - (j-1)\cdot k_i$.
The definition of the compaction operation and Observation~\ref{obs:numberOfBits} imply that
section $\lceil\log_2(N_i / k_i)\rceil$ (i.e., the leftmost section of the second half of the buffer) is involved only in one compaction represented by an $i$-node on any leaf-to-root path in $T$.

\paragraph{Bounding the number of important compaction operations.}
As in Section~\ref{s:relativecompactor}, the key part of the analysis is bounding the number of level-$h$ compaction operations
that introduce some error for the fixed item $y$; recall that we call such compactions important and that by Observation~\ref{obs:even y},
a compaction is important if and only if it removes an odd number of important items from the buffer.
Also, recall that we call items $x\le y$ important and that for $h>0$, $\rank_h(y)$ denotes the total number of important items
promoted to level $h$ during compaction operations at level $h-1$ (represented by any node in $T$).
For level $0$, we have $\rank_0(y) = \rank(y)$.

The bound on the number of important  level-$h$ compactions in Lemma~\ref{lem:single layer mergeability - stronger} below is more involved than in the streaming setting (Section~\ref{s:relativecompactor}),
but this complexity allows for the tightest and most general analysis, presented in Section~\ref{s:k0sec}.
In particular, for any $0\le a\le \lambda$,
we will need a bound on the number of important level-$h$ compactions represented by $i$-nodes for $i\in [a,\lambda]$.

To state the bound, we first give a few definitions.
We say that a compaction \emph{involves important items} iff it removes at least one important item from the buffer; note that compactions
involving important items are a superset of important compactions. 
Let $Q_h$ be the set of nodes $t$ such that (i) $t$ is an $i$-node for $i\le \lambda$ that represents a level-$h$ compaction involving important items
(this compaction may or may not be important),
and (ii) there is no node $t'$ on the path from the parent of $t$ to the topmost $\lambda$-node containing $t$ in its subtree such that $t'$ represents a level-$h$ compaction involving important items.
Intuitively, $Q_h$ captures ``maximal'' nodes (disregarding $i$-nodes for $i > \lambda$, if any) that represent a level-$h$ compaction removing one or more important items from level $h$.
Note that an important item that remains in the level-$h$ buffer represented by a node $t\in Q_h$
(after performing the compaction operation represented by $t$) is never removed from the level-$h$ buffer during compactions represented by $i$-nodes for $i \le \lambda$, by the definition of $Q_h$.
For $i\in [0,\lambda]$, let $Q^i_h$ be the set of $i$-nodes in $Q_h$.

For some $0\le a\le \lambda$, let $\rank^{[a,\lambda]}_h(y)$ be the number of important items that are
either (i) removed from level $h$ during a compaction represented by an $i$-node for $i\in [a,\lambda]$,
or (ii) remain at the level-$h$ buffer of the sketch represented by a node $t\in Q^i_h$ for $i\in [a,\lambda]$
(after the compaction operation represented by $t$ is performed).
Note that important items in (ii) also belong to the level-$h$ buffer represented by a topmost $\lambda$-node since the level-$h$ buffer is not subject to a compaction that removes an important item and is represented by a node on the path from $t\in Q^i_h$ to its corresponding topmost $\lambda$-node, by the definition of $Q^i_h$.
We remark that the level-$h$ buffers represented by topmost $\lambda$-nodes may contain important items not present
in the level-$h$ buffers represented by nodes in $Q_h$ (these are items promoted from level $h-1$ to level $h$ during merge operations represented by nodes on the path from a node $t\in Q_h$ to a topmost $\lambda$-node).

We now state the bound on the number of important level-$h$ compactions represented by $i$-nodes
for $i\le \lambda$. Let $m^i_h$ be the number of important
compaction operations at level $h$ represented by $i$-nodes.

\begin{lemma}\label{lem:single layer mergeability - stronger}
For any level $h$ and any $0\le a\le \lambda$, it holds that
\begin{equation}\label{eqn:single layer mergeability - stronger}
\sum_{i = a}^{\lambda} m^i_h\cdot k_i \le 4\rank^{[a,\lambda]}_h(y)\,.
\end{equation}
\end{lemma}

\paragraph{Proof overview.}
The proof is an extension of the charging argument in Lemma~\ref{lem:single layer} to the mergeability setting.
In a nutshell, we will again charge each important compaction represented by an $i$-node for some  $a\le i\le \lambda$
to $k_i$ important items that are removed from the level-$h$ buffer (during a compaction represented by an $i'$-node for some  $i\le i'\le \lambda$) or that remain in the level-$h$ buffer represented by a node in $Q^{i'}_h$ for $i\le i'\le \lambda$.
However, unlike in the streaming setting, we will not identify specific important items to which we charge an important compaction.

Instead, for each node $t$ in the subtree of a node in $Q_h$, we will maintain the overall charge from $t$'s subtree to the (level-$h$) buffer represented by $t$.
Intuitively, when two buffers are merged during the merge procedure represented by an $i$-node $t$ for $a\le i\le \lambda$, the charge to the resulting buffer is the sum of the charges to the two buffers
increased or decreased by the following:
\begin{itemize}
	\item when the level-$h$ compaction represented by $i$-node $t$ (if any) is important, we increase the charge to the buffer by $k_i$,

	\item removing $r$ important items during the compaction operation (not necessarily important) decreases the charge to the buffer by $3r$, and

	\item if a child $t'$ of $t$ is a topmost $i'$-node for $i' < i$ such that there is an important compaction represented by an $i'$-node in the subtree of $t'$,
	we decrease by $2k_i$ the charge in the buffer represented by $t$ (not by $t'$).
\end{itemize}
The latter decrease helps us to deal with merge operations in which parameters $k$ and $B$ of the level-$h$ buffer change (in particular, $k_i$ decreases and therefore, we need to create a slack in the analysis).
We prove below that (i) the charge to any buffer is always bounded by the number of important items in the buffer
and that (ii) these properties imply~\eqref{eqn:single layer mergeability - stronger}, proving the lemma.
Showing (ii) is not difficult given (i); the only non-trivial part is bounding the total decrease of the charge 
from the third bullet above, which is done in the parents of topmost $i$-nodes.

Proving (i) relies on the compaction schedule. We in particular show that for each $i$-node $t$ 
either there is slack at $t$, i.e., the charge to $t$ is smaller by at least $k_i$ than the number of important items in the level-$h$ buffer represented by $t$,
or the schedule state $C$ guarantees that at least $k_i$ important items would be removed if a compaction is executed.\footnote{A somewhat simpler but weaker proof of the lemma appears in the previous version of this manuscript; see \url{https://arxiv.org/abs/2004.01668v3}. However, this earlier analysis required a modified (and slightly more involved) merge procedure.}

\begin{proof}[Proof of Lemma~\ref{lem:single layer mergeability - stronger}]
For simplicity, when we refer to a buffer or a compaction operation represented by a node we implicitly mean the one at level $h$.
For any node $t$ in the subtree of a node in $Q_h$,
we define its \emph{charge} $\chi(t)$ (implicitly w.r.t.\ item $y$ and level $h$) recursively as follows:
\begin{itemize}
\item If $t$ is a leaf node or an $i$-node for $i < a$, we set $\chi(t) = 0$.
\item Otherwise, let $t'$ and $t''$ be the children of $t$ and let $i\in [a, \lambda]$ be such that $t$ is an $i$-node.
To define $\chi(t)$, we need a few quantities and indicators:
\begin{description}
\item{$r(t) = $} the number of important items removed from the buffer during the compaction represented by $t$ (we use $r(t) = 0$ if there is no compaction operation represented by $t$);
\item{$I(t)$} is the indicator whether the compaction represented by $t$ (if any) is important, i.e., $I(t) = 1$ if there is an important compaction represented by $t$, and $I(t) = 0$ otherwise; and
\item{$J(t)$} is the indicator whether for a child $\hat{t}\in \{t', t''\}$ of $t$, it holds that $\hat{t}$ is a topmost $i'$-node for some $a\le i' < i$ and there is an important level-$h$ compaction represented by an $i'$-node in the subtree of $\hat{t}$.
\end{description}
Then, we define
\begin{equation}\label{eqn:charge_def}
\chi(t) = \max\{\chi(t') + \chi(t'') - 3r(t) + I(t)\cdot k_i - J(t)\cdot 2\cdot k_i, 0\}\,.
\end{equation}
\end{itemize}
(We do not define $\chi(t)$ for nodes that are not in the subtree of a node in $Q_h$.)
This recursive definition implies that
$\chi(t) > 0$ only if there is an important compaction represented by an $i'$-node (for $a\le i'\le i$ with $i'\le \lambda$)
in the subtree of $t$, including $t$ (the converse may not be true).
The key part is to prove that $\chi(t)$ as defined above is always bounded by the number of important items in the buffer represented by $t$.

\begin{claim}\label{clm:charging arg for mergeability}
For any node $t$ in the subtree of a node in $Q_h$, it holds that $\chi(t) \le \rank(y; \calB_h(t))$,
where $\calB_h(t)$ is the level-$h$ buffer represented by $t$
and $\rank(y; \calB_h(t))$ is the number of important items in that buffer.
\end{claim}

\begin{proof}
We start with some notation. Let $C_h(t)$ be the state of the compaction schedule of the level-$h$ buffer
represented by a node $t$, and for a state $C$ and $j\ge 1$, let $C[j]$ be the $j$-th bit from the right in the binary representation of $C$.

We prove by an induction over the tree $T$ a stronger claim: \textit{If $t$ is an $i$-node for $a\le i\le \lambda$ in the subtree of a node in $Q_h$,
then one of the following holds:
\begin{enumerate}[label=(\roman*)]
\item $\chi(t) \le \max\{\rank(y; \calB_h(t)) - k_i, 0\}$, or \label{chi-cond1}
\item there is an important level-$h$ compaction represented by an $i$-node in the subtree of $t$ and moreover, \label{chi-cond2}
letting $j\ge 1$ be the smallest index of a section which contains important items only,
$\chi(t) \le B_i - (j-1)\cdot k_i$ and, provided that $j > 1$, $C_h(t)[j-1] = 1$. 
\end{enumerate}
}
Note that both \ref{chi-cond1} and \ref{chi-cond2} are stronger requirements than $\chi(t) \le \rank(y; \calB_h(t))$; specifically,
in \ref{chi-cond2}, it holds that $\rank(y; \calB_h(t)) \ge B_i - (j-1)\cdot k_i$ by the definition of $j$ (recall that sections are indexed from the right).

The claim in \ref{chi-cond1} clearly holds if $\chi(t) = 0$ and thus, \ref{chi-cond1} holds for any leaf node or for any $i$-node $t$ for $i < a$ as we define $\chi(t) = 0$ in both of these cases.

Consider a non-leaf $i$-node $t$ with $a\le i\le \lambda$ and $\chi(t) > 0$, and let $t'$ and $t''$ be the children of $t$.
Note that $\rank(y; \calB_h(t))\ge \rank(y; \calB_h(t')) + \rank(y; \calB_h(t'')) - r(t)$;
on the RHS of this inequality, we do not take into account important items added from level $h-1$ during a compaction represented by $t$, if any.
We consider several cases, using the first case that applies:

\mycase{A} $r(t) \ge k_i$, i.e., the compaction operation represented by $t$ removes at least $k_i$ important items from the level-$h$ buffer.
Then, from \eqref{eqn:charge_def}, we obtain 
\begin{align*}
\chi(t) &= \max\{\chi(t') + \chi(t'') - 3r(t) + I(t)\cdot k_i - J(t)\cdot 2\cdot k_i, 0\}
\\
&\le \max\{\chi(t') + \chi(t'') - r(t) - k_i, 0\}
\\
&\le \max\{\rank(y; \calB_h(t')) + \rank(y; \calB_h(t'')) - r(t) - k_i, 0\}
\le \max\{\rank(y; \calB_h(t)) - k_i, 0\}\,,
\end{align*}
where the first inequality follows from the case condition and $I(t)\le 1$ (that is, we use that $2r(t)\ge I(t)\cdot k_i + k_i$),
and the second inequality uses the induction hypothesis.
This shows that \ref{chi-cond1} holds for $t$.

\mycase{B} $\chi(t') = 0$ and $\chi(t'') = 0$.
If the compaction operation represented by $t$ (if any) is not important,
then $\chi(t) = 0$ and \ref{chi-cond1} holds for $t$.
Otherwise, there is an important compaction represented by $t$, which may happen
if many important items are added to level $h$ during the level-$(h-1)$ compaction.
Then, \eqref{eqn:charge_def} and $\chi(t') = \chi(t'') = 0$ imply that
$$\chi(t) \le k_i \le B_i / 2 - k_i \le \rank(y; \calB_h(t)) - k_i\,,$$
where the second inequality uses $B_i\ge 2\cdot k_i\cdot \log_2(N_i / k_i)$ and $N_i \ge 4\cdot k_i$, 
and the last inequality follows from that there must be at least $B_i / 2$ important items
remaining in the buffer after the important compaction represented by $t$.
Hence, \ref{chi-cond1} holds for $t$.

\mycase{C} \ref{chi-cond1} holds for $t'$ and $\chi(t') > 0$, or \ref{chi-cond1} holds for $t''$ and $\chi(t'') > 0$, or both.
This condition implies that 
\begin{equation}\label{eqn:clm-chi-caseC-cond}
\chi(t') + \chi(t'') \le \rank(y; \calB_h(t')) + \rank(y; \calB_h(t'')) - k_i\,;
\end{equation}
note that $t'$ or $t''$ may be an $i'$-node for some $i' < i$,
but this inequality still holds as $k_i\le k_{i'}$ for $i' < i$.
We consider two subcases:

\mycase{C.1}
$I(t) = 0$, i.e., there is no important compaction represented by $t$.
Then, \ref{chi-cond1} holds for $t$ as
\begin{align*}
\chi(t) \le \max\{\chi(t') + \chi(t'') - 3r(t), 0\}
&\le \max\{\rank(y; \calB_h(t')) + \rank(y; \calB_h(t'')) - k_i - r(t), 0\}
\\
&\le \max\{\rank(y; \calB_h(t)) - k_i, 0\}\,.
\end{align*}

\mycase{C.2}
$I(t) = 1$, i.e., there is an important compaction represented by $t$.
In this case, we show that \ref{chi-cond2} holds for $t$.
Since the (level-$h$) compaction represented by $t$ is important, it removes $0 < r(t) < k_i$ important items from the buffer
(for $r(t)\ge k_i$, case~A applies).
Let $j$ be the smallest index of a section that contains important items only;
it must be the same before and after the compaction as $r(t) < k_i$ and as only the whole sections are compacted.
Note that we must have $j > 1$ as section~1 is involved in any compaction.
Since the compaction does not involve section $j$, we have $C'_h(t)[j-1] = 0$ for the state $C'_h(t)$ before the compaction
(recall that $C'_h(t)[j-1]$ is the $(j-1)$-st bit from the right in $C'_h(t)$). Moreover, $C'_h(t)[j'] = 1$ for all $0 < j' < j-1$ 
as the compaction involves section $j-1$.
Thus, after the compaction, it holds that $C_h(t)[j-1] = 1$.
Next, observe that $\rank(y; \calB_h(t)) = B_i - (j-1)\cdot k_i$ since the compaction involves the first $j-1$ sections and it is important.
It thus remains to obtain a suitable upper bound on $\chi(t)$:
\begin{align*}
\chi(t) \le \max\{\chi(t') + \chi(t'') - 3r(t) + k_i, 0\}
&\le \max\{\rank(y; \calB_h(t')) + \rank(y; \calB_h(t'')) - k_i - r(t) + k_i, 0\}
\\
&\le \rank(y; \calB_h(t)) = B_i - (j-1)\cdot k_i\,,
\end{align*}
where the second inequality uses~\eqref{eqn:clm-chi-caseC-cond}.
Hence, \ref{chi-cond2} holds for $t$.

\mycase{D} None of the previous cases applies.
Since we have that $\chi(t') > 0$ or $\chi(t'') > 0$ (as case~B does not apply)
and since case~C does not apply, 
property \ref{chi-cond2} holds for $t'$ or for $t''$ or for both of them.
For simplicity, we assume that \ref{chi-cond2} holds for $t'$ as the other case follows by symmetric arguments.
Let $i'\le i$ be such that $t'$ is an $i'$-node and
let $j'$ be the index from property \ref{chi-cond2} for $t'$. To recall,
$j'\ge 1$ is the smallest index of a section which contains important items only in $\calB_h(t')$,
and it holds that $\chi(t') \le B_{i'} - (j'-1)\cdot k_{i'}$ and, provided that $j' > 1$, $C_h(t')[j'-1] = 1$.
Let $C'_h(t)$ be the state of the compaction schedule just before the compaction
represented by $t$.
Since we use the bitwise OR when merging states of the compaction schedule,
we also have that $C'_h(t)[j'-1] = 1$ if $j' > 1$; see Fact~\ref{factstar}.

We consider a few further subcases:

\mycase{D.1} $i' < i$. Thus, $t'$ is a topmost $i'$-node, which together with property~\ref{chi-cond2} for $t'$ implies that $J(t) = 1$
(here, we also use that $t$ is in the subtree of a node in $Q_h$).
Then, \eqref{eqn:charge_def} becomes
\begin{align*}
\chi(t) &= \max\{\chi(t') + \chi(t'') - 3r(t) + I(t)\cdot k_i - 2\cdot k_i, 0\}
\\
&\le \max\{\chi(t') + \chi(t'') - r(t) - k_i, 0\}
\\
&\le \max\{\rank(y; \calB_h(t')) + \rank(y; \calB_h(t'')) - r(t) - k_i, 0\}
\le \max\{\rank(y; \calB_h(t)) - k_i, 0\}\,,
\end{align*}
where the second inequality uses the induction hypothesis for $t'$ and $t''$,
namely, that $\chi(t') \le \rank(y; \calB_h(t'))$ and $\chi(t'') \le \rank(y; \calB_h(t''))$.
This shows \ref{chi-cond1}.

\mycase{D.2} $i' = i$. We show that $\chi(t'') = 0$ in such a case.
Indeed, for a contradiction suppose that $\chi(t'') > 0$, which implies that property~\ref{chi-cond2} holds for $t''$ since otherwise, case~C applies.
Then, if $t''$ is an $i'$-node for $i' < i$, we use case~D.1 with $t''$ acting as $t'$.
Thus, $t''$ is an $i$-node and \ref{chi-cond2} holds for both $t'$ and $t''$, from which we obtain 
 $\rank(y; \calB_h(t'))\ge B_i/2$ and  $\rank(y; \calB_h(t''))\ge B_i/2$
 (as there is an important compaction represented by an $i$-node in the subtree of each of $t'$ and $t''$).
It follows that $\rank(y; \calB_h(t')) + \rank(y; \calB_h(t''))\ge B_i$ and since
all items with position at least $B_i - k_i + 1$ in the sorted buffer when merging (cf.~line~\ref{ln-merge:firstSlotCompaction} in Algorithm~\ref{code:merge}) are always involved in the compaction, we must have that $r(t)\ge k_i$ --- a contradiction with the assumption that case~A does not apply.
This shows $\chi(t'') = 0$.

Let $a(t)$ be the number of important items added to the level-$h$ buffer from level $h-1$
during the level-$(h-1)$compaction represented by $t$, if any.
Let $\calB'_h(t)$ be the sorted buffer obtained from merging the (level-$h$) buffers of $t'$
and $t''$ and adding $a(t)$ important items from level $h-1$, but before performing
the level-$h$ compaction represented by $t$, if any; thus $\calB'_h(t)$ may contain
more than $B_i$ items.

Note that section $j'$ is not involved in the level-$h$ compaction represented by $t$ (if any), otherwise we would have $r(t) \ge k_i$ as section $j'$ contains important items only
in $\calB_h(t')$ and thus also in $\calB'_h(t)$.
This implies that section $j'-1$ is not involved in the compaction either,
which follows from $C'_h(t)[j'-1] = 1$ (here, we refer to the state just before the compaction).
We consider a few further subcases:

\mycase{D.2.a}
$\rank(y; \calB_h(t')) + \rank(y; \calB_h(t'')) + a(t) < B_i - (j'-2)\cdot k_i$.
This means that in $\calB'_h(t)$, section $j'-2$ (for $j' > 2$) contains no important items and moreover, section $j'-1$ contains an item $z > y$ (we suppose buffer $\calB'_h(t)$ is sorted).
Since section $j'-1$ is not involved in the compaction,
it follows that $r(t) = 0$, so the compaction represented by $t$ (if any) is not important, i.e., $I(t) = 0$.
As $\chi(t'') = 0$, we get
\begin{align*}
\chi(t) \le \chi(t') \le B_i - (j'-1)\cdot k_i\,.
\end{align*}
We show \ref{chi-cond2} holds for $t$.
Indeed, section $j'-1$ of $\calB_h(t)$ contains a non-important item,
so $j'$ is still the smallest index of a section with important items only.
Furthermore, $C_h(t)[j'-1] = 1$ and there is an important level-$h$ compaction represented by an $i$-node in the subtree of $t'$ as \ref{chi-cond2} holds for $t'$, concluding that \ref{chi-cond2} holds for $t$.

\mycase{D.2.b} $\rank(y; \calB_h(t')) + \rank(y; \calB_h(t'')) + a(t) \ge B_i - (j'-2)\cdot k_i$ and $I(t) = 0$, i.e.,
the compaction represented by $t$ (if any) is not important.
As section $j'-1$ is not involved in this compaction, 
we have that $\rank(y; \calB_h(t)) \ge B_i - (j'-2)\cdot k_i$.
Then, \ref{chi-cond1} holds for $t$ as
\begin{align*}
\chi(t) \le \chi(t') \le B_i - (j'-1)\cdot k_i \le \rank(y; \calB_h(t)) - k_i\,.
\end{align*}

\mycase{D.2.c} $\rank(y; \calB_h(t')) + \rank(y; \calB_h(t'')) + a(t) \ge B_i - (j'-2)\cdot k_i$ and $I(t) = 1$, i.e.,
the compaction represented by $t$ is important. We show that \ref{chi-cond2} holds.
Let $j > 1$ be the smallest index of a section that contains important items only in $\calB_h(t)$, i.e.,
after the compaction. By the case condition and since section $j'-1$ is not involved in the compaction,
we have $j\le j'-1$.
Observe that $\rank(y; \calB_h(t)) = B_i - (j-1)\cdot k_i$ as the compaction removes important items from the buffer
and thus, it involves the first $j-1$ sections by the definition of $j$.
After the compaction, the $(j-1)$-st bit of the state is set to 1, i.e., $C_h(t)[j-1] = 1$, by the definition of the compaction.
Finally, we upper bound $\chi(t)$ as follows:
\begin{align*}
\chi(t) \le \chi(t') + k_i
\le B_i - (j'-1)\cdot k_i + k_i = B_i - (j'-2)\cdot k_i\le B_i - (j-1)\cdot k_i\,,
\end{align*}
where the last inequality uses $j\le j'-1$.
Hence, \ref{chi-cond2} holds.
\end{proof}

It remains to show that Claim~\ref{clm:charging arg for mergeability} together with the definition of $\chi(t)$ in~\eqref{eqn:charge_def}
implies Lemma~\ref{lem:single layer mergeability - stronger}, i.e., that~\eqref{eqn:single layer mergeability - stronger} holds.
To this end, first note that the definition of $\chi(t)$ for a non-leaf $i$-node $t$ with $a\le i\le \lambda$ implies
\begin{equation}\label{eqn:charge_def_wo_max}
\chi(t) \ge \chi(t') + \chi(t'') - 3r(t) + I(t)\cdot k_i - J(t)\cdot 2\cdot k_i\,,
\end{equation}
where $t'$ and $t''$ are the children of $t$.
For a node $q\in Q^j_h$ with $a\le j\le \lambda$,
consider the sum of~\eqref{eqn:charge_def_wo_max} over all non-leaf $i$-nodes $t$ for $a\le i\le j$ such that $t$ is in the subtree of $q$, and
observe that $\chi(t)$ either appears exactly once on both sides of the resulting inequality,
or $t$ appears only on the right-hand side and $\chi(t) = 0$, or $t = q$ and $q$ appears only on the left-hand side. Letting $T_q$ denote the subtree of $q$ and $T^i_q$ be the set of $i$-nodes in $T_q$, we obtain
\begin{equation}\label{eqn:charge_def_sum-for_q}
\chi(q) \ge \sum_{i = a}^{j}\sum_{t\in T^i_q} - 3r(t) + I(t)\cdot k_i - J(t)\cdot 2\cdot k_i\,.
\end{equation}
Next, consider the sum of~\eqref{eqn:charge_def_sum-for_q} over all nodes $q\in Q^j_h$ for $a\le j\le \lambda$.
Observe that if an $i$-node $t$ for $a\le i\le \lambda$ represents a compaction removing at least one important item,
then $t$ must be in the subtree $T_q$ of a node $q\in Q^j_h$ for $a\le j\le \lambda$.
Furthermore, subtrees $T_q$ are disjoint by the definition of $Q_h$.
Letting $r(T^{[a,\lambda]})$ be the total number of important items removed from level $h$
during a compaction represented by an $i$-node for $a\le i\le \lambda$ that is in the subtree of a node in $Q_h$,
we thus have the following two equalities:
\begin{align*}
\sum_{j = a}^{\lambda} \sum_{q\in Q^j_h}\sum_{i = a}^{j}\sum_{t\in T^i_q} r(t)
&= r(T^{[a,\lambda]})
\\
\sum_{j = i}^{\lambda} \sum_{q\in Q^j_h}\sum_{t\in T^i_q} I(t)
&= m^i_h \quad \text{for\ any\ }i\in [a, \lambda]\,.
\end{align*}
Hence, summing~\eqref{eqn:charge_def_sum-for_q} over all nodes $q\in Q^j_h$ for $a\le j\le \lambda$,
we get
\begin{equation}\label{eqn:charge_def_sum}
\sum_{j = a}^{\lambda} \sum_{q\in Q^j_h} \chi(q) \ge -3\cdot r(T^{[a,\lambda]}) + \sum_{i = a}^{\lambda} m^i_h\cdot k_i -\sum_{j = a}^{\lambda} \sum_{q\in Q^j_h}\sum_{i = a}^{j}\sum_{t\in T^j_q} J(t)\cdot 2\cdot k_i\,.
\end{equation}
We now upper bound the last term on the RHS of~\eqref{eqn:charge_def_sum}.
Let $\tau_{i'}$ be the number of topmost $i'$-nodes  $t'$ for $i'\le \lambda$ satisfying that there is an important level-$h$ compaction represented by an $i'$-node in the subtree of $t'$ and that $t'$ is in the subtree of a node $q\in Q^j_h$ for $a\le j\le \lambda$.
Recall that if $J(t) = 1$ for an $i$-node $t$, then at least one of the children of $t$ is a topmost $i'$-node for $a\le i' < i$ accounted for in $\tau_{i'}$.
Using $k_0\ge k_1\ge \cdots \ge k_\lambda$, we thus have 
\begin{equation}\label{eqn:topmost_nodes_w_important_bound_1}
\sum_{j = a}^{\lambda} \sum_{q\in Q^j_h}\sum_{i = a}^{j}\sum_{t\in T^j_q} J(t)\cdot 2\cdot k_i
\le \sum_{i'= a}^\lambda \tau_{i'}\cdot 2\cdot k_{i'}\,.
\end{equation}
For any $i'\in [a,\lambda]$, we claim that 
\begin{equation}\label{eqn:topmost_nodes_w_important}
\tau_{i'}\cdot \frac{B_{i'}}{2} \le \rank^{[a,\lambda]}_h(y)\,.
\end{equation}
Indeed, any topmost $i'$-node $t'$ accounted for in $\tau_{i'}$
has an important level-$h$ compaction represented by an $i'$-node in the subtree of $t'$.
At the time of this compaction operation, the buffer needs to have more than 
$B_{i'}/2$ important items (otherwise, the compaction would not be important).
Since the lowest-ranked $B_{i'}/2$ important items are never removed from 
the buffer (when its capacity is $B_{i'}$), the buffer represented by $t'$ 
has at least $B_{i'}/2$ important items.
Furthermore, these sets of at least $B_{i'}/2$ important items are disjoint for any two topmost $i'$-nodes $t'\neq t''$ accounted for in $\tau_{i'}$. Finally, all these important items are accounted for in $\rank^{[a,\lambda]}_h(y)$ as they are either removed from the level-$h$ buffer
by a compaction represented by an $i''$-node for $i' \le i''\le \lambda$,
or remain at the level-$h$ buffer represented by a node $q\in Q^j_h$ for $i'\le j\le \lambda$.
This shows~\eqref{eqn:topmost_nodes_w_important}.

Thus, the last term on the RHS of~\eqref{eqn:charge_def_sum} is bounded by
\begin{align}
\sum_{j = a}^{\lambda} \sum_{q\in Q^j_h}\sum_{i = a}^{j}\sum_{t\in T^j_q} J(t)\cdot 2\cdot k_i
&\le \sum_{i' = a}^{\lambda} \tau_{i'}\cdot 2\cdot k_{i'}
\nonumber\\
&\le  2\cdot \sum_{i' = a}^{\lambda} \frac{\tau_{i'}\cdot \frac12\cdot B_{i'}}{\log_2(N_{i'} / k_{i'})}
 \le 2\cdot \sum_{i' = a}^{\lambda} \frac{\rank^{[a,\lambda]}_h(y)}{\log_2(N_{i'} / k_{i'})}
 \le \rank^{[a,\lambda]}_h(y)\,,
 \label{eqn:topmost_nodes_w_important_bound_2}
\end{align}
where the first inequality is~\eqref{eqn:topmost_nodes_w_important_bound_1},
the second inequality uses the definition of $B_{i'}$ in~\eqref{eqn:setk-merge},
the third inequality follows from~\eqref{eqn:topmost_nodes_w_important},
and the last inequality holds as $\sum_{i' = a}^{\lambda} 1 / \log_2(N_{i'} / k_{i'}) \le 2 \log_2(N_0 / k_0) \le \frac12$ (in more detail, here we use that $\log_2(N_{i'} / k_{i'})$ increases with $i'$ by a factor of at least $2$ for $i'\le \lambda$ and that $\log_2(N_0 / k_0) \ge 4$,
which holds by the definition of $N_0$).

To upper bound the LHS of~\eqref{eqn:charge_def_sum},
we use Claim~\ref{clm:charging arg for mergeability} for each $q\in Q^j_h$ 
with $a\le j\le \lambda$ to get that $\chi(q) \le \rank(y; \calB_h(q))$ for any such $q$.
Plugging this together with~\eqref{eqn:topmost_nodes_w_important_bound_2} into~\eqref{eqn:charge_def_sum}, we obtain
\begin{equation}\label{eqn:charge_def_sum_2}
\sum_{j = a}^{\lambda} \sum_{q\in Q^j_h}\rank(y; \calB_h(q)) \ge -3\cdot r(T^{[a,\lambda]}) + \sum_{i = a}^{\lambda} m^i_h\cdot k_i -\rank^{[a,\lambda]}_h(y)\,,
\end{equation}
which implies~\eqref{eqn:single layer mergeability - stronger} by rearranging and using
$\rank^{[a,\lambda]}_h(y) = r(T^{[a,\lambda]}) + \sum_{j = a}^{\lambda} \sum_{q\in Q^j_h}\rank(y; \calB_h(q))$ (the second term equals the total number of important items
in the buffers represented by nodes in $Q^j_h$ for $a\le j\le \lambda$).
\end{proof}

Lemma~\ref{lem:single layer mergeability - stronger} with $a = 0$ has a simple corollary.

\begin{corollary}\label{cor:single layer mergeability}
Consider level $h$ and let $m^{\le \lambda}_h  = \sum_{i = 0}^{\lambda} m^i_h$ be the total number of important compactions at level $h$
represented by $i$-nodes for $i\le \lambda$.
Suppose that $\rank_h(y) \le 2^{-h+2}\rank(y)$ and let $i(h)\ge 0$ be the largest
integer $0\le i\le \lambda$ satisfying $2^{-h+2}\rank(y) > B_i / 2$.
Then $m^{\le \lambda}_h \le 4\rank_h(y) / k_{i(h)}$. 
\end{corollary}
\begin{proof}
	This follows from Lemma~\ref{lem:single layer mergeability - stronger}
	by observing that $m^i_h = 0$ for $i(h) < i \le \lambda$ and
	by using $k_i\ge k_{i(h)}$ for any $i\le i(h)$ and $\rank^{[0,\lambda]}_h(y) \le \rank_h(y)$
\end{proof}

\subsection{Analysis of the Full Sketch for Mergeability}
\label{s:k0sec}

In this section, we complete the proof of full mergeability that matches our result in the streaming setting (Theorem~\ref{thm:sketch}).
The crucial part of analyzing the full sketch, similarly as in the streaming setting (Section~\ref{s:fullanalysis}), is bounding the variance of $\err(y)$, using the bounds on the number of important level-$h$ compactions from the previous section. The bound of this section is, however, substantially more involved than in the streaming setting, mainly because parameters $k$ and $B$ of the sketches change as merge operations are processed. 
Here, we again stress that we assume no advance knowledge of $n$, the total size of the input.

Before presenting the most general and tight analysis, we will however describe that
a simple extension of the arguments used in the streaming setting readily gives the result with an additional factor of $\min\{\log\log(\eps n), \log(\epsilon^{-1}) + \log\log(\delta^{-1})\}$
in the asymptotic space complexity, relative to our result in the streaming setting (Theorem \ref{thm:sketch}).\footnote{
We provide a detailed description of a simpler analysis with an additional $\log\log(\eps n)$ factor in a prior version of this manuscript; see \url{https://arxiv.org/abs/2004.01668v3}.
}
This simpler, non-tight analysis of the full sketch is less delicate than our analysis that
avoids the additional factor, thereby establishing Theorem \ref{thm:mergeability-full-intro}.
We nevertheless do not assume any advance knowledge about the final input size $n$.

\subsubsection{A Sketch of a Simpler Analysis with an Additional Double Logarithmic Factor}
\label{sec:mergeability analysis w/ additional factor}

The key trick that allows to apply similar arguments as in Section~\ref{s:fullanalysis} is to modify the definition of $k_i$ for $i\ge 0$ compared to~\eqref{eqn:setk-merge}, as follows:
\begin{equation}\label{eqn:setk-merge - weaker}
k_i = \Theta(1)\cdot \left\lceil\frac{\min\{i+1, \lambda\}\cdot \hat{k}}{\sqrt{\log_2(\overline{N_i} / \hat{k})}}\right\rceil\,,
\end{equation}
where $\lambda$ and $\overline{N_i}$ are defined similarly as in~\eqref{eqn:setk-merge} and the multiplicative constant is set appropriately.
In particular, relative to Equation \eqref{eqn:setk-merge}, note the extra factor of $\min\{i+1, \lambda\}$;
including it considerably simplifies the analysis,
but it is responsible for an additional $\min\{\log\log(\eps n), \lambda\}$ term in the space bound, where $\lambda = O(\log(\epsilon^{-1}) + \log\log(\delta^{-1}))$.
Recall that $B_i = 2\cdot k_i\cdot \left\lceil \log_2 (\overline{N_i} / k_i)\right\rceil$.

We omit the detailed analysis
and only highlight where we use the modified definition of the parameter $k_i$.
As in the subsequent tight analysis, the error from compactions represented by $i$-nodes for $i>\lambda$ (if any)
will be analyzed separately (and is much simpler to deal with).
In particular, a similar calculation as in~\eqref{eqn:boundOn_kB} gives us that 
for $0\le i\le \lambda$,
\begin{equation}\label{eqn:kB lower bound mergeability - weaker - strong}
k_i\cdot B_i \ge \Theta(1)\cdot \frac{(i+1)^2}{\eps^2}\cdot \ln\frac{1}{\delta}\,,
\end{equation}
so we have an extra factor of $(i+1)^2$ compared to~\eqref{eqn:boundOn_kB}.
Using Corollary~\ref{cor:single layer mergeability}, one can show that
\begin{align*}
\V[\err(y)]
\le \sum_{i = 0}^{\lambda} \Theta(1)\cdot \frac{\rank(y)^2}{k_i\cdot B_i}
\le \frac{\eps^2\cdot \rank(y)^2}{4\cdot \ln(1/\delta)}\cdot \sum_{i = 0}^{\lambda} \frac{1}{(i+1)^2}
\le \frac{\eps^2\cdot \rank(y)^2}{2\cdot \ln(1/\delta)}\,,
\end{align*} 
where the second inequality uses~\eqref{eqn:kB lower bound mergeability - weaker - strong}
and the last step holds as $\sum_{i = 0}^{\lambda} 1/(i+1)^2 < \pi^2 / 6 < 2$;
the fact that this sum is bounded allows us to deal with the challenge of 
changing parameters $k$ and $B$ in a simple way.
The application of the tail bound for sub-Gaussian variables and the derivation of the space bound
is otherwise the same as in Theorem~\ref{thm:sketch}.

\subsubsection{A Tight Analysis}

Recall that $\hat{k} = 4\eps^{-1}\cdot \sqrt{\ln(1/\delta)}$ and that by~\eqref{eqn:setk-merge}, $k_i = 2^5\cdot \left\lceil \hat{k} / \sqrt{\log_2(\overline{N_i} / \hat{k})}\right\rceil$
and $B_i = 2\cdot k_i\cdot \lceil \log_2(\overline{N_i} / k_i)\rceil$,
where $\overline{N_i}  = \min\left\{N_i\,,\, N_\lambda\right\}$ by~\eqref{eqn:setk-merge}.
Here, $\lambda\ge 0$ is the smallest integer $i$ such that
$\hat{k} / \sqrt{\log_2(N_i / \hat{k})} \le 1$.
If $\lambda > \ell$, we decrease $\lambda$ to $\ell$ for convenience.
Using a similar calculation as in Claim~\ref{clm:kBlowerBound}, we show a lower bound on $k_i\cdot B_i$.

\begin{claim}\label{clm:kBlowerBound - mergeability}

	Parameters $k_i$ and $B_i$ set according to~\eqref{eqn:setk-merge}

	satisfy
	\begin{equation}\label{eqn:kB lower bound mergeability - strong}
	k_i\cdot B_i \ge 2^{14}\cdot \frac{1}{\eps^2}\cdot \ln\frac{1}{\delta}\,.
	\end{equation}
\end{claim}

\begin{proof}
	We first need to relate
	$\log_2(\overline{N_i}/k_i)$ (used to define $B_i$) and $\log_2(\overline{N_i} / \hat{k})$ (that appears in the definition of $k_i$).
	As $k_i \le 2^5\cdot \hat{k}$, it holds that $\log_2(\overline{N_i}/k_i) \ge \log_2(\overline{N_i}/\hat{k}) - 5 \ge \log_2(\overline{N_i}/\hat{k}) / 2$,
	where we use that $\overline{N_i}\ge N_0\ge 2^{10}\cdot \hat{k}$, so $\log_2(\overline{N_i}/\hat{k}) \ge 10$.
	Using this, we bound $k_i\cdot B_i$ as follows:
	\begin{align*}
	k_i\cdot B_i & = 2\cdot k_i^2\cdot \left\lceil\log_2 \frac{\overline{N_i}}{k_i}\right\rceil 
	\ge 2\cdot 2^{10}\cdot \frac{16}{\eps^2}\cdot \frac{\ln\frac{1}{\delta}}{\log_2(\overline{N_i} / \hat{k})}\cdot \frac{\log_2(\overline{N_i}/\hat{k})}{2}
	= 2^{14}\cdot \frac{1}{\eps^2}\cdot \ln\frac{1}{\delta}\,.
	\end{align*}
\end{proof}

For analyzing the case $n > N_\lambda$, the following bound will be useful:

\begin{align}
B_\lambda &\ge 2^5\cdot \hat{k}^2 \label{eqn:B_lambda_bnd_HatKsquared}
\end{align}
This is because the definition of $\lambda$ implies that
$\sqrt{\log_2(N_\lambda / \hat{k})} \ge \hat{k}$ while $k_\lambda \ge 2^5$,
thus
$B_\lambda
\ge 2^6\cdot \log_2 \left(N_\lambda/k_\lambda\right)
\ge 2^6\cdot \log_2 \left(N_\lambda/\hat{k}\right) / 2
\ge 2^5\cdot \hat{k}^2
$,
where the second inequality follows from the same argument as in Claim~\ref{clm:kBlowerBound - mergeability}.

For any $0\le i\le \ell$, let $H_i(y)$ be the minimal $h$ for which $2^{-h+2} \rank(y) \leq B_i/2$.
As $y$ is fixed, we write $H_i$ rather than $H_i(y)$ for brevity. 
In particular, by considering $h = H_i -1$ (assuming $H_i > 0$), it can be seen that $2^{3-H_i} \rank(y) \geq B_i/2$,
or equivalently
\begin{equation}\label{eqn:H_i-bound mergeability}
2^{H_i} \leq 2^4\cdot \rank(y) / B_i\,.
\end{equation}
As increasing $i$ by one increases $B_i$, we have $H_0\ge H_1\ge \cdots \ge H_\ell$.

We show below that no important item (i.e., one smaller than or equal to $y$) can ever reach level $H_0+1$.

\begin{lemma} \label{lem:decay rank mergeability}
Assuming $H_0 > 0$,
with probability at least $1 - \delta$ it holds that $\rank_h(y) \leq 2^{-h+2}\rank(y)$ for any $h \le H_0$.
\end{lemma}

\begin{proof}
The proof is similar to that of Lemma~\ref{lem:decay rank}, except that we need to deal with parameters $k$ and $B$ changing over time.
To this end, we use an idea from the KLL paper~\cite{karnin2016optimal} to analyze the top $\log\log 1/\delta$ levels deterministically.
We define $$H'_0 = \max\left(0, H_0 - \left\lceil\log_2\sqrt{\ln\frac{1}{\delta}}\right\rceil\right)\,.$$
(Note that for $\delta \le 0.5$, we have $\left\lceil\log_2\sqrt{\ln\frac{1}{\delta}}\right\rceil\ge 0$.)

We first show by induction on $0\le h \le H'_0$ that $\rank_h(y) \leq 2^{-h+1}\rank(y)$ with probability at least
$1 - \delta\cdot 2^{h - H'_0 - 1}$, conditioned on $\rank_{h'}(y) \leq 2^{-h'+1}\rank(y)$ for any $h' < h$.
The base case holds by $\rank_0(y) = \rank(y)$.

Consider $0 < h \le H'_0$,
and recall that $m_{h'}$ denotes the number of important compactions at level $h'$ over all merge operations represented in the merge tree $T$. 
As in the proof of Lemma~\ref{lem:decay rank}, 
\[\Pr[\rank_h(y) > 2^{-h+1}\rank(y)] \le \Pr[Z_h > 2^{-h}\rank(y)],\]
where $Z_h = \sum_{h' = 0}^{h - 1} 2^{-h + h'}\cdot \mathrm{Binomial}(m_{h'})$ is a zero-mean sub-Gaussian random variable. 
To bound the variance of $Z_h$, first note that for any $h' < h$, 
since each important compaction needs to remove at least one important item from the buffer,
we have that
$m_{h'}\le \rank_{h'}(y) \le 2^{-h'+1}\cdot \rank(y)$, using the assumption that $\rank_{h'}(y) \leq 2^{-h'+1}\cdot \rank(y)$.
(While this may seem like a very crude bound compared to Lemma~\ref{lem:single layer mergeability - stronger},
it is sufficient due to analyzing top levels deterministically and furthermore, it can be used for compactions
represented by $i$-nodes for $i > \lambda$, where we do not use the deterministic compaction schedule.)

As $\V[\mathrm{Binomial}(n)] = n$, the variance of $Z_h$ is
\begin{align*}
\V[Z_h] \le \sum_{h' = 0}^{h - 1} 2^{-2h + 2h'}\cdot m_{h'}
	 \le \sum_{h' = 0}^{h - 1} 2^{-2h + 2h'}\cdot 2^{-h'+1}\cdot \rank(y)
	 =  \sum_{h' = 0}^{h - 1} 2^{-2h + h' + 1}\cdot \rank(y)
	\le 2^{-h + 1}\cdot \rank(y)\,.
\end{align*}

To bound $\Pr[Z_h > 2^{-h}\cdot \rank(y)]$, we apply the tail bound for sub-Gaussian variables (Fact~\ref{fact:tailBound-subGaussian})
to get
\begin{align*}
\Pr[Z_h > 2^{-h}\cdot \rank(y)]
&< \exp\left(-\frac{2^{-2h}\cdot \rank(y)^2}{2\cdot (2^{-h + 1}\cdot \rank(y))} \right)
\\
&= \exp\left(-2^{-h - 2}\cdot \rank(y) \right)
\\
&= \exp\left(-2^{H_0 - H'_0}\cdot 2^{-h + H'_0 - 6}\cdot 2^{4 - H_0} \rank(y) \right)
\\
&\le \exp\left(-\sqrt{\ln\frac{1}{\delta}}\cdot 2^{-h + H'_0 - 6}\cdot B_0 \right)
\\
&\le \exp\left(-\ln\frac{1}{\delta}\cdot 2^{-h + H'_0 + 1} \right)
= \delta^{2^{- h + H'_0 + 1}}
\le \delta\cdot 2^{ - H'_0 + h - 1}\,,
\end{align*}
where the second inequality uses the definition of $H'_0$ and $2^{4 - H_0} \rank(y) \ge B_0$ by~\eqref{eqn:H_i-bound mergeability},
the third inequality follows from~$B_0\ge 2\cdot k_0\cdot \log_2 (N_0 / k_0)\ge 2^7\cdot \sqrt{\ln(1/\delta)}$, and the last inequality uses $\delta \le 0.5$.
Hence, taking the union bound over levels $h \le H'_0$, 
with probability at least $1 - \delta$ it holds that $\rank_h(y) \leq 2^{-h+1}\rank(y)$ for any $h \le H'_0$.

Finally, consider level $h$ with $H'_0 < h \le H_0$ and condition on $\rank_{H'_0}(y) \leq 2^{-H'_0+1}\rank(y)$.
(In the case $H'(y) = 0$, we have $\rank_0(y) = \rank(y)$.)
We again proceed by induction and assume that $\rank_{h'}(y) \le 2^{-h'+2}\cdot \rank(y)$
for any $h' < h$.
First, we argue that for any  $h$ with $H'_0 < h \le H_0$ it holds that $\sum_{i > \lambda}m^i_{h'} = 0$, so we can use Corollary~\ref{cor:single layer mergeability}.
Indeed, it is sufficient to show $\rank_{H'_0}(y) \leq B_\lambda / 2$ as follows:
\begin{align}	
\rank_{H'_0}(y)
\le 2^{-H'_0+1}\rank(y)
=   2^{H_0 - H'_0}\cdot 2^{-2}\cdot 2^{-H_0+3}\rank(y)
\le 2\sqrt{\ln\frac{1}{\delta}} \cdot 2^{-2}\cdot B_0
\le 2^4\cdot \hat{k}^2
\le \frac12 B_\lambda\,, \label{eqn:rank_lambda_bnd}
\end{align}
where the penultimate inequality uses the definitions of $\hat{k}$ and $B_0$ in~\eqref{eqn:setk-merge} and the last inequality is by~\eqref{eqn:B_lambda_bnd_HatKsquared}.

We now observe that for any $H'_0 < h' \le h$, it holds that $\rank_{h'}(y) \leq \frac12 \cdot (1 + 4/k_{i(h')})\cdot \rank_{h'-1}(y)$, where $i(h') \le \lambda$ is the largest integer $i$ satisfying $2^{-h'+3}\rank(y) > B_i / 2$.
Indeed, $\rank_{h'}(y) \le \frac12\cdot \left(\rank_{h' - 1}(y) + \mathrm{Binomial}(m_{h'-1})\right)$
(see~Equation~\ref{eqn:rank_process}) and $\mathrm{Binomial}(m_{h'-1}) \le m_{h'-1} \le 4\rank_{h' - 1}(y) / k_{i(h')}$
by Corollary~\ref{cor:single layer mergeability}, using the definition of $i(h')$
and the induction hypothesis for level $h'-1$, i.e., $\rank_{h'-1}(y) \le 2^{-h'+3}\cdot \rank(y)$.
That is, regardless of the outcome of the random choices, we always
obtain this bound on the rank of an item.
By using this deterministic bound for levels $H'_0 < h'\le h$, we get
\begin{align}
	\rank_h(y)
	\le 2^{-h + H'_0} \cdot \rank_{H'_0}(y)\cdot 
		\prod_{h' = H'_0 + 1}^{h} \left(1 + \frac{4}{k_{i(h')}}\right)
	\le 2^{-h + H'_0}\cdot 2^{-H'_0+1}\cdot \rank(y)\cdot 
		\prod_{h' = H'_0 + 1}^{h} \left(1 + \frac{4}{k_{i(h')}}\right)\,.
	\label{eqn:rank_bound_mergeability_det}
\end{align}
It remains to show that the product $\prod_{h' = H'_0 + 1}^{h} \left(1 + \frac{4}{k_{i(h')}}\right)$
is bounded by $2$, which implies $\rank_h(y)\le 2^{-h+2}\cdot \rank(y)$.
We first observe that $k_{i(H'_0 + 1)}\ge k_\lambda \ge 2^5$, since $i(H'_0 + 1)\le \lambda$. 
Next, recall that the sequence of $k_i$'s decreases exponentially with a factor of $\sqrt{2}$ (up to rounding)
with increasing~$i$. Thus, it is sufficient to show that the sequence $i(h')$ \emph{decreases}
for $h' = H'_0 + 1, \dots, h$. More precisely, we show that $i(h'+1) \le i(h') - 1$ for $h' = H'_0 + 1, \dots, h-1$
This latter inequality holds as increasing $h'$ by one in $2^{-h'+3}\rank(y) > B_i / 2$
implies that the largest $i$ satisfying the inequality should decrease by at least one
(recall that the sequence of $B_i$'s increases by a factor of $\sqrt{2}$ (up to rounding) with increasing $i$).
Note that we always have $2^{-h'+3}\rank(y) > B_0 / 2$ as $h'\le h\le H_0$.
Summing up, we get
\begin{align*}
	\prod_{h' = H'_0 + 1}^{h} \left(1 + \frac{4}{k_{i(h')}}\right)
	&\le \prod_{j\ge 0} \left(1 + \frac{1}{2^3\cdot \sqrt{2}^j}\right)
	\\
	&\le \exp\left(\sum_{j \ge 0} \log\left(1 + \frac{1}{2^3\cdot \sqrt{2}^j}\right)\right)
	\le \exp\left(\sum_{j \ge 0} \frac{1}{2^3\cdot \sqrt{2}^j}\right) \le 2\,.
\end{align*}
We remark that the last inequality has a slack, which is sufficient to deal with the rounding issues mentioned above.
\end{proof}

As a corollary, we obtain a bound on the highest level with a compaction removing important items
from the level-$h$ buffer (no matter whether such a compaction is important or not).
Recall from Section~\ref{s:mergeability-singleLevel} that a compaction \emph{involves important items} iff it removes at least one important item from the buffer.
Recall that we only consider a compaction to be important if it
affects an odd number of important items, so these compactions
involving important items are a superset of the important compactions. 

\begin{lemma}\label{lem:no error at Hi}
Conditioned on the bounds in Lemma~\ref{lem:decay rank mergeability} holding,
for any $0\le i\le \ell$,
no compaction involving important items occurs at level $H_i$ or above
during any merge procedure represented by any $i$-node in the merge tree $T$. 
\end{lemma}

\begin{proof}
By Lemma~\ref{lem:decay rank mergeability}, $\rank_{H_i}(y) \leq 2^{-H_i+2}\rank(y) \le B_i / 2$,
where the second inequality follows from the definition of $H_i$.
Hence, no important item is ever removed from level $H_i$ during merge operations represented by $i$-nodes when the buffer size is $B_i$.
The same argument also works for any level $h > H_i$.
\end{proof}

Consider level $h$.
Recall from Section~\ref{s:mergeability-singleLevel} that $Q_h$ is the set of nodes $t$ such that (i) $t$ is an $i$-node for $i\le \lambda$ that represents a level-$h$ compaction involving important items
(this compaction may or may not be important),
and (ii) there is no node $t'$ on the path from the parent of $t$ to the topmost $\lambda$-node containing $t$ in its subtree such that $t'$ represents a level-$h$ compaction involving important items.
Intuitively, $Q_h$ captures ``maximal'' nodes (with index $i\le \lambda$) that represent a level-$h$ compaction removing one or more important items from level $h$.
Note that an important item that remains in the level-$h$ buffer represented by a node $t\in Q_h$
(after performing the compaction operation represented by $t$) is never removed from the level-$h$ buffer, by the definition of $Q_h$.
For $0\le i\le \lambda$, let $Q^i_h$ be the set of $i$-nodes in $Q_h$ and
let $q^i_h = |Q^i_h|$.

Note that $q^i_h = 0$ for $h \ge H_0$ by Lemma~\ref{lem:no error at Hi}
(conditioned on the bounds in Lemma~\ref{lem:decay rank mergeability} holding).
Now we observe that values $q^i_h$ for $i = 0,\dots, \lambda$  give upper bounds on the number of important items
at level $h$. This follows from the fact that the level-$h$ buffer represented by a node in $Q^i_h$
contains at most $B_i$ items.

\begin{obs}\label{obs:qi_h UB}
For any $h\ge 0$ and $0\le g\le \lambda$, the level-$h$ buffers of the sketches represented by nodes in $Q^i_h$ for some $i\ge g$
contain at most  $\sum_{i = g}^\lambda q^i_h\cdot B_i$ important items in total
(after performing compaction operations represented by these nodes).
\end{obs}

Next, in Observation~\ref{obs:qi_h LB}, we show that the $q^i_h$ values can as well be used to lower bound the total number of important items
at level $h$ in topmost $\lambda$-nodes. Combined with Lemma~\ref{lem:rank initial bound}, this will give us a useful bound
on $\sum_{h\ge 0} \sum_{i = 0}^\lambda 2^h\cdot q^i_h\cdot B_i$ at the very end of the analysis.

In the observation, we also take into account items added to level $h$ from compactions (at level $h-1$ if $h > 0$)
that are \emph{not} represented by a node in the subtree of a node in $Q_h$.
Namely, for $h>0$ and any $0\le i\le \lambda$, let $z^i_h$ be the number of items added to level $h$ during merge operations
represented by $i$-nodes that are \emph{not} in the subtree of a node in $Q_h$.
For $h=0$, we define $z^i_0 = 0$ for any $i$.

\begin{obs}\label{obs:qi_h LB}
For any level $h$, the level-$h$ buffers of topmost $\lambda$-nodes contain at least 
$\sum_{i = 0}^\lambda q^i_h\cdot B_i / 2 + z^i_h$ important items.
\end{obs}

\begin{proof}
Consider an $i$-node $t\in Q^i_h$ and the level-$h$ buffer represented by $t$.
As the level-$h$ compaction represented by $t$ removes one or more important items and as $t$ is an $i$-node,
there must be at least $B_i / 2$ important items in the level-$h$ buffer that remain there after the compaction operation is done.
Furthermore, by condition~(ii) in the definition of $Q_h$, these $B_i / 2$ important items are not removed from the level-$h$ buffer
and the sets of these $B_i / 2$ important items for two nodes $t, t'\in Q_h$ are disjoint.
Finally, the $z^i_h$ items added to level $h$ during merge operations represented by $i$-nodes
that are not in the subtree of a node in $Q_h$ are disjoint (w.r.t.\ index $i$) and distinct from items 
in the buffers of nodes in $Q_h$, which shows the claim.
\end{proof}

Note that using Observation~\ref{obs:qi_h LB}, the values of $\sum_{i = 0}^\lambda q^i_h\cdot B_i / 2 + z^i_h$ 
give a lower bound on the rank of $y$ estimated by the topmost $\lambda$-nodes 
(if $\ell = \lambda$, then the only topmost $\lambda$-node is the root of the merge tree $T$).
We now complement it with an upper bound showing that the rank of $y$ estimated by the topmost $\lambda$-nodes 
cannot be too far from $\rank(y)$.
This can be seen as an initial bound on the error which will be used within the proof of the final, more refined bound
on the variance of $\err(y)$.

\begin{lemma} \label{lem:rank initial bound}
	Conditioned on the bounds in Lemma~\ref{lem:decay rank mergeability} holding,
	with probability at least $1 - \delta$ it holds that 
	\begin{equation}\label{eqn:rank_initital_bound}
	\sum_{i = 0}^\lambda \sum_{h' = 0}^{H_i} 2^{h'}\cdot \left(q^i_{h'}\cdot \frac{B_i}{2} + z^i_{h'}\right)
	\le 2\rank(y)
	\end{equation}

\end{lemma}

\begin{proof}
	Note that $q^i_h = 0$ for $h \ge H_i$ and that there is no important compaction
	represented by an $i$-node at any level $h\ge H_i$ by Lemma~\ref{lem:no error at Hi}.
	Let $\err^{\le \lambda}(y)$ be the error introduced by compactions represented 	by $i$-nodes for $i\le \lambda$.
	By Observation~\ref{obs:qi_h LB}, it is sufficient to show that $\err^{\le \lambda}(y)\le \rank(y)$.
	Recall that $\err^{\le \lambda}(y)$ is a zero-mean sub-Gaussian random variable. 
	Similarly as in Lemma~\ref{lem:decay rank mergeability}, we define
	$$H'_0 = \max\left(0, H_0 - \left\lceil\log_2\sqrt{\ln\frac{1}{\delta}}\right\rceil\right)\,.$$

	We split $\err^{\le \lambda}(y)$, the error of the rank estimate for $y$, into two parts (we drop the superscript $\le \lambda$ for simplicity):
	$$\err'(y) = \sum_{h=0}^{H'_0 - 1} 2^h\cdot \err_h(y) \quad\text{and}\quad \err''(y) = \sum_{h=H'_0}^{H_0-1} 2^h\cdot \err_h(y)\,.$$
	Note that $\err^{\le \lambda}(y) = \err'(y) + \err''(y)$; we bound both these parts by $\frac12\rank(y)$ w.h.p., starting with $\err'(y)$.
	If $H'_0 = 0$, then clearly $\err'(y) = 0$.
	Otherwise, we analyze the variance of the zero-mean sub-Gaussian variable $\err'(y)$ as follows:
	\begin{align*}
	\V[\err'(y)] &= \sum_{h=0}^{H'_0 - 1} 2^{2h}\cdot \V[\err_h(y)]
	\\
	&\le \sum_{h=0}^{H'_0 - 1} 2^{2h}\cdot \rank_h(y)
	\\
	&\le \sum_{h=0}^{H'_0 - 1} 2^{2h}\cdot 2^{-h+2}\rank(y)
	\\
	&\le 2^{H'_0+2}\cdot \rank(y)
	= 2^{H'_0-H_0 + 2}\cdot 2^{H_0}\cdot  \rank(y)
	\le 2^{H'_0-H_0 + 6}\cdot \frac{\rank(y)^2}{B_0}
	\le \frac{\rank(y)^2}{8\ln\frac{1}{\delta}}
	\end{align*}
	where the first inequality is using a simple bound of $\V[\err_h(y)]\le \rank_h(y)$, the second follows from Lemma~\ref{lem:decay rank mergeability},
	and the fourth inequality uses $2^{H_0} \leq 2^4\cdot \rank(y) / B_0$ by~\eqref{eqn:H_i-bound mergeability},
	and the last inequality follows from the definition of $H'_0$ and $B_0\ge 2^9\cdot \sqrt{\ln(1/\delta)}$ by~\eqref{eqn:setk-merge}.
	We again apply Fact~\ref{fact:tailBound-subGaussian} to obtain
	\begin{align*}
	\Pr\left[ \err'(y) > \frac12\rank(y) \right] 
	< \exp\left( -\frac{\rank(y)^2\cdot \ln\frac{1}{\delta}\cdot }{4\cdot2\cdot \frac18 \rank(y)^2} \right) 
	= \exp\left( -\ln \frac{1}{\delta} \right)
	= \delta\,.
	\end{align*}
	
	Finally, we use deterministic bounds to analyze $\err''(y)$, using that we only care about $i$-nodes for $i\le \lambda$
	As in Lemma~\ref{lem:decay rank mergeability}, let $i(h) \le \lambda$ be the largest integer $i$ satisfying $2^{-h+2}\rank(y) > B_i / 2$. 
	Then
	\begin{align*}
	\err''(y) & \le \sum_{h=H'_0}^{H_0-1} 2^h\cdot m^{\le \lambda}_h\\
	& \le \sum_{h=H'_0}^{H_0-1} 2^h\cdot \frac{4\rank_h(y)}{k_{i(h)}}
	\le \sum_{h=H'_0}^{H_0-1} 2^h\cdot \frac{2^{-h+4}\rank(y)}{k_{i(h)}}
	= \sum_{h=H'_0}^{H_0-1} \frac{2^4\cdot \rank(y)}{k_{i(h)}}
	\le \frac{\rank(y)}{2}\,,
	\end{align*}
	where the second inequality is by Corollary~\ref{cor:single layer mergeability}, the third by Lemma~\ref{lem:decay rank mergeability}, and
	the last inequality uses that $k_{i(H'_0)}\ge 2^5$ and that the 
	values of $k_{i(h)}$ for $h\in [H'_0, H_0-1]$ increase exponentially with
	increasing $h$ (by a factor of $\sqrt{2}$), which follows from similar arguments as in the paragraph below~\eqref{eqn:rank_bound_mergeability_det} in Lemma~\ref{lem:decay rank mergeability}.
\end{proof}

The following technical lemma bounds the variance on each level in a somewhat different way than in the streaming setting (Section~\ref{s:fullanalysis}).
The idea is to bound the variance in terms of the $q^i_h$ values so that we can then use Observation~\ref{obs:qi_h LB}.
To this end, we first use Observation~\ref{obs:qi_h UB} to bound $\rank_h(y)$ in terms of the $q^i_h$ values,
using the following observation:
For each important item at level $h+1$, there are roughly two important items removed from level $h$.
Here, ``roughly'' refers to the fact that each level-$h$ compaction operation that promotes $b\ge 1$ important items
removes at most $2b + 1\le 3b$ important items from the level-$h$ buffer.
Applying this observation together with Observation~\ref{obs:qi_h UB}, we show by an induction on $h$ that
$\rank^{[0,\lambda]}_h(y) \le \sum_{i = 0}^\lambda \sum_{h' \ge h} 2\cdot 3^{h' - h}\cdot (q^i_{h'}\cdot B_i + z^i_{h'})$.
Recall that $\rank^{[a,\lambda]}_h(y)$ is the number of important items that are
either removed from level $h$ during a compaction represented by an $i$-node for $a\le i\le \lambda$,
or remain at the level-$h$ buffer represented by a node $t\in Q^i_h$ for $a\le i\le \lambda$ (after the compaction operation represented by $t$ is done).
Note that this provides alternative rank bounds to Lemma~\ref{lem:decay rank mergeability}.

Then we apply Lemma~\ref{lem:single layer mergeability - stronger} to get our variance bound,
which however brings additional technical difficulties. To overcome them,
we use a careful proof by induction over $g\in [0,\lambda]$.
We will only focus on $i$-nodes with $i\le \lambda$ and on levels $h \ge H_{\lambda+1}$;
the error from remaining nodes and levels will be analyzed later.

\begin{lemma}\label{lem:var bound for full mergeability}
Conditioned on the bounds in Lemma~\ref{lem:decay rank mergeability} holding,
for any $h\ge H_{\lambda+1}$, it holds that
\begin{equation}\label{eqn:var bound for full mergeability - lemma}
\V[\err_h(y)] \le \sum_{i = 0}^\lambda \sum_{h' \ge h} \frac{8\cdot 3^{h' - h}\cdot (q^i_{h'}\cdot B_i + z^i_{h'})}{k_i}\,.
\end{equation}
\end{lemma}

\begin{proof}
We first note that by Lemma~\ref{lem:no error at Hi} (conditioned on Lemma~\ref{lem:decay rank mergeability}),
there is no important compaction at any level $h\ge H_{\lambda+1}$ represented by an $i$-node for $i > \lambda$.
Therefore, our focus will again be solely on $i$-nodes for $i\le \lambda$.
As outlined above, we first bound $\rank^{[g,\lambda]}_h(y)$ for any $0\le g\le \lambda$ and in particular, we prove by a ``backward'' induction on $h = H, H-1,\dots, H_{\lambda+1}$
that the following inequality holds for any fixed $0\le g\le \lambda$:
\begin{equation}\label{eqn:rank bound for full mergeability}
\rank^{[g,\lambda]}_h(y) \le \sum_{i = g}^\lambda \left( \sum_{h' \ge h+1} \left( 2\cdot 3^{h' - h}\cdot (q^i_{h'}\cdot B_i + z^i_{h'})\right)
	+  2\cdot q^i_h\cdot B_i	\right)\,.
\end{equation}

At level $h = H$, there is no important compaction, implying that $\rank^{[g,\lambda]}_H(y) = 0$ and $q^i_H = 0$ for any~$i$,
which establishes the base case.

Consider $h < H$ and suppose that~\eqref{eqn:rank bound for full mergeability} holds for $h+1$, i.e., we have that
\begin{equation}\label{eqn:rank bound for full mergeability - IndHyp}
\rank^{[g,\lambda]}_{h+1}(y) \le \sum_{i = g}^\lambda \left(\sum_{h' \ge h+2} \left(2\cdot 3^{h' - h - 1}\cdot (q^i_{h'}\cdot B_i + z^i_{h'})\right)
			+ 2\cdot q^i_{h+1}\cdot B_i\right)\,.
\end{equation}

To show~\eqref{eqn:rank bound for full mergeability}, we first bound the number of important items removed from level $h$
in terms of $\rank^{[g,\lambda]}_{h+1}(y)$.
For brevity, let $z^{[g,\lambda]}_{h+1} = \sum_{i=g}^\lambda z^i_{h+1}$.
Note that there are at most $\rank^{[g,\lambda]}_{h+1}(y) + z^{[g,\lambda]}_{h+1}$ important items added to level $h+1$
during compactions represented by $i$-nodes for some $i\in [g, \lambda]$, since each such important item
either gets removed from level $h+1$ or remains in the level-$(h+1)$ buffer represented by a node in $Q^i_{h+1}$ for some $i\in [g, \lambda]$
or is added to level $h+1$ during a merge operation represented by an $i$-node $t$ for $i\in[g, \lambda]$
such that $t$ is not in the subtree of a node in~$Q_{h+1}$.
Further, observe that each compaction which adds $b$ important items to level $h+1$ removes at most $2b+1$ important items
from the level-$h$ buffer --- more precisely, it removes $2b$ important items if it is not important,
and otherwise, it removes either $2b-1$, or $2b+1$ important items.
The number of important compactions represented by $i$-nodes for some $i\in [g, \lambda]$
is at most $\rank^{[g,\lambda]}_h(y) / 5$ by Lemma~\ref{lem:single layer mergeability - stronger} with $a = g$
and by $k_i\ge 20$ for any $i\le \lambda$.
Thus, the number of important items removed from level $h$ during compactions represented by $i$-nodes for $i \in [g, \lambda]$
is upper bounded by $2\rank^{[g,\lambda]}_{h+1}(y) + 2z^{[g,\lambda]}_{h+1} + (\rank^{[g,\lambda]}_h(y) / 5)$.

By Observation~\ref{obs:qi_h UB}, at most $\sum_{i = g}^\lambda q^i_h\cdot B_i$ important items remain at
the level-$h$ buffers of the sketches represented by nodes in $Q^i_h$ for some $i\ge g$.
We thus have that 
\[\rank^{[g,\lambda]}_h(y)\le 2\rank^{[g,\lambda]}_{h+1}(y) + 2z^{[g,\lambda]}_{h+1} + \frac15\cdot \rank^{[g,\lambda]}_h(y) + \sum_{i = g}^\lambda q^i_h\cdot B_i.\]
After subtracting $\rank^{[g,\lambda]}_h(y) / 5$ from both sides of this inequality, and then multiplying both sides of the inequality by $5/4$, we get
\begin{align*}
\rank^{[g,\lambda]}_h(y)
&\le \frac52\cdot \rank^{[g,\lambda]}_{h+1}(y) + \frac52\cdot z^{[g,\lambda]}_{h+1} + \frac{5}{4}\cdot \sum_{i = g}^\lambda q^i_h\cdot B_i
\\
&\le \frac52\cdot \Bigg(\sum_{i = g}^\lambda \bigg(\sum_{h' \ge h+2} \left(2\cdot 3^{h' - h - 1}\cdot (q^i_{h'}\cdot B_i + z^i_{h'})\right)
		+ 2\cdot q^i_{h+1}\cdot B_i\bigg)\Bigg) 
	 + \frac52\cdot z^{[g,\lambda]}_{h+1} + \frac{5}{4}\cdot \sum_{i = g}^\lambda q^i_h\cdot B_i
\\
&\le \sum_{i = g}^\lambda \left( \sum_{h' \ge h+1} \left( 2\cdot 3^{h' - h}\cdot (q^i_{h'}\cdot B_i + z^i_{h'})\right)
	+  2\cdot q^i_h\cdot B_i	\right)\,,
\end{align*}
where the second inequality uses the induction hypothesis~\eqref{eqn:rank bound for full mergeability - IndHyp}.
Thus, \eqref{eqn:rank bound for full mergeability} holds.

Using $z^i_h\ge 0$, we simplify~\eqref{eqn:rank bound for full mergeability} and get
\begin{equation}\label{eqn:rank bound for full mergeability - simpler}
\rank^{[g,\lambda]}_h(y) \le \sum_{i = g}^\lambda \sum_{h' \ge h} \left( 2\cdot 3^{h' - h}\cdot (q^i_{h'}\cdot B_i + z^i_{h'})\right)\,.
\end{equation}

Finally, we bound the variance $\V[\err_h(y)]$, which is at most $\sum_{i = 0}^\lambda m^i_h$ as $m^i_h = 0$ for $i > \lambda$ and $h\ge H_{\lambda+1}$,
by Lemma~\ref{lem:no error at Hi}.
Recall from Section~\ref{s:mergeability-singleLevel}
that $m^i_h$ is the number of important compaction operations at level $h$ represented by $i$-nodes.
We prove by a ``backward'' induction on $g = \lambda, \lambda - 1, \dots, 0$
that the following inequality holds for any $h\ge H_{\lambda+1}$:
\begin{equation}\label{eqn:var bound for full mergeability}
\sum_{i = g}^\lambda m^i_h \le \sum_{i = g}^\lambda \frac{1}{k_i}\cdot\sum_{h' \ge h} 8\cdot 3^{h' - h}\cdot (q^i_{h'}\cdot B_i + z^i_{h'})\,.
\end{equation}
Note that~\eqref{eqn:var bound for full mergeability} for $g = 0$ gives~\eqref{eqn:var bound for full mergeability - lemma}
and that for $h = H$, there is no (important) compaction, thus we have that $h < H$.
Consider $0\le g\le \lambda$ and suppose that for any $g' > g$ (in the case $g < \lambda$), we have that
\begin{equation}\label{eqn:var bound for full mergeability - IndHyp}
\sum_{i = g'}^\lambda m^i_h \le \sum_{i = g'}^\lambda \frac{1}{k_i}\cdot\sum_{h' \ge h} 8\cdot 3^{h' - h}\cdot (q^i_{h'}\cdot B_i + z^i_{h'})\,.
\end{equation}
To show~\eqref{eqn:var bound for full mergeability},
we use Lemma~\ref{lem:single layer mergeability - stronger} with $a = g$
to get $\sum_{i = g}^{\lambda} m^i_h\cdot k_i \le 4\rank^{[g,\lambda]}_h(y)$.
Dividing this inequality by $k_g$ and using~\eqref{eqn:rank bound for full mergeability - simpler} gives
$$\sum_{i = g}^\lambda \frac{k_i}{k_g}\cdot m^i_h
	\le \sum_{i = g}^\lambda \frac{1}{k_g}\cdot\sum_{h' \ge h} 8\cdot 3^{h' - h}\cdot (q^i_{h'}\cdot B_i + z^i_{h'})\,.$$
If $g = \lambda$, this proves the base case of the induction.
Otherwise, for every $g' > g$, we add inequality~\eqref{eqn:var bound for full mergeability - IndHyp} 
(that holds by the induction hypothesis) multiplied by $(k_{g' - 1} - k_{g'}) / k_g$
(which is non-negative as $k_{g' - 1} \ge k_{g'}$) to obtain
\begin{align}
	\nonumber
\sum_{i = g}^\lambda &\left(\frac{k_i}{k_g} + \sum_{g' = g + 1}^i \frac{k_{g' - 1} - k_{g'}}{k_g}\right) \cdot m^i_h 
\\
&\le \sum_{i = g}^\lambda \left(\frac{k_i}{k_g\cdot k_i} + \sum_{g' = g + 1}^i \frac{k_{g' - 1} - k_{g'}}{k_g\cdot k_i} \right)
	\cdot \sum_{h' \ge h} 8\cdot 3^{h' - h}\cdot (q^i_{h'}\cdot B_i + z^i_{h'})\,.
\label{eqn:var bound for full mergeability - 2}
\end{align}
Note that the sum of fractions of $k_i$'s on the RHS of~\eqref{eqn:var bound for full mergeability - 2}
equals $1/k_i$ for any $i$, since the numerators in $\sum_{g' = g + 1}^i (k_{g' - 1} - k_{g'}) / (k_g\cdot k_i)$ form a telescoping sum, which equals $k_g - k_i$.
Similarly, the sum of fractions of $k_i$'s on the LHS of~\eqref{eqn:var bound for full mergeability - 2} equals 1 for any $i$,
so the LHS equals $\sum_{i = g}^\lambda m^i_h$.
This shows~\eqref{eqn:var bound for full mergeability}.
\end{proof}

Finally, we have all ingredients needed to show that we can match the streaming result of Theorem~\ref{thm:sketch}
even when creating the sketch using an arbitrary sequence of merge operations without any advance knowledge about
the total size of the input. That is, we now prove the full mergeability claim of Theorem~\ref{thm:mergeability-full-intro},
which we restate for convenience.

\thmintro*

\begin{proof}
We condition on the bounds from Lemmas~\ref{lem:decay rank mergeability} and~\ref{lem:rank initial bound},
which together hold with probability at least $1 - 2\delta$.
Using Lemma~\ref{lem:var bound for full mergeability},
we first bound the error on levels $h\ge H_{\lambda+1}$,
for which we have that $m^i_h = 0$ for $i > \lambda$ by Lemma~\ref{lem:no error at Hi}:
\begin{align}
\sum_{h\ge H_{\lambda+1}} 2^{2h} \cdot \V[\err_h(y)]
&\le \sum_{h\ge H_{\lambda+1}} 2^{2h} \cdot \sum_{i = 0}^\lambda \sum_{h' \ge h} \frac{8\cdot 3^{h' - h}\cdot (q^i_{h'}\cdot B_i + z^i_{h'})}{k_i}
\nonumber\\
&= \sum_{i = 0}^\lambda \sum_{h'\ge H_{\lambda+1}} \sum_{h = H_{\lambda+1}}^{h'} 2^{2h} \frac{8\cdot 3^{h' - h}\cdot (q^i_{h'}\cdot B_i + z^i_{h'})}{k_i}
\nonumber\\
&\le \sum_{i = 0}^\lambda \sum_{h'\ge H_{\lambda+1}} \frac{2^{2h' + 5}\cdot (q^i_{h'}\cdot B_i + z^i_{h'})}{k_i}
\label{eqn:var bound for full mergeability - ineq 2}\\
&\le \sum_{i = 0}^\lambda \sum_{h' = H_{\lambda+1}}^{H_i} \frac{2^{h' + 9}\cdot (q^i_{h'}\cdot B_i + z^i_{h'})\cdot \rank(y)}{k_i\cdot B_i}
\label{eqn:var bound for full mergeability - ineq 3}\\
&\le \frac{\eps^2\cdot \rank(y)}{2^5\ln(1/\delta)}\cdot \sum_{i = 0}^\lambda \sum_{h' = H_{\lambda+1}}^{H_i} 2^{h'}\cdot (q^i_{h'}\cdot B_i + z^i_{h'})\,, \label{eqn:var bound for full mergeability - ineq 4}
\end{align}
where inequality~\eqref{eqn:var bound for full mergeability - ineq 2} follows from
\begin{align*}
\sum_{h = H_{\lambda+1}}^{h'} 2^{2h} \cdot 8\cdot 3^{h' - h}
= 8\cdot 3^{h'}\cdot \sum_{h = H_{\lambda+1}}^{h'} \left(\frac43\right)^h
\le 8\cdot 3^{h'}\cdot 3\cdot \left(\frac43\right)^{h'+1}
= 8\cdot 4^{h'+1} = 2^{2h' + 5}\,,
\end{align*}
inequality~\eqref{eqn:var bound for full mergeability - ineq 3} uses that $q^i_{h'} = 0$ and $z^i_{h'} = 0$ for $h' > H_i$ by Lemma~\ref{lem:no error at Hi}
and that $2^{H_i} \leq 2^4\cdot \rank(y) / B_i$ by~\eqref{eqn:H_i-bound mergeability},
and inequality~\eqref{eqn:var bound for full mergeability - ineq 4} follows from the bound on $k_i\cdot B_i$ in~\eqref{eqn:kB lower bound mergeability - strong}.

By Lemma~\ref{lem:rank initial bound}, $\sum_{i = 0}^\lambda \sum_{h' = H_{\lambda+1}}^{H_i} 2^{h'}\cdot (q^i_{h'}\cdot B_i + z^i_{h'}) \le 4\rank(y)$,
which implies our final variance bound for levels $h \ge H_{\lambda+1}$:
$$\sum_{h\ge H_{\lambda+1}} 2^{2h} \cdot \V[\err_h(y)]
\le \frac{\eps^2\cdot \rank(y)^2}{8\ln(1/\delta)}\,.$$
Let $\err_{\ge H_{\lambda+1}}(y)$ be the error in the estimate of $y$ from compactions at levels $h\ge H_{\lambda+1}$.
Plugging the variance bound into the tail bound for sub-Gaussian variables (Fact~\ref{fact:tailBound-subGaussian}) we conclude that
\begin{align*}
\Pr\left[ |\err_{\ge H_{\lambda+1}}(y)| > \eps \rank(y)/2 \right] 
< 2\exp\left( -\frac{\eps^2\cdot \rank(y)^2}{8\cdot \eps^2\cdot \rank(y)^2/(8\ln(1/\delta))} \right) 
= 2\exp\left( -\ln \frac{1}{\delta} \right)
= 2\delta\,.
\end{align*}

Next, we bound the error from compactions on levels below $H_{\lambda+1}$, denoted $\err_{< H_{\lambda+1}}(y)$.
The variance of this error is
\begin{align}
\sum_{h= 0}^{H_{\lambda+1}-1} 2^{2h} \cdot \V[\err_h(y)]
&\le \sum_{h= 0}^{H_{\lambda+1}-1} 2^{2h} \cdot \rank_h(y)
\nonumber\\
&\le \sum_{h= 0}^{H_{\lambda+1}-1} 2^{2h} \cdot 2^{-h+2}\cdot \rank(y)
\label{eqn:var bound for full mergeability - sampling - ineq 2}\\
&\le 2^{H_{\lambda+1}+2} \cdot \rank(y)
\le 2^6\cdot \frac{\rank(y)^2}{B_{\lambda+1}}
\le \frac{\eps^2\cdot \rank(y)^2}{8\ln(1/\delta)}
\nonumber
\end{align}
where~\eqref{eqn:var bound for full mergeability - sampling - ineq 2} 
uses Lemma~\ref{lem:decay rank mergeability}
and the last two steps use~\eqref{eqn:H_i-bound mergeability} and~\eqref{eqn:B_lambda_bnd_HatKsquared}, respectively
(note that $B_{\lambda+1} \ge B_\lambda$).
Using Fact~\ref{fact:tailBound-subGaussian} as above
we get that $\Pr\left[ |\err_{< H_{\lambda+1}}(y)| > \eps \rank(y)/2 \right] < 2\delta$.
Rescaling $\delta$, this completes the calculation of the failure probability.

Lastly, we bound the size of the final sketch $S$. Let $H$ be the index of the highest level in $S$.
Observe that $H\le \lceil \log_2(n / B_0)\rceil$. Indeed, since each item at level $h =\lceil \log_2(n / B_0)\rceil$ has weight $2^h$,
there are fewer than $B_0$ items inserted to level $h$ and consequently, level $h$ is never compacted
(here, we also use that $B_0\le B_1\le \cdots \le B_\ell$). Hence, there are $O(\log(\eps n))$ levels in $S$ as $B_0\ge 1/\eps$.
Each level has capacity $B_\ell = 2\cdot k_\ell\cdot \lceil \log_2(\overline{N_\ell} / k_\ell) \rceil$, where $\overline{N_\ell} = \min\{N_\ell, N_\lambda\}$, so the total
memory requirement of $S$ is
\begin{align*}
 O\left(\log(\eps n)\cdot k_\ell\cdot \log\left(\frac{\overline{N_\ell}}{k_\ell}\right)\right)
 &= O\left(\log(\eps n)\cdot \frac{\hat{k}}{\sqrt{\log(\overline{N_\ell} / \hat{k})}}\cdot \log\left(\frac{\overline{N_\ell}}{k_\ell}\right)\right)
 \\
 &= O\left(\log(\eps n)\cdot \hat{k}\cdot \sqrt{\log\left(\frac{\overline{N_\ell}}{\hat{k}}\right)}\right)
 \\
 &= 
 O\left(\eps^{-1}\cdot \log^{1.5}(\eps n)\cdot \sqrt{\log\frac1\delta}\right)
 \,,
\end{align*}
where we use that $\log_2 (\overline{N_\ell} / k_\ell) = O(\log (\overline{N_\ell} / \hat{k})) = O(\log(\epsilon \overline{N_\ell})) = O(\log(\epsilon n))$ (as $k_\ell\ge \hat{k} / \sqrt{\log_2(\overline{N_\ell} / \hat{k})}$, 
$\hat{k} \ge 1/\eps$, and $\overline{N_\ell} \le n^2$).
\end{proof}

\section{Analysis with Extremely Small Failure Probability}
\label{s:extremeDelta}

In this section, we provide a somewhat different analysis of our algorithm, which
yields an improved space bound for extremely small values of
$\delta$, at the cost of a worse dependency on $n$. In particular, we show a space upper bound of
$O(\eps^{-1}\cdot \log^2(\eps n)\cdot \log\log(1/\delta))$ for any $\delta > 0$.
For simplicity, we only give the subsequent analysis in the streaming setting, although we conjecture
that an appropriately adjusted analysis in Section~\ref{s:mergeability} would yield the same bound under arbitrary merge operations.
We further assume foreknowledge of (a polynomial bound on) $n$, the
stream length; this assumption can be removed in a similar fashion to Section~\ref{s:unknownlength}.
As a byproduct, we show at the end of this section that this result implies a deterministic space upper bound of $O(\eps^{-1}\cdot \log^3(\eps n))$
for answering rank queries with multiplicative error $\eps$, thus matching the state-of-the-art result
of Zhang and Wang~\cite{zhangwang}. 

To this end, we use Algorithm~\ref{code:full} with a different setting of $k$, namely,
\begin{equation}\label{eqn:extremeDelta-setk}
k = 2^4\cdot \left\lceil\frac{1}{\eps}\cdot \log_2\ln\frac{1}{\delta}\right\rceil.
\end{equation}
We remark that, unlike in Section~\ref{s:fullanalysis}, the value of $k$ does not depend on $n$ directly
(only possibly indirectly if $\delta$ or $\eps$ is set based on $n$).
Note that the analysis of a single relative-compactor in Section~\ref{s:relativecompactor} still applies and
in particular, there are at most $\rank_h(y) / k$ important steps at each level $h$ by Lemma~\ref{lem:single layer}.

We enhance the analysis for a fixed item $y$ of Section \ref{s:fullanalysis}.
The crucial trick to improve the dependency on $\delta$ from $\sqrt{\ln(1/\delta)}$ to $\log_2\ln(1/\delta)$
is to analyze the sketch using Chernoff bounds only below a certain level $H'(y)$ and provide
deterministic bounds for levels $H'(y) \le h < H(y)$, where $H(y)$ is defined as in Section~\ref{s:fullanalysis} as the  minimal $h$ for which $2^{2-h} \rank(y) \leq B/2$.
The idea to analyze a few top levels deterministically was first used by Karnin et al.~\cite{karnin2016optimal} and we apply it also in Section~\ref{s:k0sec}.
We define $$H'(y) = \max\left(0, H(y) - \lceil\log_2\ln(1/\delta)\rceil\right)\,.$$
Next, we provide modified rank bounds.

\begin{lemma}\label{lem:rank det bound}
Assuming $H(y) > 0$,
for any $h < H(y)$ it holds that $\rank_h(y) \leq 2^{-h+2}\rank(y)$ with probability at least $1 - \delta$.
\end{lemma}

\begin{proof}
We first show by induction on $0\le h < H'(y)$ that $\rank_h(y) \leq 2^{-h+1}\rank(y)$ with probability at least
$1 - \delta\cdot 2^{h - H'(y)}$, conditioned on $\rank_{\ell}(y) \leq 2^{-\ell+1}\rank(y)$ for any $\ell < h$.
This part of the proof is similar to that of Lemma~\ref{lem:decay rank}.
The base case holds by $\rank_0(y) = \rank(y)$.

Consider $0 < h < H'(y)$. As in Lemma~\ref{lem:decay rank}, \[\Pr[\rank_h(y) > 2^{-h+1}\rank(y)] \le \Pr[Z_h > 2^{-h}\rank(y)],\]
where $Z_h$ is a zero-mean sub-Gaussian variable with variance at most $\V[Z_h] \le 2^{-h + 1}\cdot \rank(y) / k$.
We apply the  tail bound for sub-Gaussian variables (Fact~\ref{fact:tailBound-subGaussian}) on $Z_h$ to get
\begin{align*}
\Pr[Z_h > 2^{-h}\rank(y)]
&< \exp\left(-\frac{2^{-2h}\cdot \rank(y)^2}{2\cdot (2^{-h + 1}\cdot \rank(y) / k)} \right)
\\
&= \exp\left(-2^{-h - 2}\cdot \rank(y)\cdot k \right)
\\
&= \exp\left(-2^{-h + H'(y) - 6}\cdot 2^{H(y) - H'(y)}\cdot 2^{4 - H(y)} \rank(y)\cdot k \right)
\\
&\le \exp\left(-2^{-h + H'(y) - 6}\cdot 2^{H(y) - H'(y)}\cdot B\cdot k \right)
\\
&\le \exp\left(-2^{-h + H'(y)}\cdot \ln\frac{1}{\delta} \right)
= \delta^{2^{- H'(y) + h}}
\le \delta\cdot 2^{ - H'(y) + h}\,,
\end{align*}
where the second inequality uses $2^{4 - H(y)} \rank(y) \ge B$ (by the definition of $H(y)$),
the third inequality follows from $2^{H(y) - H'(y)} \ge \ln\frac{1}{\delta}$ and $B\cdot k\ge k^2 \ge 2^6$, and the last inequality uses $\delta \le 0.5$.
This concludes the proof by induction.
Taking the union bound over levels $h < H'(y)$, 
it holds that $\rank_h(y) \leq 2^{-h+1}\rank(y)$ for any $h < H'(y)$ with probability at least $1 - \delta$.

Finally, consider level $h\ge H'(y)$ and condition on $\rank_{H'(y) - 1}(y) \leq 2^{-H'(y)+2}\rank(y)$.
(In the case $H'(y) = 0$, we have $\rank_0(y) = \rank(y)$.)
Note that for any $\ell > 0$, it holds that $\rank_\ell(y) \leq \frac12 \cdot (1 + 1/k)\cdot \rank_{\ell-1}(y)$.
Indeed, $\rank_\ell(y) \le \frac12\cdot \left(\rank_{\ell - 1}(y) + \mathrm{Binomial}(m_{\ell-1})\right)$
(see~Equation~\ref{eqn:rank_process}) and $\mathrm{Binomial}(m_{\ell-1}) \le m_{\ell-1} \le \rank_{\ell - 1}(y) / k$
by Lemma~\ref{lem:single layer}.
That is, regardless of the outcome of the random choices, we always
obtain this weaker bound on the rank of an item.

By using this deterministic bound for levels $H'(y)\le \ell\le h$, we get
\begin{align*}
\rank_h(y)
	&\le 2^{-h + H'(y) - 1}\cdot \left(1 + \frac{1}{k}\right)^{h - H'(y) + 1} \cdot \rank_{H'(y) - 1}(y)
	\\
	&\le 2^{-h + H'(y) - 1}\cdot \left(1 + \frac{1}{k}\right)^{0.5\cdot k}\cdot 2^{-H'(y)+2}\cdot \rank(y)
	\le 2^{-h+2}\cdot \rank(y)\,,
\end{align*}
where in the second inequality, we use $h - H'(y) + 1 \le 0.5\cdot k$ (which follows from $h < H(y)$ and $H(y) - H'(y) \le 2\cdot \log_2\ln\frac{1}{\delta} \le 0.5\cdot k$)
together with the bound on $\rank_{H'(y) - 1}(y)$, 
and the last inequality uses the fact that $(1 + 1/k)^{0.5\cdot k} \le \sqrt{e} < 2$.
\end{proof}

We now state the main result of this section, which proves Theorem~\ref{thm:extremeDelta intro}
assuming an advance knowledge of (a polynomial upper bound on) the stream length $n$.
This assumption can be removed using the technique described in Section~\ref{s:unknownlength}.

\begin{theorem} \label{thm:extremeDelta}
Assume that (a polynomial upper bound on) the stream length $n$ is known in advance.
For any parameters $0 < \delta \le 0.5$ and $0 < \eps \le 1$,
let $k$ be set as in~\eqref{eqn:extremeDelta-setk}.
Then, for any fixed item~$y$, Algorithm~\ref{code:full} with parameters $k$ and $n$ computes an estimate
$\hat{\rank}(y)$ of $\rank(y)$ with error $\err(y) = \hat{\rank}(y) - \rank(y)$ such that 
$$ \Pr\left[ |\err(y)| > \eps \rank(y) \right] < 3\delta\,.$$
The overall memory used by the algorithm is $O\left(\eps^{-1}\cdot \log^2(\eps n)\cdot \log\log(1/\delta)\right)$.
\end{theorem}

\begin{proof}
We condition on the bounds in Lemma~\ref{lem:rank det bound}, which together hold with probability at least $1-\delta$.
We split $\err(y)$, the error of the rank estimate for $y$, into two parts:
$$\err'(y) = \sum_{h=0}^{H'(y) - 1} 2^h\cdot \err_h(y) \quad\text{and}\quad \err''(y) = \sum_{h=H'(y)}^{H} 2^h\cdot \err_h(y)\,.$$
Note that $\err(y) = \err'(y) + \err''(y)$; we bound both these parts by $\frac12\eps\rank(y)$ w.h.p., starting with $\err'(y)$.
If $H'(y) = 0$, then clearly $\err'(y) = 0$.
Otherwise, we analyze the variance of the zero-mean sub-Gaussian variable $\err'(y)$ as follows:
\begin{align*}
\V[\err'(y)] &= \sum_{h=0}^{H'(y) - 1} 2^{2h}\cdot \V[\err_h(y)]
\\
&\le \sum_{h=0}^{H'(y) - 1} 2^{2h}\cdot \frac{\rank_h(y)}{k}
\\
&\le \sum_{h=0}^{H'(y) - 1} 2^{2h}\cdot \frac{2^{-h+2}\rank(y)}{k}
\\
&\le 2^{H'(y)+2}\cdot \frac{\rank(y)}{k}
= 2^{H'(y)-H(y) + 2}\cdot 2^{H(y)}\cdot  \frac{\rank(y)}{k}
\le 2^{H'(y)-H(y) + 6}\cdot \frac{\rank(y)^2}{k\cdot B}
\end{align*}
where the first inequality is by Lemma~\ref{lem:single layer}, the second by Lemma~\ref{lem:rank det bound},
and the last inequality uses $2^{H(y)} \leq 2^4\cdot \rank(y) / B$, which follows from the definition of $H(y)$.

We again apply Fact~\ref{fact:tailBound-subGaussian} to obtain
\begin{align*}
\Pr\left[ |\err'(y)| > \frac{\eps \rank(y)}{2} \right] 
&< 2 \exp\left( -\frac{\eps^2\cdot \rank(y)^2}{4\cdot 2\cdot 2^{H'(y)-H(y) + 6}\cdot \rank(y)^2/(k\cdot B)} \right) 
\\
&= 2 \exp\left( -\eps^2\cdot k\cdot B\cdot 2^{-H'(y)+H(y) - 9} \right)
\\
&\le 2 \exp\left( -2^{-H'(y)+H(y)} \right)
= 2 \exp\left( -\ln \frac{1}{\delta} \right)
= 2\delta\,,
\end{align*}
where the second inequality uses $k\cdot B\ge 2\cdot k^2\ge \eps^{-2}\cdot 2^9$.

Finally, we use deterministic bounds to analyze $\err''(y)$.
Note that $$\rank_{H(y)}(y) \leq 2^{-H(y)+2}\rank(y) \le B/2,$$ where the first inequality holds
because we have conditioned on the bounds of Lemma \ref{lem:rank det bound} holding, and the second inequality holds by the definition of $H(y)$. It follows that
there is no important step at level $H(y)$, and hence no error introduced at any level $h\ge H(y)$, i.e.,
$\err_h(y) = 0$ for $h\ge H(y)$. 
We thus have
\begin{align*}
\err''(y) & = \sum_{h=H'(y)}^{H(y)-1} 2^h\cdot \err_h(y)\\
& \le \sum_{h=H'(y)}^{H(y)-1} 2^h\cdot \frac{\rank_h(y)}{k}
\le \sum_{h=H'(y)}^{H(y)-1} 2^h\cdot \frac{2^{-h+2}\rank(y)}{k}
\le \sum_{h=H'(y)}^{H(y)-1} \frac{\eps \rank(y)}{2\cdot \lceil\log_2\ln\frac{1}{\delta}\rceil}
\le \frac{\eps \rank(y)}{2}\,,
\end{align*}
where the first inequality is by Lemma~\ref{lem:single layer}, the second by Lemma~\ref{lem:rank det bound},
the third inequality follows from the definition of $k$ in~\eqref{eqn:extremeDelta-setk},
and the last step uses that the sum is over $H(y) - H'(y) \le \lceil\log_2\ln\frac{1}{\delta}\rceil$ levels.
This concludes the analysis of $\err(y)$ and the calculation of the failure probability.

Regarding the space bound, there are at most $H\le \lceil\log_2(n/B)\rceil + 1 \le \log_2(\eps n)$ relative-compactors by
Observation~\ref{obs:H bound}, and each requires
$B = 2\cdot k\cdot \lceil\log_2(n/k)\rceil = O\left(\eps^{-1}\cdot \log\log(1/\delta)\cdot \log(\eps n)\right)$ memory words.
\end{proof}

The proof of Theorem~\ref{thm:extremeDelta} implies a deterministic sketch of size $O(\eps^{-1}\cdot \log^3(\eps n))$,
which matches the state-of-the-art result by Zhang and
Wang~\cite{zhangwang}.
Indeed, when $\log_2\ln (1/\delta)\ge \log_2(\eps n) \ge H$ (i.e.,
$\delta < \exp(-\eps n)$), we have $H'(y) = 0$, and in this case, inspecting the proofs of Lemma \ref{lem:rank det bound} 
and Theorem \ref{thm:extremeDelta} yields that the entire analysis holds with probability 1. In more detail, 
when $H'(y)=0$, 
the bounds in  Lemma \ref{lem:rank det bound} hold with probability 1, and 
 the quantity $\err'(y)$ in the proof of Theorem~\ref{thm:extremeDelta} is deterministically 0,
 while the bound on $\err''(y)$ in the proof of Theorem~\ref{thm:extremeDelta} holds with probability 1 as well.
 This is sufficient to conclude that the error guarantee holds for \emph{any} choice of the algorithm's internal randomness.
 The resulting algorithm is reminiscent of deterministic algorithms for the uniform
quantiles problem~\cite{manku1998approximate}.

\section{Discussion and Open Problems}
\label{s:conclusion}
For constant failure probability $\delta$, we have shown an $O(\eps^{-1}\cdot \log^{1.5}(\eps n))$ space upper bound for relative error  quantile approximation over data streams.
Our algorithm is provably more space-efficient than any deterministic comparison-based algorithm~\cite{cormode2019tight}, and is within an $\widetilde{O}\left(\sqrt{\log(\eps n)}\right)$ factor of the known lower bound for
 randomized algorithms (even non-streaming algorithms, see Appendix~\ref{appendix}).
Moreover, the sketch output by our algorithm is fully mergeable, with the same accuracy-space trade-off as in the streaming setting,
rendering it suitable for a parallel or distributed environment.
  The main open question is to close the aforementioned $\widetilde{O}(\sqrt{\log(\eps n)})$-factor gap.  

\paragraph{Acknowledgments.}
The authors wish to thank anonymous reviewers for many helpful suggestions.
The research is performed in close collaboration with DataSketches \url{https://datasketches.apache.org/}, the Apache open source project for streaming data analytics.

\bibliographystyle{plain}
\bibliography{relative_error_quantiles}

\newpage
\appendix

\allowdisplaybreaks
\section{A Lower Bound for Non-Comparison Based Algorithms}
\label{appendix}
Cormode and Vesel{\'{y}} \cite[Theorem~6.5]{cormode2019tight} proved an $\Omega(\eps^{-1} \cdot \log^2(\eps n))$ lower 
bound on the number of items stored by any deterministic comparison-based streaming algorithm for the relative-error quantiles problem.
Below, we provide a lower bound which also applies to offline, non-comparison-based randomized algorithms, but at the (necessary) cost of losing a $\log(\eps n)$ 
factor in the resulting space bound.
This result appears not to have been explicitly stated in the literature, though it follows from an argument similar to
\cite[Theorem 2]{Cormode:2005:ECB:1053724.1054027}.
We provide details in this appendix for completeness.

\label{app:lower}
\begin{theorem}
\label{t:lb} For any randomized algorithm that processes a data stream of items from universe $\mathcal{U}$
of size $|\mathcal{U}|\ge \Omega(\eps^{-1} \cdot \log(\eps n))$
and outputs a sketch that solves the all-quantiles approximation problem for multiplicative error $\eps$ with probability at least $2/3$
requires the sketch to have size $\Omega\left(\eps^{-1} \cdot \log(\eps n) \cdot \log(\eps |\mathcal{U}|)\right)$ bits of space. 
\end{theorem}
\begin{proof}
We show that any multiplicative-error sketch for all-quantiles approximation  can be used to 
losslessly encode an arbitrary subset $S$ of the data universe $\mathcal{U}$ of size $|S| = \Theta\left(\eps^{-1} \log(\eps n)\right)$. 
	This requires $\log_2{|\mathcal{U}| \choose |S|} = \Theta\left(\log((|\mathcal{U}|/|S|)^{|S|})\right) = \Theta\left(|S| \log\left(\eps |\mathcal{U}|\right)\right)$ bits of space. 
	The theorem follows.
	
	Let $\ell = 1/(8\eps)$ and $k=\log_2(\eps n)$; for simplicity, we assume that both $\ell$ and $k$ are integers.
	Let $S$ be a subset of $\mathcal{U}$ of size $s:=\ell \cdot k$. We will construct a stream $\sigma$ of length less than $\ell \cdot 2^k \leq n$
	such that a sketch solving the  all-quantiles approximation problem for $\sigma$ enables reconstruction of~$S$. To this end, let $\{y_1, \dots, y_s\}$ denote
	the elements of $S$ in increasing order. Consider the stream $\sigma$ where items $y_1, \dots, y_{\ell}$ each appear once,
	items $y_{\ell+1}, \dots, y_{2\ell}$ appear twice, and in general items $y_{i \ell + 1}, \dots, y_{(i+1) \ell}$ appear $2^i$ times, for $i=0, \dots, k-1$.
	Let us refer to all universe items in the interval $[y_{i \ell + 1}, y_{(i+1) \ell}]$ as ``phase-$i$'' items.
	
	The construction of $\sigma$ means that the multiplicative error $\eps$ in the estimated rank of any phase-$i$ item is at most $2^{i+1}/8 < 2^{i-1}$. This means 
	that for any phase $i \geq 0$ and integer $j \in [1, \ell]$, one can identify item $y_{i \ell + j}$ by finding the smallest universe item whose estimated rank 
	is strictly greater than $(2^{i}-1) \cdot \ell + 2^i\cdot j -  2^{i-1}$. Here, $(2^{i}-1) \cdot \ell$ is the number of stream updates
	corresponding to items in phases $0, \dots, i-1$, while $2^{i-1}$ is an upper bound on the error of the estimated rank of any phase-$i$ item. 
	Hence, from any sketch solving the all-quantiles approximation problem for $\sigma$ one can
	obtain the subset $S$, which concludes the lower bound.
\end{proof}

Theorem \ref{t:lb} is tight up to constant factors as an optimal summary consisting of $O(\eps^{-1} \cdot \log(\eps n))$ items can be constructed offline. 
For $\ell = \eps^{-1}$,
this summary stores all items of rank $1, \dots, 2\ell$ appearing in the stream and assigns them weight one, stores every other item of rank between $2\ell +1$ and $4 \ell$ and assigns them weight 2, 
stores every fourth item of rank between $4\ell+1$ and $8\ell$ and assigns them weight 4, and so forth.  This yields a weighted coreset $S$ for the
relative-error quantiles approximation, consisting of $|S| = \Theta\left( \ell \cdot \log({\eps n})\right)$ many items. Such a set $S$ can be
represented with $\log_2{|\mathcal{U}| \choose |S|} = \Theta\left(\eps^{-1} \cdot \log(\eps n) \cdot \log(\eps |\mathcal{U}|)\right)$ many bits. 

\section{Proof of Corollary \ref{corollaryintro}}
\label{app:corollaryintro}
Here we prove Corollary \ref{corollaryintro}, restated for the reader's convenience. 

\corintro*

\begin{proof}
Let $\Sstar$ be the offline optimal summary of the stream with multiplicative error $\eps/3$, i.e.,
a subset of items in the stream such that for any item $x$, there is $y\in \Sstar$ with
$|\rank(y) - \rank(x)| \le (\eps/3)\cdot \rank(x)$. Here, $y$ is simply the closest item to $x$ in the total order that is an element of $\Sstar$.
Observe that $\Sstar$ has $O(\eps^{-1} \cdot \log(\eps n))$ items; see the remark below Theorem~\ref{t:lb} in Appendix~\ref{app:lower} for a construction of $\Sstar$.

Thus, if our sketch with parameter $\eps' = \eps/3$ is able to compute for any $y\in \Sstar$ a rank estimate $\hat{\rank}(y)$
such that $|\hat{\rank}(y)-\rank(y)| \le (\eps/3)\cdot \rank(y)$,
then we can approximate $\rank(x)$ by $\hat{\rank}(y)$ using $y\in \Sstar$ with
$|\rank(y) - \rank(x)| \le (\eps/3)\cdot \rank(x)$ and the multiplicative guarantee for $x$ follows
from
\begin{align*}
	|\hat{\rank}(y) - \rank(x)|
& \le |\hat{\rank}(y)-\rank(y)| + |\rank(y) - \rank(x)|\\
&\le \frac{\eps}{3}\cdot \rank(y) + \frac{\eps}{3}\cdot \rank(x) \\
&\le \left(\frac{\eps}{3}\cdot (1 + \frac{\eps}{3}) + \frac{\eps}{3}\right)\cdot \rank(x)\\
&\le \eps\cdot \rank(x)\,.
\end{align*}

It remains to ensure that our algorithm provides a good-enough rank estimate for any $y\in \Sstar$.
We apply Theorem \ref{thm:mergeability-full-intro} with error parameter $\eps' = \eps/3$ and with failure probability set to
$\delta' = \delta / |\Sstar| = \Theta\left(\delta \cdot \eps/\log(\eps n)\right)$.
By the union bound, with probability at least $1-\delta$, 
the resulting sketch satisfies the $(1 \pm \eps/3)$-multiplicative error guarantee for any item in $\Sstar$.
In this event, the previous paragraph implies that the $(1\pm \eps)$-multiplicative guarantee holds for all $x \in \mathcal{U}$.
The space bound follows from Theorem \ref{thm:mergeability-full-intro} with $\eps'$ and $\delta'$ as above.
\end{proof}

\end{document}